\documentclass[numbers, lm]{sciprep}
\usepackage{setspace}
\usepackage{amsthm}
\usepackage{mathtools}
\usepackage{tcolorbox}
\usepackage{svg}
\usepackage{placeins}

\usepackage{float}
\usepackage{algorithm}
\usepackage{algpseudocode}
\usepackage{booktabs}
\usepackage{marvosym} %

\newcommand{\toctable}[1]{\addcontentsline{toc}{suppfig}{Table~\thetable: #1}}

\theoremstyle{plain}
\newtheorem{proposition}{Proposition}

\theoremstyle{definition}

\newtcolorbox[auto counter]{mybox}[1][]{
    title=Box~\thetcbcounter,
    #1 
}
\newfloat{boxfigure}{htbp}{lbf}
\floatname{boxfigure}{Box Figure}

\title{Learning developmental dynamics\\from lineage context}

\renewcommand{\corrauthor}{\textsuperscript{\Letter}}
\author{%
  Gokul~Gowri\affil{1--5,}\eqcontrib\corrauthor,
  Nic~Fishman\affil{6,7,}\eqcontrib,
  William~N.~Colgan\affil{2--5,}\eqcontrib,
  Minming~Wang\affil{8,9},
  Yusuke~Kijima\affil{9},
  Sachin~A.~Kumar\affil{2--5, 10},
  Aleksandr~Zimin\affil{13},
  Luke~W.~Koblan\affil{2--5},
  Marinka~Zitnik\affil{7,12--14},
  Zachary~D.~Smith\affil{8,9},
  Stephen~Bates\affil{1},\\
  Philippe~Rigollet\affil{11,}\corrauthor,
  Jonathan~S.~Weissman\affil{2--5,}\corrauthor%
}
\affiliation{%
  \footnotesize
  \affilmark{1}Department of Electrical Engineering and Computer Science, Massachusetts Institute of Technology, Cambridge, MA, 02139, USA\\
  \affilmark{2}Whitehead Institute for Biomedical Research, Cambridge, MA, 02142, USA\\
  \affilmark{3}Howard Hughes Medical Institute, Chevy Chase, MD, 20815, USA\\
  \affilmark{4}David H. Koch Center for Integrative Cancer Research, Massachusetts Institute of Technology, Cambridge, MA, 02139, USA\\
  \affilmark{5}Department of Biology, Massachusetts Institute of Technology, Cambridge, MA, 02139, USA\\
  \affilmark{6}Department of Statistics, Harvard University, Cambridge, MA, USA\\
  \affilmark{7}Department of Biomedical Informatics, Harvard Medical School, Boston, MA, USA\\
  \affilmark{8}Department of Genetics, Yale Stem Cell Center, Yale School of Medicine, New Haven, CT, 06519, USA\\
  \affilmark{9}Max Planck Institute of Molecular Genetics, Max Planck Society, Berlin, 14195, Germany\\
  \affilmark{10}Department of Pediatric Oncology, Dana-Farber Cancer Institute and Division of Hematology/Oncology, Boston Children's Hospital, Harvard Medical School, Boston, MA, USA\\
  \affilmark{11}Department of Mathematics, Massachusetts Institute of Technology, Cambridge, MA, 02139, USA\\
  \affilmark{12}Kempner Institute for the Study of Natural and Artificial Intelligence at Harvard University, Allston, MA, USA\\
  \affilmark{13}Broad Institute of MIT and Harvard, Cambridge, MA, USA\\
  \affilmark{14}Harvard Data Science Initiative, Cambridge, MA, USA
}
\titlenote{\eqcontrib These authors contributed equally. \\ \indent \corrauthor Co-corresponding authors: \texttt{weissman@wi.mit.edu, rigollet@mit.edu, gokulg@mit.edu}}

\begin{document}

\maketitle

\begin{abstract}

A central goal of biology is to predict the dynamics of multicellular systems. Lineage context (the states of a cell's sisters, cousins, etc.) could inform these predictions, but how best to leverage it remains unclear. Here we introduce \textit{lineage inpainting}, a generative modeling approach that uses lineage context with tree-structured transformers and stochastic interpolants to model cell state dynamics across lineage trees. Applied to a lineage-traced transcriptional atlas of mouse embryogenesis, it recovers dynamics not resolved by gene expression alone, including multiple transcriptional convergence events where distinct lineages are drawn to a common state. Despite being fit on unperturbed development, lineage inpainting can generate hypotheses about perturbation responses by simulating counterfactuals. These simulations agree in many cases with observed developmental phenotypes, including a single-cell measurement of the Tbx5 knockout response. Our results show that lineage information provides a rich source of context for learning multicellular dynamics.

\end{abstract}

\section{Introduction}

Biologists have long sought to understand how a pluripotent cell gives rise to the many cell types in a fully developed organism \cite{Whitman1878-fb, Wilson1892-fp, Wagner2020-xr, Liberali2024-uy}. Technologies such as single-cell RNA-sequencing (scRNA-seq) have enabled atlases which enumerate the cell states that appear during development \cite{Wagner2018-cc, Briggs2018-po, Packer2019-xl, Qiu2022-xl, Qiu2024-cq}. While these atlases have mapped out cell state landscapes in remarkable detail, they do not provide explicit information about how states change over time, because transcriptomic measurements are typically destructive.

A large body of work has focused on the problem of inferring developmental dynamics from molecular snapshots alone \cite{Rashid2017-wa, Schiebinger2019-kc, Tong2020-ix, Qiu2022-ms, Huguet2022-sm, Tong2023-mf, Bunne2024-lq, Shen2024-ey}. A fundamental challenge of these efforts is that the problem is underconstrained: multiple possible dynamics can give rise to the same snapshots \cite{Weinreb2018-wy}. 
Many approaches, including optimal transport methods, resolve this ambiguity through a prescribed transport cost or dynamical regularization that favors short paths supporting a putative least-action principle driving cell dynamics. These assumptions are useful but can obscure biological transitions when molecular similarity is an unreliable proxy for shared ancestry or fate \cite{Packer2019-xl,Wagner2020-xr,Weinreb2020-uk}.

Lineage tracing provides complementary information about cell fate dynamics. Rather than inferring lineage relationships from molecular similarity, these techniques directly measure shared ancestry. In small, transparent embryos, such as \textit{C. elegans}  and the early embryo of \textit{D. rerio}, cell divisions can be tracked by microscopy \cite{Sulston1983-dj, olivier_cell_2010, kwak_live_2013}. Mammalian embryos are larger and less accessible, so lineage relationships are instead reconstructed from heritable marks that cells pass to their progeny \cite{lawson_clonal_1991, kretzschmar_lineage_2012, rodriguez-fraticelli_charting_2026}. These marks can be naturally occurring DNA mutations, static barcodes introduced at a given time point, or DNA edits that are continuously introduced over time by an engineered lineage tracing system \cite{frieda_synthetic_2017, mckenna_whole-organism_2016, alemany_whole-organism_2018, spanjaard_simultaneous_2018, raj_simultaneous_2018, kalhor_developmental_2018, bowling_engineered_2020}. In particular, recent methods including DNA Typewriter and PEtracer use prime editing to continuously and irreversibly modify synthetic edit sites in genomes \cite{Koblan2025-ac, Choi2022-yv, Yu2026-bh}. With PEtracer, we recently reconstructed lineage trees that resolve $\sim 75\%$ of cell divisions across early embryogenesis \cite{colgan}. However, again due to the destructive nature of transcriptomics, the internal node states (and the state dynamics) cannot be observed directly -- regardless of whether trees are measured by imaging or molecular recording.

The first step to shed light on the full progression of states is to make measurements across developmental times. If we observe both the lineage tree and transcriptional states at two or more timepoints, we can try to integrate them by matching the leaves of an earlier tree to the internal nodes of a later one. In \textit{C. elegans}, every individual organism follows the same sequence of cell divisions, producing the same tree \cite{Sulston1983-dj, Packer2019-xl}. Trees from different organisms can be directly overlaid, and the exact correspondence between cells at different timepoints is known.

In mammals, however, development is non-deterministic, so lineage relationships do not specify a fixed correspondence between cells in different embryos \cite{Tarkowski1967-bn, gardner_complete_1998}. Lineage nevertheless provides a \textit{context} in which to interpret each cell's measured state. This context can be informative both because development is regulative, with cell fates shaped by interactions with other cells \cite{Tarkowski1959-ve}, and because developmental history can bias fate in ways not captured by cell state measurements \cite{Weinreb2020-uk, Jindal2024-ik}. Considering a cell together with its relatives could therefore help distinguish dynamics that are ambiguous without lineage context. This motivates a shift from pairing cells to instead modeling distributions over states on a tree, considering the cell's state and its recorded lineage context simultaneously.

Recent methods have begun to use lineage information for dynamic inference. LineageOT \cite{Forrow2021-vs} and moslin \cite{Lange2024-ts} refine optimal transport couplings between timepoints using lineage information, returning a coupling over observed cells rather than a generative model. Consequently, one cannot directly use these methods to sample trajectories, forecast states, or predict the effects of perturbations. GENOT and related methods \cite{Klein2024-mk, Ventre2026-bg} provide approaches to lift such couplings into generative models, but operate at the single-cell level, predicting a cell's future or past from its current state alone.

Here we show that lineage context enables more faithful modeling of developmental trajectories beyond single-cell transcriptional states. We introduce \textit{lineage inpainting}, a generative approach that jointly samples unobserved cell states across a reconstructed lineage tree using the states of related cells observed at a single time point. The transformer architecture provides a natural foundation for this task: their attention mechanisms can integrate information across variable numbers of cells \cite{Vandergrift2026-aa, Fishman2025-gq}. We develop a tree-structured attention mechanism to incorporate lineage relationships and use it to build a transformer-based stochastic interpolant model \cite{Ho2020,Lipman2022-qc,Albergo2023}. We fit the model using a consistency objective: generated states must match observations from other trees collected at the corresponding developmental times \textit{in distribution}. This objective allows lineage inpainting to integrate measurements across embryos without requiring exact cell-to-cell correspondences or a predefined transport cost between cell states. 

We apply lineage inpainting to PEtracer data spanning mouse embryogenesis from E7.5 to E10.0 \cite{colgan}. The model yields a data-driven map of state transitions and recovers transcriptional convergence events missed by alternatives using gene expression alone, revealing how distinct progenitors can reach a common transcriptional state. We further show that a model trained on unperturbed development can be used to make  predictions about the developmental effects of genetic perturbations at single-cell resolution. Together, these results demonstrate that lineage context can improve the reconstruction and prediction of developmental trajectories.

\section{Results}

\begin{figure}[p]
    \centering
    \includegraphics[width=\linewidth]{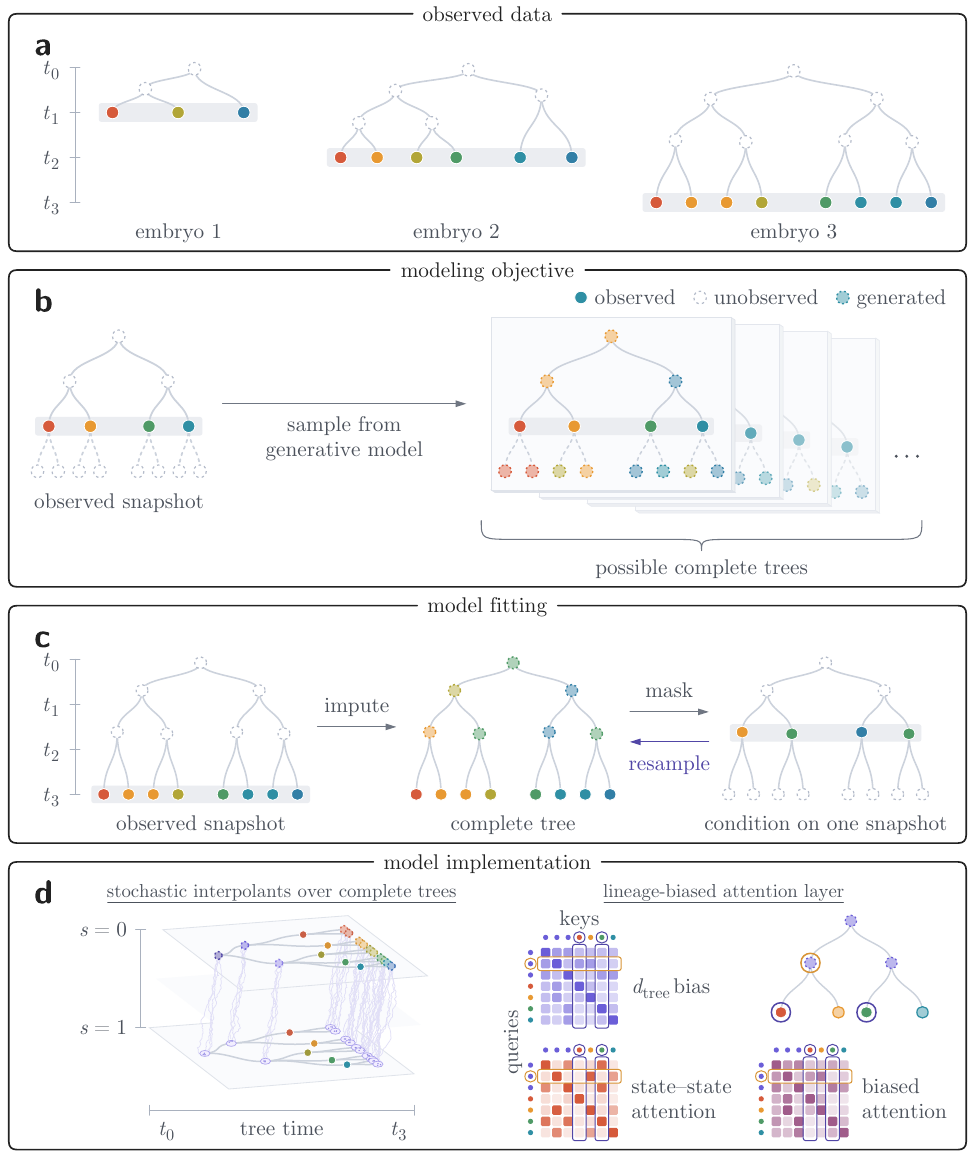}
    \caption{\textbf{Reconstructing developmental dynamics using lineage inpainting.} \captionmore}
    \label{fig:schematic}
\end{figure}

\contcaption{
    \textbf{(a)} Form of preprocessed lineage-traced scRNA-seq snapshots. States at the leaves of the tree are measured, and the tree topology and internal node timings are measured by tree reconstruction methods.
    \textbf{(b)} Our modeling objective: given a population of cells and an associated lineage tree with unobserved cells, sample from the distribution of possible state-filled trees. Nodes represent cells, and bold colors represent observed states. Colorless nodes are cells with unobserved states, pale colored nodes represent predicted node states.
    \textbf{(c)} A general approach to fitting a model of developmental dynamics to lineage-traced snapshots. Since direct measurements only ever provide states at the leaves of a tree, we first fill in internal nodes using an imputation model. We then mask all but one layer of the tree, and fit a generative model to match the masked node states. We refer to the two steps of the process as imputation and sampling.
    \textbf{(d)} (left) Lineage inpainting's sampling model is a stochastic interpolant model over trees, where states across an entire tree are simultaneously denoised using a tree-aware transformer architecture. (right) Lineage inpainting is built on a modified attention layer which incorporates a tree distance ($d_{\rm{tree}}$) bias.
}

The development of an embryo can be represented as a lineage tree branching through time as cells divide \cite{Wilson1892-fp}. Retrospective lineage tracing methods such as PEtracer \cite{Koblan2025-ac, colgan} allow us to reconstruct these trees. For each sequenced cell, these methods provide both its transcriptional state and its position among the leaves of a reconstructed lineage tree. The internal nodes of that tree represent ancestral cells, each with an inferred time from molecular clock-based branch length estimation \cite{Jones2020-jg, Prillo2025-ik}. Because sequencing is destructive each tree is observed only at a single time. We see the ``snapshot'' of cells alive at the time of capture (the leaves of the reconstructed tree) while every earlier ancestor and later descendant is unobserved (\textbf{Fig.~\ref{fig:schematic}a}). Sampling multiple embryos across development yields many trees, observed across different times.

Our goal is to turn these partial observations into a model of developmental dynamics which can condition on any observed snapshot and sample from the distribution of possible ancestors and descendants of each cell in the snapshot (\textbf{Fig.~\ref{fig:schematic}b}). We split this into two subproblems (\textbf{Fig.~\ref{fig:schematic}c}). First, we need to somehow fill in the unobserved states at internal nodes of measured trees, providing us with samples of complete trees. Second, we need to fit a generative model capable of sampling from a distribution of states forward or backward in time, conditional on a single snapshot. We refer to these subproblems as imputation and sampling respectively.

Below, we show how to address these problems. We first introduce a method for filling in internal node states which learns an imputation function rather than imposing any explicit transport cost, a second generative model for sampling states which operates over entire trees rather than individual cells and a tree-aware attention mechanism which underlies both models (\textbf{Fig.~\ref{fig:schematic}d}).

\subsection{Lineage inpainting}

\subsubsection{Learning to impute ancestral states}

Imputing states at the internal nodes of a tree is a classic problem in phylogenetics. In  ancestral state reconstruction, one has access to trait measurements (e.g., tail length) at the leaves of a phylogenetic tree (i.e., extant species), and the goal is to reconstruct the traits of ancestral species represented by internal nodes of the phylogenetic tree \cite{Felsenstein1985-ih, Pagel1999-go, Harmon2018-kk, Xu2026-uy}. This problem has one key distinction from our lineage tracing setting, which is that in phylogenetics only a single evolutionary tree is available, whereas in our developmental system we are able to observe multiple trees which have different observation times. 

A common approach to ancestral state reconstruction is to assume that evolutionary dynamics are Brownian. Under this assumption, Felsenstein's independent contrasts summarizes descendant states by averages weighted by their distance in the tree and their uncertainty \cite{Felsenstein1985-ih, Teo2025-graphicalmodels}.
The assumption of Brownian dynamics is sensible and perhaps necessary when one only has access to a single snapshot of a single tree. But with access to observations across multiple trees at different times, we can hope to learn a more specific ancestral reconstruction operation which is consistent with our observations.

The central challenge of fitting such an ancestral reconstruction model is that no internal node is ever measured, so there are no ground-truth ancestral states to regress against. Instead, we optimize our model using a loss function which encourages imputed node states to be consistent across our trees. Imputing the ancestors of a tree captured late in development yields a prediction for the populations at earlier timepoints, and we can compare these against true measurements of other embryos at earlier timepoints. A good ancestral reconstruction will match the distribution of the states at each time from the trees we observed at each of those times. Because each tree anchors this comparison to real data at a different time, requiring all trees to agree constrains the inferred dynamics to be consistent with our snapshot observations.

We implement this idea as an alignment loss that minimizes the distributional distance, across trees, between what we call \emph{tree features}. The simplest and most essential tree feature is the state of a single cell and its corresponding time. Matching the distribution of these across trees makes the cell population at each time point consistent from one tree to the next. This alone is not sufficient, because a model can reproduce the correct distribution of states at every time while still completely distorting the tree structure. So we also align higher-order tree features, including pairs of sister-cell states and trajectories of states across time (see Section~\ref{sec:coupling} for full details). Our alignment loss is the sum of the distributional distance of several tree features, between each pair of trees (\textbf{Supplemental Figs.~\ref{fig:supp-ising-running-example}--~\ref{fig:supp-ising-phase}}).

We compute these imputations with a transformer that takes the nodes of a tree as input, representing each observed node by its measured state and time, and each unobserved node starts as noise, with the model learning to turn this noise into consistent imputations. Tree structure is encoded in the model through a modified attention layer, where we bias each attention score by the tree distance between the two nodes it relates (\textbf{Box~\ref{box:attention}}). 

This imputation model provides predictions backward in time, filling in internal node states given the leaf states as input. However, it cannot sample from the distribution of descendants of a cell. To make predictions both forward and backward in time, we fit a separate sampling model.

\begin{mybox}[colback=white,colframe=blue!30!black,
title=\textbf{Box 1:} \sf A tree-structured transformer,label=box:attention, before upper={\parindent15pt}]
\begin{spacing}{1}

\small{\sf

A cell's fate and history are not fully determined by its current transcriptional state alone. Cells sharing a recent ancestor inherit common factors that indicate their origin and shape how they develop \cite{Weinreb2020-uk, Jindal2024-ik}. A single-cell level predictor is not able to leverage this shared information. Instead, we would ideally like a model that can process an entire tree --- each cell alongside all of its relatives --- at once. Here, we show how a simple extension of the transformer architecture allows us to do this.

The transformer is best known as the architecture at the heart of modern language modeling \cite{Vaswani2017, Turner2023-df}, but it is also a general purpose tool for processing sets of ``tokens''. In language, each ``token'' is a few characters in a sentence, and the transformer uses a mechanism called attention to allow each token to aggregate information from all others. Here, instead each token represents a cell state. Learned projections map token $i$ to a query $q_i$, key $k_i$, and value $v_i$, and attention updates it by averaging the values according to how well its query matches each key:

\begin{equation}
    h_i = \sum_j \alpha_{ij} v_j,
    \qquad
    \alpha_{ij} \propto \exp\left(\frac{q_i^\top k_j}{\sqrt{d_k}}\right),
\end{equation}

with $\alpha_{ij}$ normalized over $j$ and $d_k$ the query/key dimension. The update therefore combines information from the cells assigned the largest weights.

On its own, attention treats its inputs as an unordered set: tokens can influence one another only through their contents, never through their arrangement. Transformers conventionally supply the missing structure with \emph{positional encodings}---embeddings of each token's position added to its input---but these presume a left-to-right ordering like a token's position in a sentence, which cells across a branching tree do not have. An alternative approach is to encode relative position information in the attention mechanism itself \cite{Press2022}, by adding distance-dependent biases to the attention logits. This, in principle, generalizes to any setting where we have a pairwise distance between tokens that can be used to bias attention (e.g., spatial information).

Here, we use tree distance rather than the distance in the sequence (\textbf{Fig.~\ref{fig:schematic}d}). In particular, we use the \emph{patristic distance} $\operatorname{pat}(i,j)$ between two cells: the total branch length along the path joining them through their most recent common ancestor. We add the patristic distance to each pair's attention logit before normalization, scaled by a learnable weight $\gamma^{(l,h)}$, one per layer $l$ and head $h$:

\begin{equation}
  \alpha_{ij} \propto
  \exp\!\left(\frac{q_i^\top k_j}{\sqrt{d_k}} + \gamma^{(l,h)}\,\operatorname{pat}(i,j)\right)
\end{equation}

Patristic distance is how phylogenetics measures relatedness: under a Brownian model of state drift \cite{Felsenstein1985-ih}, the variance of the difference between two cells' states is proportional to their patristic distance, so biasing attention by it tells the model that recently diverged cells should resemble one another.

Early layers have large negative $\gamma^{(l,h)}$, concentrating attention on closer relatives, with positive weights only occurring in some heads of later layers (layers 3 and 4) favoring more distant relatives \textbf{(Supplemental Fig.~\ref{fig:supp_attn})}.

} %

\end{spacing}

\end{mybox}

\subsubsection{Learning to sample ancestor and descendant states}

Given trees with filled in internal states, our goal is to be able to condition on all the cells alive at any given time and sample the states for all the descendant and ancestral nodes. 
This problem resembles settings in computer vision called ``inpainting'' where we observe one part of an image and try to sample the rest, matching the natural distribution over images \cite{Lugmayr2022,Albergo2024-pq}. Building on these approaches, we introduce a tree-structured generative model that enables conditionally sampling ancestors and descendants based on the observed snapshot (Section~\ref{sec:diffusion}). 
Denoising diffusion models (and related flow matching and stochastic interpolant constructions) are widely used generative models that start from random noise and take denoising steps, with each step leading to sharper generated samples. We use a training recipe similar to the standard diffusion approach \cite{Ho2020}. We pick a snapshot of the filled in tree to condition on, which we treat as the observed states. We then corrupt the remaining nodes -- treated as unobserved -- with noise, and train the model to predict the unnoised clean states. Similar to the imputation model, we use a tree-based transformer architecture for the sampling model, so that the model can learn to integrate information over the entire lineage tree. When we want to sample the ancestors or descendants given an observed snapshot we can then initialize all the unobserved nodes as random noise and iteratively de-noise them, conditioning on the observed nodes throughout (\textbf{Fig.~\ref{fig:schematic}d}).

\subsubsection{Jointly learning to impute and sample}

Although in principle the imputation model and sampling model can be completely separate objects, we optimize the two models jointly (that is, end-to-end and simultaneously). At each training step, the imputation model first completes a partially observed tree, and that completed tree is then used to train the sampling model. 

Because the internal node states used to train the sampler are imperfect synthetic states produced by the imputation model, training directly on these outputs could lead the sampler to reproduce low-quality outputs from the imputer. To avoid this, we instead snap the imputer outputs to their nearest observed real cell at the corresponding or adjacent developmental snapshot, thereby providing sharp (realistic) training targets. We interleave these targets with training steps that retain the imputed states themselves, allowing the sampling loss to remain differentiable through the imputer. In this way, the sampler learns to generate realistic cell states while preserving the developmental correspondence learned by the imputation model (\textbf{Supplemental Fig.~\ref{fig:grid_stopgrad}}). Full training details are provided in Section~\ref{sec:training}.

This decomposition into imputation and sampling is more general than the particular models we use here. The imputation step can be any procedure that supplies a correspondence between states at different times, and the sampling step can be any conditional generative model trained on the resulting completed trees. GENOT, for example, fits into this framework: discrete optimal transport supplies the imputation, and conditional flow matching learns to sample from the resulting coupling. Among the combinations of imputation and sampling methods we evaluated, jointly optimizing tree-structured transformer models for imputation and sampling performs best (\textbf{Supplemental Text~\ref{supp:grid}}, \textbf{Supplemental Figs.~\ref{fig:grid_aligner}--\ref{fig:grid_overview}}).

\subsection{Evaluating lineage inpainting with synthetic developmental processes}

Cells do not always follow minimal displacement paths in gene expression \cite{Packer2019-xl}, and their trajectories are not always independent or Markovian with respect to transcriptomic state of a single cell \cite{Weinreb2020-uk, Tong2026-og}. We test the extent to which lineage inpainting is able to (1) learn non-least-action dynamics, and (2) learn context-dependent dynamics. Neither capability can be directly tested on real data, since no state of an internal node is ever measured in real lineage-traced scRNA-seq data. As such, we turn to synthetic and semi-synthetic datasets.

\subsubsection{Lineage inpainting can learn non-least-action dynamics}

We first ask whether lineage inpainting can learn developmental dynamics that violate least-action optimal transport assumptions. Using what we refer to as branch-modified SDEs (Section~\ref{sec:synthetic}), we construct two synthetic developmental processes that have similar marginal distributions, but distinct dynamics: one follows the least-action principle, while the other has trajectories that cross -- violating optimal transport assumptions. We take snapshots of this process, yielding lineage trees with states at the leaves, at various times, and ask if the true dynamics can be recovered from these snapshots.

We compare lineage inpainting against three baselines: multimarginal flow matching~\cite{Lipman2022-qc} (a deterministic cell-state flow baseline that does not use tree structure), GENOT with Sinkhorn couplings~\cite{Klein2024-mk}, which fits couplings between consecutive time points via entropic optimal transport (also ignoring tree structure), and GENOT with fused Gromov--Wasserstein (FGW) couplings~\cite{Klein2024-mk,Lange2024-ts}, which incorporates tree structure into the coupling cost. Full implementations and training details are given in Section~\ref{sec:baselines}.

Across every dimensionality tested, lineage inpainting recovers the true dynamics with high accuracy (\textbf{Fig.~\ref{fig:synthetic}a,b}, \textbf{Supplemental Fig.~\ref{fig:supp_synthetic}}). The GENOT baselines perform well on the first process, where the dynamics are optimal-transport-like, but degrade on the second, where the crossing structure cannot be recovered by a minimal-displacement coupling alone. The gap persists even for the FGW-based baseline, which does have access to the tree: its coupling still includes a state-space displacement cost, which favors the wrong assignment in the crossing process (Section~\ref{supp:baselines}). Multimarginal flow matching is far from the truth in both processes at every dimension, as expected for a deterministic method that cannot represent branching.

To understand if lineage inpainting can perform with more complicated data distributions, we next construct semi-synthetic datasets where real observed cells are wired together through artificial rules such that the true dynamics are known. Our approach is to construct large lineage trees out of small cell type motifs. A cell type motif is simply a repeated pattern of cell types on a lineage tree. For example, a size three motif could be a cell of type $A$ with daughters of type $B$ and type $C$ \cite{Tran2024-eu}. As visually depicted in \textbf{Fig.~\ref{fig:synthetic}c}, we construct cell type motifs by picking two real E9.0 cell types as progenitors, and four E9.5 cell types as descendants. Given such a motif, we can then construct a full lineage tree by repeatedly patterning the same motif, and randomly sampling different cells of the same type across different instances of the motif.

We test the ability of models to reconstruct the true dynamics from snapshots across 32 random choices of 6-node cell type motifs (\textbf{Fig.~\ref{fig:synthetic}c}). Since these cell types are randomly chosen, they can be optimal-transport-like (where child cell types closely resemble their parent cell types) by chance, but they need not be. We find that GENOT models generally perform well when dynamics are OT-like but fail otherwise, as we might expect. Lineage inpainting similarly performs well on OT-like motifs, but maintains strong performance across a broader range of dynamics. We can  characterize the dynamics where lineage inpainting succeeds. We call a motif Felsenstein-like if assigning the two sibling-group means to their imposed progenitors is cheaper than swapping the groups -- this echoes the inductive bias built into our architecture. We find that lineage inpainting recovers essentially every Felsenstein-like motif and fails only outside that set. In some cases, the randomly chosen cell type motif is neither OT-like or Felsenstein-like. Such cases can be pathological -- such as when two very distinct cell types switch fates (\textbf{Fig.~\ref{fig:synthetic}c}). In the sampled motifs, every motif classified as optimal-transport-like is also Felsenstein-like. Here lineage inpainting captures a broader range of motifs than the OT-based generative models.

\subsubsection{Lineage inpainting can learn context-dependent dynamics}

We next ask if lineage inpainting models are able to learn context-dependent dynamics. That is, if they are able to extract information from a cell's lineage context to improve predictions of its past and future. Again, we first test this using fully synthetic data generated by branch-modified SDEs. We construct a dataset where a progenitor commits stochastically to one of two fates at an early time. For several subsequent divisions the committed fate biases each descendant's state only weakly, such that it can be obscured by noise (with scale $\sigma$). Several generations later the two fates diverge clearly.

When $\sigma$ is large, a method that processes cells one at a time has access only to the very noisy information about early commitment steps. In contrast, a method that aggregates evidence across a clone should be able to make a more accurate prediction by averaging over it. Because the lineage inpainting model attends over cells in a subtree simultaneously, it can in principle perform this aggregation. Indeed, we find that lineage inpainting is able to extract this early fate bias from lineage context (\textbf{Fig.~\ref{fig:synthetic}d,e}). At low noise ($\sigma=0.2$), where the fate bias is clear from a single cell, all methods perform well, and increasing the context adds little. However, as noise increases, lineage inpainting performance improves with context size and further outperforms alternatives.

Having shown that lineage inpainting benefits from lineage context in synthetic datasets, we now move to a semi-synthetic setting with real cells. We again use the same kind of cell type motif, but with a different motif structure, which resembles a transcriptional convergence event. We build trees using a mixture of two motifs that share a child type $B$ (\textbf{Fig.~\ref{fig:synthetic}f}) -- in other words, there are two paths to making the same cell type. In both motifs, the parent of the shared child type is determined by the type of its cousin. To infer the ancestor of a type $B$ cell, a model must therefore copy its cousin's cell type.

Models which operate at the single-cell level cannot perform this context copying operation, since they take only a single cell as input. Lineage inpainting, however, has access to the cousin cell type at test time. Across 10 random instances of these motifs, lineage inpainting is able to correctly assign the cell type of the ancestral cells reliably (\textbf{Fig.~\ref{fig:synthetic}f}).

Together, these benchmarks show that lineage inpainting recovers couplings that minimal displacement cannot, and resolves dynamics that no single cell's state determines, in both simulated and real cell states.
\FloatBarrier

\begin{figure}[p]
  \centering
  \includegraphics[width=\linewidth]{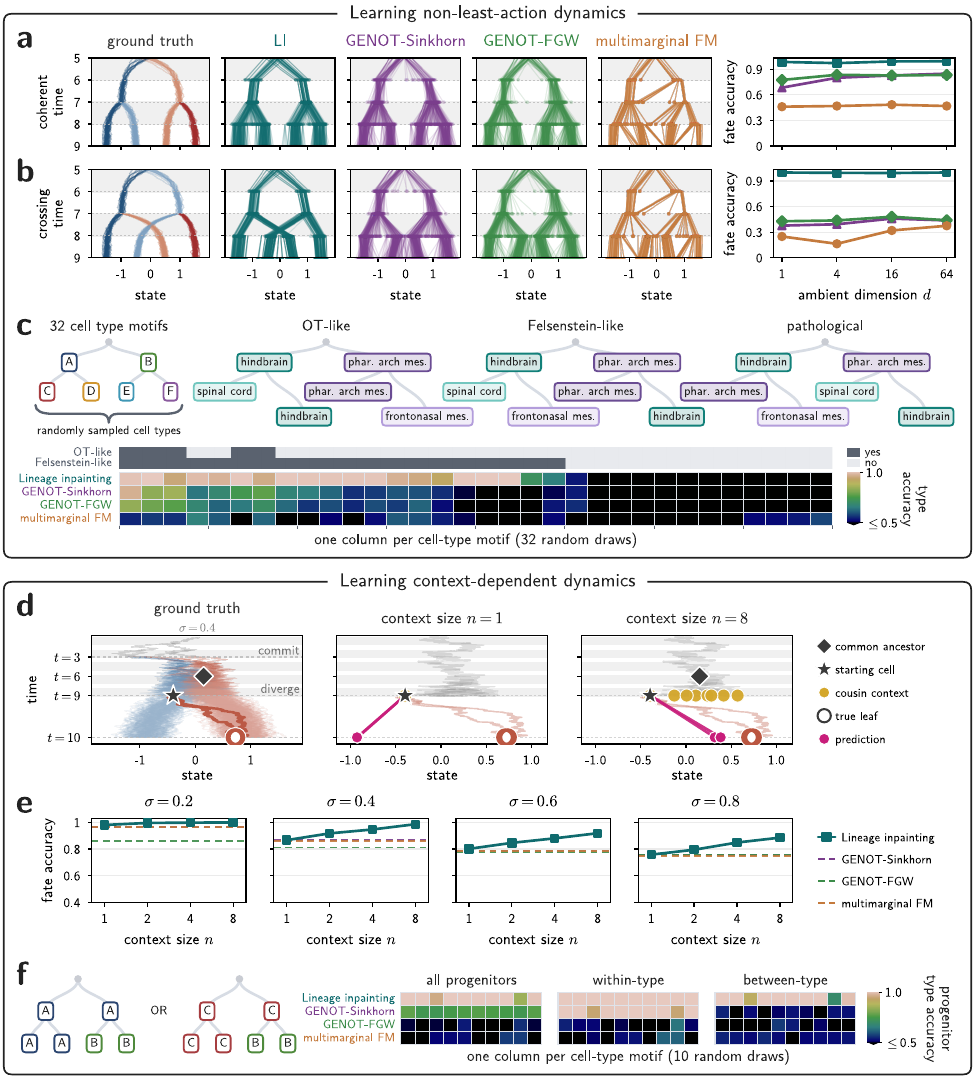}
  \caption{\textbf{Lineage inpainting leverages tree topology and lineage context to reconstruct synthetic developmental processes.} \captionmore}
  \label{fig:synthetic}
\end{figure}

\contcaption{%
  \textbf{(a)} (left) Ground truth and reconstructed dynamics for an optimal-transport-like developmental process. (right) Mean fate accuracy of $7 \to 8$ predictions of each model, across the 1D dynamics linearly projected to higher dimensions.
  \textbf{(b)} Identical to (a), but with crossing trajectories that violate optimal transport assumptions.
  \textbf{(c)} (top left) Cartoon illustrating design scheme for synthetic datasets constructed from random cell type motifs. (top right) Example possible cell type motifs, classified as optimal-transport-like, Felsenstein-like, or neither. (bottom) Performance across the 32 sampled motifs for each method, measured by child-type accuracy (colorbar clipped below at $0.5$).
  \textbf{(d)} Ground truth and reconstructed dynamics for a noisy early fate commitment, across different context sizes. Input cousin context shown in yellow, and lineage inpainting predictions shown in pink.
  \textbf{(e)} Lineage inpainting performance as a function of context size, across early fate commitment dynamics with varying noise levels. Baselines (which cannot scale with context) are shown in dotted lines.
  \textbf{(f)} (left) Cartoon illustrating a design scheme for synthetic datasets with context-dependent dynamics, in which the progenitor identity of type $B$ children is determined by their cousins. (right) Progenitor type accuracy (whether the imputed ancestor is automatically assigned the correct cell type) over 10 random draws of motifs. Accuracy shown across all progenitors and split by whether the transition is within or between types (e.g. $A \to A$ or $A \to B$).
}

\subsection{Modeling mouse embryogenesis using a lineage-traced atlas}

Having established the model's performance on synthetic data, we next apply lineage inpainting to our primary setting: mouse embryogenesis. We assemble a $\sim$4-million-cell atlas spanning E7.5 to E10.0, combining 1.4 million lineage-traced cells from \cite{colgan} with non-lineage-traced cells from the same dataset and from \citet{Qiu2024-cq}. The cells from these two datasets are integrated in a single scVI model \cite{Lopez2018}. We fit a single lineage inpainting model to this combined dataset, operating on the scVI latent space (see Sections~\ref{sec:data} and~\ref{sec:coupling}).

To visualize the ancestor--descendant coupling learned by the model, we construct a cell-type-level transition map spanning E7.5 to E10.0 (\textbf{Fig.~\ref{fig:f3_summary}a}). For each observed time point, we take the trees with leaves at that time point, sample their internal node states near the previous observed time point. We then label these imputed states by their nearest neighbors and aggregate into a cell type transition map (Section~\ref{sec:sankey-transition-maps}). Despite being completely blind to labels during training time, lineage inpainting recovers a cell-type transition map which largely recapitulates known transitions. Compared with the manually curated list of cell state transitions presented in \citet{Qiu2024-cq}, 93\% of inferred single-cell transitions are consistent, with the remainder detailed in \textbf{Supplemental Fig.~\ref{fig:supp_summary}}. Notably, this map captures the multiple mesodermal origins of endothelial cells \cite{Stone2019-hb, pardanaud_two_1996, couly_angiogenic_1995, Sahai-Hernandez2023-uy}. %
We next ask if the individual state transitions are faithful at single-cell resolution. We evaluate the model at the single-cell level using six metrics which evaluate the distribution of states, individual trajectories, and the structure of the tree as a whole (Section~\ref{sec:metrics}). Two metrics are purely data-driven and do not rely on cell type labels or biological prior knowledge: \emph{clonal distortion}, which measures whether pushing a clone back to an ancestral state and then forward recovers the clone's original distribution of states, and \emph{marginal distortion}, which measures the disagreement between the model's imputed distributions and the empirical cell-state distribution at each observed time point. Three metrics assess individual trajectories and draw on cell type annotations: how often inferred transitions cross germ layers (\textit{germ layer crossing rate}) or oscillate between cell-type labels (\textit{lineage oscillation rate}), and how well ancestor and descendant expression is correlated for Hox genes known to remain stable within neural ectoderm and neural crest cell types (\emph{Hox code distortion}) \cite{Alexander2009-hox}. Finally, \emph{linkage distortion} compares the ancestral-linkage matrix \cite{Fang2022-tg, colgan} -- a summary of how recently pairs of cell types shared a common progenitor -- between real and model-imputed trees, quantifying whether the model preserves tree structure rather than individual trajectories or state distributions alone.

Across all six metrics, lineage inpainting outperforms two optimal transport baselines fit to the same atlas, GENOT-Sinkhorn and GENOT-cbFGW \cite{Klein2024-mk} (\textbf{Fig.~\ref{fig:f3_summary}b}, \textbf{Supplemental Fig.~\ref{fig:supp_holdout}}).

We next ask how much lineage inpainting uses lineage context to refine its predictions, separately for ancestor and descendant prediction (Section~\ref{sec:context-scaling}). Because neither direction has ground-truth states, we score both by cycle consistency: whether states pushed to another time and back recover the starting distribution. For ancestral inference, we fix a target internal node and vary the amount of descendant context provided as input. For descendant inference, we infer a fixed target leaf set from ancestral snapshots with varying numbers of cousins and conditioning times.

In both directions, increasing context improves performance (\textbf{Fig.~\ref{fig:f3_summary}c,e}). At the smallest context sizes, lineage inpainting is comparable to or worse than the GENOT baselines. However, with additional context, it continues to improve well beyond them. The benefit of added context is not uniform across cell types (\textbf{Fig.~\ref{fig:f3_summary}d,f}). For ancestral inference, the largest gains appear in epidermal progenitors, while the cervical spinal cord benefits least. For descendant inference, the endothelial cells gain the most and the cervical spinal cord the least.

Taken together, these results indicate that lineage inpainting learns biologically sensible cell-state transitions from the mouse embryogenesis atlas.  Its performance increases as a result of its ability to make use of lineage context.

\begin{figure}[p]
    \centering
    \includegraphics[width=\linewidth]{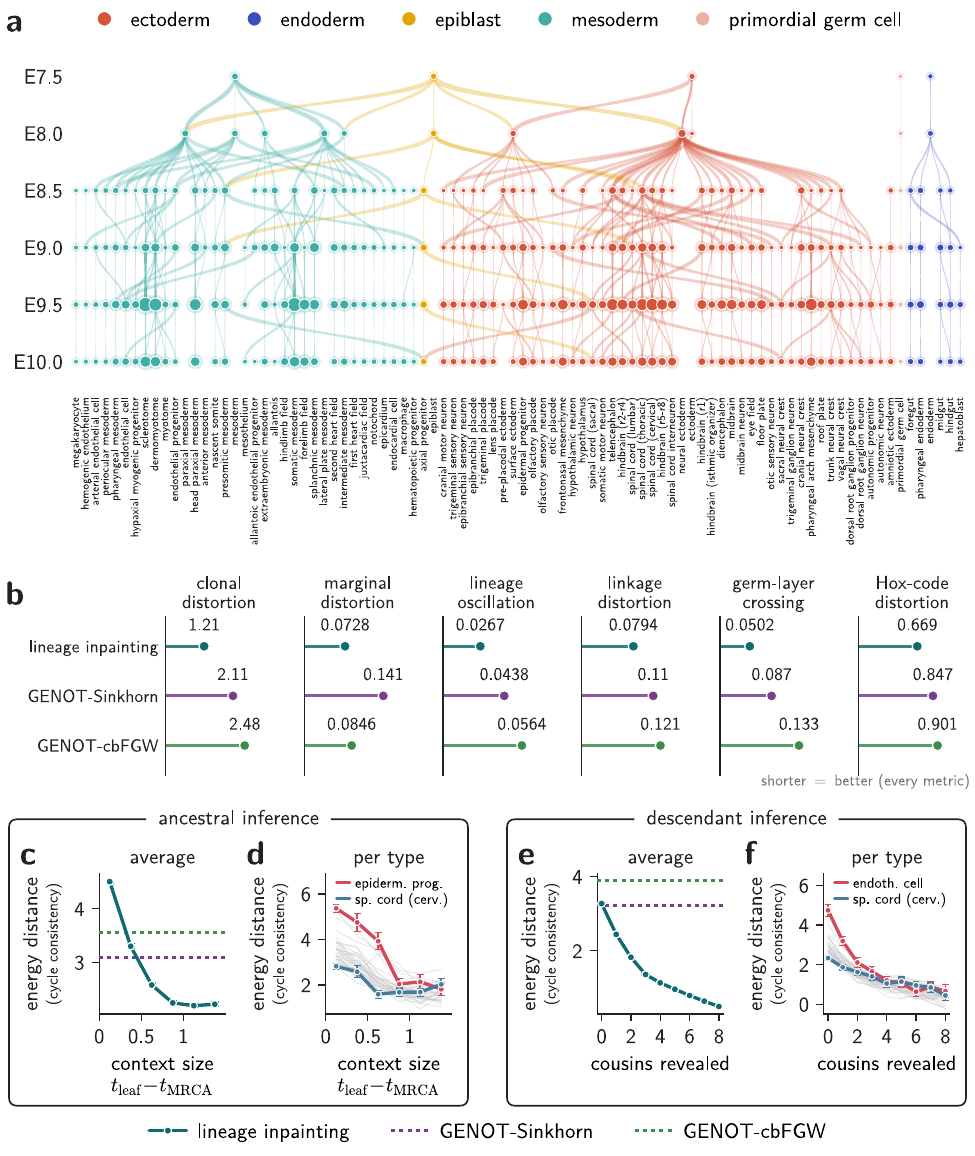}
    \caption{\textbf{Summarizing cell state transitions inferred by lineage inpainting.} \captionmore}
    \label{fig:f3_summary}
\end{figure}

\contcaption{%
    \textbf{(a)} Transition map inferred by lineage inpainting, coarse-grained to the cell type level, ordered by timepoint and colored by germ layer. Node size scales with cell count. Edge thresholding method is provided in Section~\ref{sec:sankey-transition-maps}.
    \textbf{(b)} Quality control metrics across transition maps inferred by lineage inpainting, GENOT-Sinkhorn and GENOT-cbFGW. Lower is better across all metrics. Metric definitions are provided in Section~\ref{sec:metrics}.
    \textbf{(c)} Model performance, measured by cycle consistency (leaves $\to$ root $\to$ leaves), as a function of the temporal extent of revealed context ($t_{\rm leaf}-t_{\rm MRCA}$, in embryo days). Dotted lines are GENOT baselines which do not scale with context.
    \textbf{(d)} Identical to (c), stratified by cell type, each shown in gray lines. The two fates gaining most and least from context are highlighted by red and blue curves respectively. Epiderm. prog. = epidermal progenitor and sp. cord (cerv.) = spinal cord (cervical)
    \textbf{(e)} Cycle consistency performance as a function of the number of cousins used in the descendant inference step.
    \textbf{(f)} Identical to (e), stratified by cell type. Again, the most and least context benefiting cell types are shown in red and blue respectively. Endoth. cell = Endothelial cell. Error bars in (d,f) are descriptive standard errors across evaluation rows.
}

\subsection{Lineage inpainting resolves transcriptional convergence dynamics}

Having fit a lineage inpainting model of mouse development, we ask how its inferred dynamics differ from those of optimal transport in the gene expression latent space. We compare against GENOT-Sinkhorn \cite{Klein2024-mk}. To make the comparison, we sample a discrete coupling between cells from each model (Section~\ref{sec:baselines}). Across all adjacent time points, we measured the disagreement between the models as the L2 distance between their two predicted progenitor states (\textbf{Fig.~\ref{fig:f4_convergence}a}).

The disagreement depends on cell type. For some cell types it is small, and OT and lineage inpainting recovered similar ancestors. For others -- endothelial cells and several classes of differentiated neurons -- the difference is large (\textbf{Fig.~\ref{fig:f4_convergence}b}). We ask what distinguishes these cell types.

From our synthetic and semi-synthetic experiments, we know that cell-level models struggle to capture convergence events, while lineage inpainting is able to resolve them using lineage context (recall \textbf{Fig.~\ref{fig:synthetic}f}). To test if this is the source of the disagreement between the models in the mouse development setting, we need a way to score convergence. Convergence is a statement about origins rather than present states: if cells that look alike today descend from a single progenitor population, their inferred ancestors should also look alike, whereas if the same state was reached independently from diverse origins, those ancestors will be scattered across the expression space.

We turn this intuition into a per-cell score by taking the 100 nearest neighbors of each cell in the scVI latent space and measuring the transcriptional variance among their inferred progenitors, relative to the local spacing of cells (Section~\ref{sec:convergence-analysis}). Because the neighbors are by construction similar in their current state, any remaining variance among their ancestors reflects divergence in origin rather than in present identity. Across cell types, the convergence score aligned with the OT-vs-lineage inpainting disagreement (Spearman $\rho=0.78$), and the same cell types were high on both (\textbf{Fig.~\ref{fig:f4_convergence}b}, \textbf{Supplemental Fig.~\ref{fig:supp_conv}}).

Endothelial cells are the clearest example of transcriptional convergence. Their inferred progenitors fall largely into five cell types: somatic mesoderm, dermomyotome, splanchnic mesoderm, head paraxial mesoderm and allantois (\textbf{Fig.~\ref{fig:f4_convergence}c}), suggesting that endothelium is generated independently by paraxial, lateral plate and extraembryonic mesoderm. This model prediction is consistent with prior literature establishing that endothelium arises from multiple distinct lineage origins \cite{Stone2019-hb, pardanaud_two_1996, couly_angiogenic_1995, Sahai-Hernandez2023-uy}.

The model reaches this conclusion from lineage context rather than expression alone. At E9.5 endothelial cells sit in clades composed predominantly of leaves of a single other type, with different endothelial cells embedded among different types (\textbf{Fig.~\ref{fig:f4_convergence}c}), and the model assigns ancestors accordingly. Cell-level models such as GENOT have no access to this clonal information at test time, and consequently struggle to distinguish endothelial cells of different origins based on their transcriptomes alone.

\subsection{Quantifying transcriptional memory and erasure in convergent differentiation}

When distinct lineages converge on a common cell type, what happens to the differences between them? In 1976, \citet{Lewis1976-zo} proposed that cells could have the same functional identity while retaining distinct signatures associated with their spatial location (or positional information), an idea they called non-equivalence. In transcriptional convergence events where multiple cell types are produced independently in different regions of the embryo, non-equivalence suggests that convergent cell types may retain signatures of their origin.

Targeted experimental studies have tested non-equivalence in individual cell types. Some cell types appear molecularly equivalent across lineages: macrophages arising from distinct origins placed in the same niche become near indistinguishable \cite{van-de-Laar2016-yd}. In contrast, fibroblasts epigenetically maintain their embryonic Hox code expression \cite{Rinn2008-zv}. These individual case studies show that origin-specific differences can persist to varying extents in different cell types. Here, we use lineage inpainting to quantify origin-specific transcriptional memory across multiple convergence events in the embryo.

We focus on three cell types with high convergence scores: endothelial cells, spinal cord interneurons and somatic motor neurons (\textbf{Fig.~\ref{fig:f4_convergence}b}). For endothelium, we compare origins in paraxial mesoderm, a tissue along the sides of the developing spine, and lateral plate mesoderm, which lies further to the sides. For the neurons, we compare origins along the anterior-posterior axis, which we separate using spinal cord and hindbrain annotations \cite{metzis_nervous_2018}. We defined transcriptional programs distinguishing these sources, together with cell type identity programs, by linear discriminant analysis between labeled cells (Section~\ref{sec:convergence-analysis}, \textbf{Fig.~\ref{fig:f4_convergence}d}).

Binning the inferred trajectories by their E8.5 source and following them along both programs shows that across all three convergence events, cells first diversify along the source-associated program, then establish cell type identity with similar dynamics across distinct sources (\textbf{Fig.~\ref{fig:f4_convergence}e}). Notably, the dynamics of diversification are not the same across types: endothelial progenitors are already separated along the paraxial versus lateral plate program at E7.5, before the endothelial identity program is established, whereas the neuron progenitors remain similar on the anterior-posterior program until after E7.5. The stability of the positional transcriptional program is also distinct between ECs and neurons. Endothelial cells begin to lose their spatial distinctness as the identity program establishes: the separation peaks near E8.5 and then narrows through E9.5. In contrast, the neurons hold theirs stably after E8.5.

These results suggest that neurons retain more memory of their origin than endothelial cells do. We next quantify this memory at the level of individual genes. For each convergence event, we measure expression differences between source populations at E8.5, for instance between paraxial-derived and lateral plate-derived endothelial progenitors, and compare them with differences between their differentiated descendants at E9.5 (\textbf{Fig.~\ref{fig:f4_convergence}f}). We call genes with origin-associated differences at both times \textit{memory genes}, and those whose differences are no longer detected at E9.5 \textit{erased genes}.

All convergence events have some memory genes and some erased genes, but the proportions vary across cell types (\textbf{Fig.~\ref{fig:f4_convergence}g}). Endothelium retains a number of memory genes but erases a far larger share of its origin-specific differences than any of the neuronal types, which retain considerably more. This is consistent with the regional specialization of neurons \cite{Jessell2000-vr, forlani_acquisition_2003, wellik_hox_2007}. We find that spatial markers -- Hox genes in particular -- are retained across every event we examined, including \texttt{Hoxa7}, \texttt{Hoxa9} and \texttt{Hoxb9} in endothelium and spinal cord interneurons and \texttt{Hoxc6} and \texttt{Hoxc9} in somatic motor neurons.

Across these convergence events, cells acquire a shared identity while retaining some transcriptional differences between their origins, consistent with the idea of non-equivalence. The extent and dynamics of transcriptional memory retention and erasure vary between cell types.

\FloatBarrier

\begin{figure}[p]
  \centering
  \includegraphics[width=\linewidth]{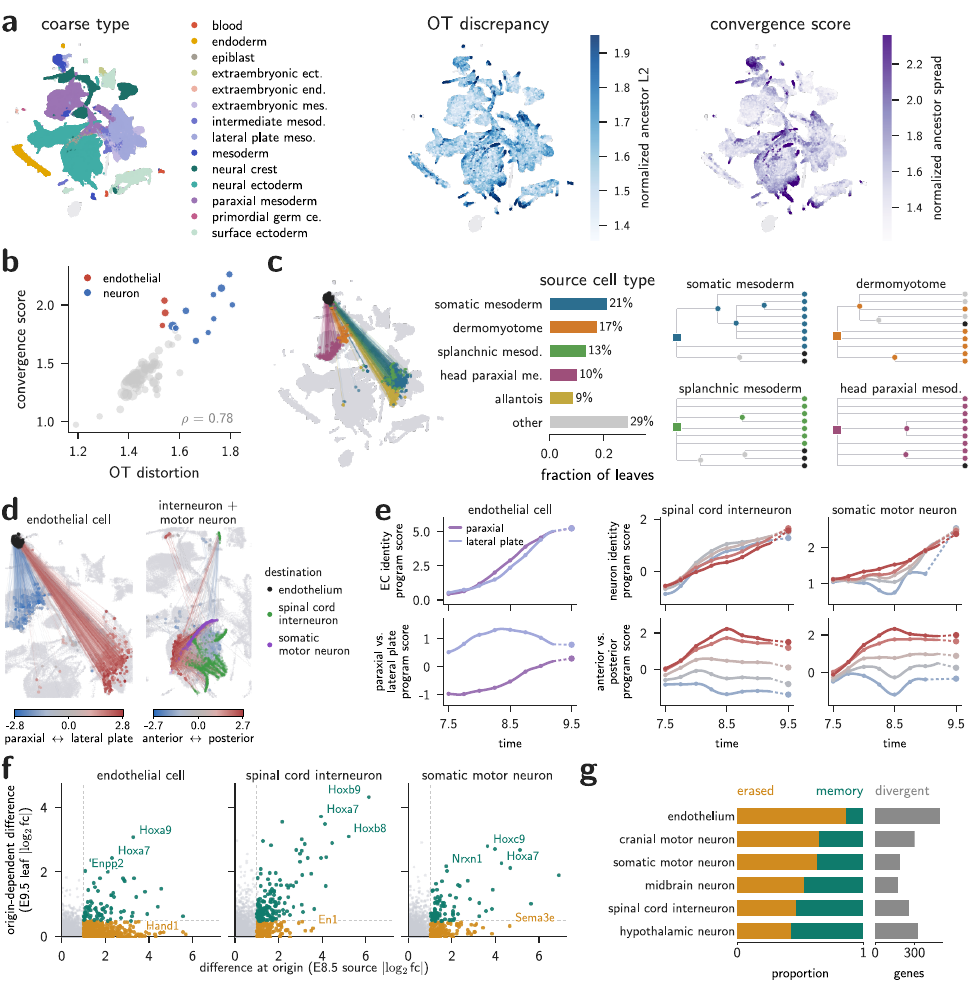}
  \caption{\textbf{Lineage inpainting captures transcriptional convergence dynamics.} \captionmore}
  \label{fig:f4_convergence}
\end{figure}

\contcaption{%
  \textbf{(a)} UMAPs of the E8.0-E10.0 lineage-traced cells. (left) Hued by coarse cell type, (center) by the smoothed discrepancy between GENOT-Sinkhorn and lineage inpainting inferred ancestors, (right) by convergence score (Section~\ref{sec:convergence-analysis}).
  \textbf{(b)} Convergence score and OT discrepancy aggregated at the cell type level (one point per cell type). Neurons and endothelial cell types are highlighted.
  \textbf{(c)} (left) Origin of E9.5 endothelial cell leaves, traced to their inferred E8.5 source ancestor and hued by source type. (center) Bar chart showing source-type composition. (right) Representative subtrees showing observed leaves and inferred ancestors of endothelial cells.
  \textbf{(d)} Origin of E9.5 leaves, hued by spatial program score: endothelial cells hued by lateral plate vs. paraxial program, and spinal cord interneurons together with somatic motor neurons, both hued by an anterior vs. posterior program.
  \textbf{(e)} Program scores along inferred trajectories. Trajectories are binned into equal sizes by spatial program score at E8.5. Hue for neurons matches color bar in (d).
  \textbf{(f)} Quantifying gene expression memory in convergence events, by comparing source-dependent differential expression at E9.5 and differential expression between sources at E8.5. Genes distinct between cells of the two origins at both times are considered memory genes (teal) while those whose E8.5 difference is gone by E9.5 are considered erased (orange).
  \textbf{(g)} Proportion of memory and erased genes across six convergent cell types, with the number of origin-distinct genes behind each bar shown in gray.
}

\begin{figure}[p]
    \centering
    \includegraphics[width=\linewidth]{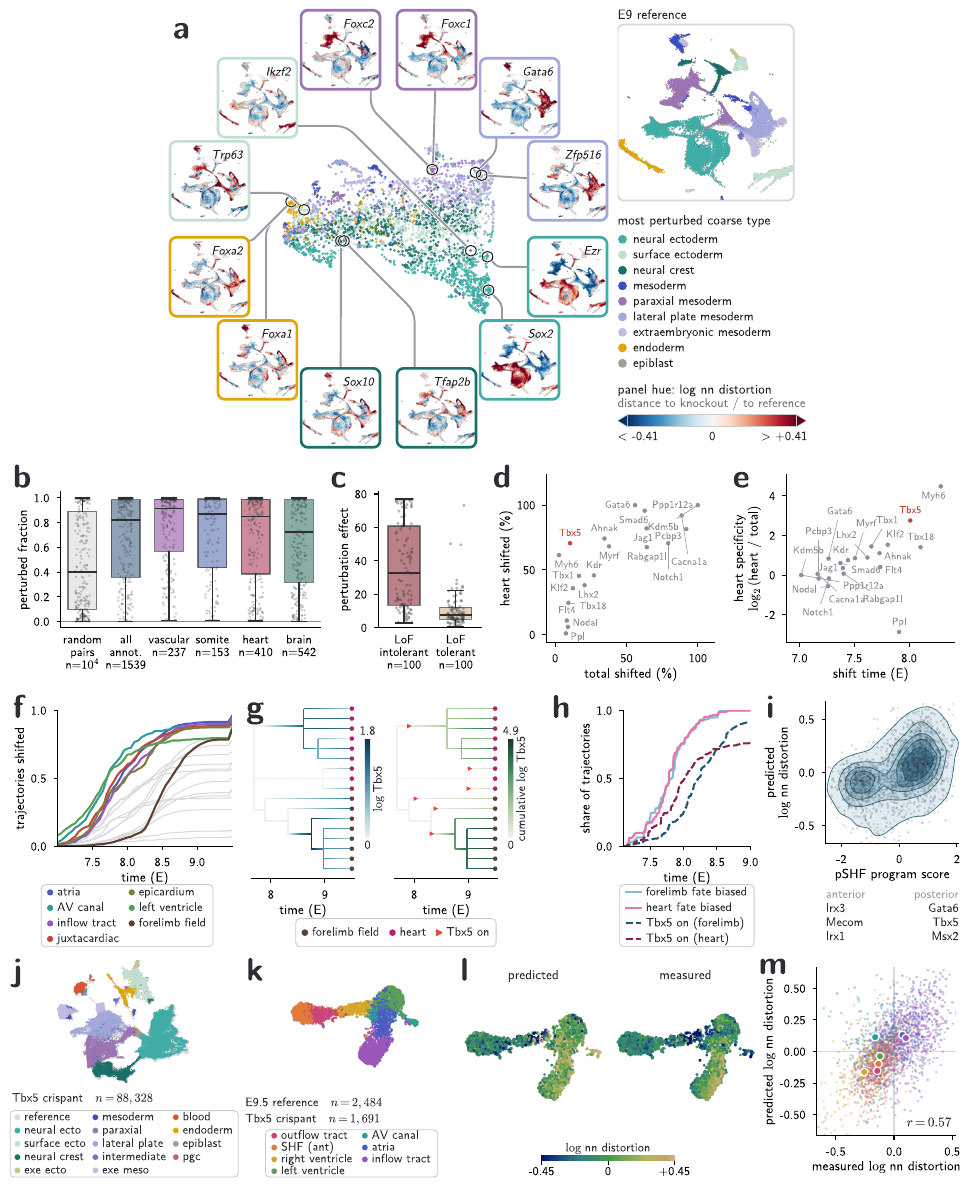}
    \caption{\textbf{Lineage inpainting suggests developmental effects of genetic perturbations.} \captionmore}
    \label{fig:f5_pert}
\end{figure}

\contcaption{%
    \textbf{(a)} Visualization of $3{,}939$ simulated perturbations by UMAP of mean perturbation effect at E9.0. Each point corresponds to a simulated knockout embryo, hued by the coarse cell type label it most depletes. Inset UMAPs show knockout effects at single cell resolution on the E9.0 embryo UMAP (log nearest-neighbor distortion: distance from reference cell to nearest neighbor in simulated knockout embryo, normalized to nearest-neighbor distance of a control; Section~\ref{sec:nn-distortion}), for knockouts which most specifically deplete selected coarse cell types.
    \textbf{(b)} Fraction of cell type $c$ perturbed by knockout of gene $g$, over $1{,}539$ annotated $(g, c)$ pairs, shown against a matched random-pair null and separated by broad tissue of $c$.
    \textbf{(c)} Predicted perturbation effect strength ($L_2$ norm of cell type abundance log fold changes) for the $100$ most and $100$ least loss-of-function constrained screen genes, binned by gnomAD LOEUF.
    \textbf{(d)} Share of heart trajectories shifted versus total trajectories shifted at E9.0 for 21 overlapping genes from \citet{Sierant2025-lx}.
    \textbf{(e)} Heart specificity (defined as log-ratio of heart shift to total shift fraction) against the mean time when trajectories are shifted from the reference tree.
    \textbf{(f)} Trajectories shifted in the \texttt{Tbx5} knockout as a function of time, resolved by cell subtype, with the heart subtypes and the forelimb field highlighted.
    \textbf{(g)} Inferred \texttt{Tbx5} expression along an example subtree pruned to forelimb field and heart, shown with instantaneous expression (left) and accumulated along each path (right), with the onset of expression marked.
    \textbf{(h)} Cumulative share of trajectories that have become biased to a single fate (solid) and that have turned \texttt{Tbx5} on (dashed), separately for heart and forelimb across all E9.5 trees.
    \textbf{(i)} Predicted nearest-neighbor distortion against posterior second heart field program score, for all heart tube reference cells. The transcription factors marking each end of the program axis are listed below the panel.
    \textbf{(j)} Integrated latent space over the E9.5 reference and $88{,}328$ \texttt{Tbx5} crispant cells, hued by lineage.
    \textbf{(k)} The E9.5 heart tube, reference and \texttt{Tbx5} crispant hued by subtype.
    \textbf{(l)} Predicted and measured nearest-neighbor distortion over the E9.5 heart tube.
    \textbf{(m)} Predicted against measured distortion, one point per reference cell, with subtype means overlaid.
}

\subsection{Reasoning about developmental effects of genetic perturbations}

Lineage inpainting provides a model of developmental dynamics in unperturbed embryos. We next ask if this purely observational model can be used to reason about the effects of genetic perturbations. 
Intuitively, a model of developmental lineage trees is a good source for reasonable hypotheses about genetic perturbation effects. For any gene, the model shows which progenitor cells express it and what those progenitors give rise to. It also shows the developmental output of similar progenitors that do not express the gene. Together, these can be used to hypothesize an effect of genetic knockout.

To test this, we treat our observational model as causal: perturbing a gene alters the states of individual cells, but leaves the state transition rules unchanged, such that the same forward model can be applied to cells that a perturbation has displaced in state space. Under this assumption, we simulate a knockout by projecting perturbed cells to their nearest neighbors that do not express target gene \cite{Fijabi2026-sq}, and propagate them forward through the tree, exactly as we do in the unperturbed case. Cells lacking nearby non-expressing observations are considered off-manifold and are not propagated further.

We simulate the knockout of each of the $3{,}956$ highly variable genes on which the model is fit and categorize perturbations based on which coarse cell type they most severely deplete. Within each category, the most specific depleting perturbations are usually its well-known regulators. In particular, \texttt{Sox10} is the most specific neural crest-depleting perturbation \cite{Britsch2001}, \texttt{Foxc1/2} for the paraxial mesoderm \cite{Wilm2004}, \texttt{Foxa1/2} for the endoderm \cite{Lee2005}, and so on. These textbook regulators were not chosen manually, but rather emerge from an untargeted screen over all highly variable genes, ordered by cell type depletion specificity.

To more broadly test lineage inpainting predictions for genetic knockout effects, we compare the predicted cell type abundances to knockout phenotypes annotated in the Mouse Genome Database (MGD) \cite{Eppig2015-bc}. We filtered to MGD phenotype annotations indicating abnormal development of a structure for which we have corresponding cell type annotations. 
This yields $1{,}539$ gene to cell type pairs, each of which represents a cell type known to be affected by a defect in the paired gene. We then measure, for each pair, the extent to which genetic knockout is predicted to perturb the paired cell type (\textbf{Fig.~\ref{fig:f5_pert}b}, \textbf{Supplemental Fig.~\ref{fig:supp_mgi}}). We find that annotated pairs have a larger fraction of perturbed cells (median $82\%$) compared with a background of randomly paired cell types and genes (median $40\%$).

Agreement with the MGD annotations varies by tissue (\textbf{Fig.~\ref{fig:f5_pert}b}). Median perturbed fraction is high for vascular ($91\%$) and cardiac ($85\%$) cell types but  lower for neural cell types ($73\%$). This likely reflects the developmental window of the model: neuronal differentiation is just beginning around E9.0, and many MGD neural phenotypes may arise later.

We next ask if these predictions are consistent with human loss-of-function effects, using gnomAD gene constraint scores \cite{Lek2016-wm, Karczewski2020-ki} as a measure of the consequences of genetic perturbation. If the model captures developmental effects that generalize across species, genes more strongly constrained in humans should have more severe predicted effects. We identify the 100 most and least constrained of the $3{,}591$ screened genes with human orthologs (LOEUF $\leq 0.22$ and $\geq 1.75$, respectively), and evaluate their predicted perturbation severity (measured as the $L_2$ norm of predicted cell type abundance log fold changes from the reference). Highly constrained genes have a predicted median effect score of $32.6$, compared to $7.4$ for unconstrained genes.

Mutational constraint indicates how essential a gene is, but not the nature of its effect. The remainder of this section focuses on the heart, where we can develop a finer-grained understanding of the model predictions and capabilities. %
We first consider the model's predictions for a set of genes implicated in congenital heart disorders (CHD) \cite{Sierant2025-lx}.
Of the 21 genes shared between our highly variable genes and the database, 19 are predicted to deplete heart cells by over 10\% (\textbf{Fig.~\ref{fig:f5_pert}d}). %

Beyond predicting whether a perturbation affects the heart, the model can also predict when a perturbation effect begins. By comparing the counterfactual trees to the control trees, we can identify when a perturbed lineage first diverges from its unperturbed counterpart, which we refer to as the shift time. Across the 21 CHD-related genes, we find that shift time aligns closely with specificity of heart depletion (\textbf{Fig.~\ref{fig:f5_pert}e}). This aligns with the intuition that a perturbation that removes a broadly required gene diverts cells early, whereas one specific to the heart cannot act until that population has been defined.

\subsection{Lineage inpainting predicts the developmental dynamics of \texttt{Tbx5}-perturbed embryos}
We next ask how the predicted effects of \texttt{Tbx5} knockout unfold across developmental trajectories. We first ask when the model predicts that trajectories will shift in the \texttt{Tbx5} knockout, resolved by cell subtype. We find that heart trajectories are affected quite early, many prior to E8.0, while forelimb field trajectories shift much later (\textbf{Fig.~\ref{fig:f5_pert}f}). This suggests the cardiac and forelimb effects of \texttt{Tbx5} arise in separate progenitor populations, rather than through a shared progenitor \cite{ivanovitch_ventricular_2021}.
This can be appreciated by looking at inferred \texttt{Tbx5} expression along the tree, where the gene turns on only after clades have become biased to only forelimb or only heart (\textbf{Fig.~\ref{fig:f5_pert}g}). To test this across all trees, we compute for each heart and forelimb trajectory when it becomes fate-biased (generating only heart or only forelimb descendants) and when it begins to express \texttt{Tbx5}. \texttt{Tbx5} turns on after restriction to the heart or forelimb fate, rather than within earlier progenitors that produce both cell types (\textbf{Fig.~\ref{fig:f5_pert}h}).

Within the heart, the predicted shifts are not uniform. The inflow tract and the atria are strongly shifted, while the outflow tract is much less so (\textbf{Fig.~\ref{fig:f5_pert}f}), so posterior cell types are more affected than anterior ones, consistent with prior work \cite{Bruneau2001-yk, buckingham_building_2005, meilhac_deployment_2018}. We next ask if this pattern persists at a more fine-grained level than cell type labels. We measure shift magnitude as \textit{nearest-neighbor distortion}, defined for each reference cell as its distance to the nearest perturbed cell, normalized by its distance to the nearest reference cell (Section~\ref{sec:nn-distortion}). We find that the pattern persists at the single-cell level: a higher posterior program score corresponds to a higher predicted nearest-neighbor distortion (\textbf{Fig.~\ref{fig:f5_pert}i}).

We next test this prediction of distortion stratified along the anterior-posterior axis experimentally. We performed scRNA-seq at E9.5 on eight \texttt{Tbx5} crispant embryos, recovering roughly $9\cdot10^4$ cells (\textbf{Fig.~\ref{fig:f5_pert}j}, Section~\ref{sec:crispant-data}). These include $1{,}691$ heart tube cells (\textbf{Fig.~\ref{fig:f5_pert}k}). For each reference cell in the heart tube, we compute the distortion of its nearest neighbors in the crispant data (Section~\ref{sec:nn-distortion}), and compare this to the distortion predicted by lineage inpainting (\textbf{Fig.~\ref{fig:f5_pert}l}). The measured and predicted distortions are correlated at the single cell level, with $r=0.57$ (\textbf{Fig.~\ref{fig:f5_pert}m}). 

The model predictions do not perfectly match measured perturbation effects. Within the heart tube, the main disagreement is the atrioventricular canal, which the model predicts to be shifted but which is not. More globally, crispant embryos show a delay in neuronal development that lineage inpainting does not predict (\textbf{Supplemental Fig.~\ref{fig:supp_tbx5}}).

Taken together, these results suggest that a purely observational model of development can recover many known responses to genetic perturbation, despite resting on rather extreme assumptions. And beyond recapitulating known coarse developmental phenotypes, lineage inpainting provides further fine-grained predictions about the timing of a perturbation's effect and the specific transcriptional states it may deplete. In the case of \texttt{Tbx5}, these predictions separate its cardiac and forelimb requirements, and resolve a graded distortion along the anterior-posterior axis of the heart that we confirm by direct measurement of crispant embryos.

\section{Discussion}

Many biological processes unfold over time across communities of cells. However, directly observing changes in transcriptional state over time is typically not possible, as comprehensive measurements of gene expression are destructive. Instead we must reconstruct dynamics from snapshots. Yet the past and future state of a cell are not in general predictable from its current state alone, since cell fate also depends on a cell's history and the external signals it encounters. In this work, we have shown that lineage information provides a rich source of context for the reconstruction of multicellular dynamics.

We established a framework for modeling cell state dynamics across entire developmental trees using lineage context, rather than modeling individual cell trajectories independently. We introduced lineage inpainting, and showed that it can recover non-optimal-transport and context-dependent cell state dynamics. By applying lineage inpainting to mouse embryogenesis data, we resolved transcriptional convergence dynamics inaccessible to alternative approaches. Finally, we showed that despite being fit on observational data alone, lineage inpainting can be used to generate reasonable hypotheses about perturbation effects. In sum, this work takes a step toward what one might call a virtual embryo -- a model which can forecast the fates and origins of cells across the embryo and simulate counterfactuals.

While we have initially applied lineage inpainting to lineage-traced scRNA-seq, in principle the same modeling approach can extend naturally to include other types of cellular context. With lineage-traced spatial transcriptomic readouts \cite{Koblan2025-ac, Jones2026-gt}, one can simply add an additional bias term to the tree-structured transformer (\textbf{Box~\ref{box:attention}}), using a function of spatial distance rather than lineage distance.

Our work raises several threads for future study, in regard to both computational methodology as well as biology. While lineage inpainting is generative in the sense that it samples from a distribution over cell states on a lineage tree, it does not generate the structure of the tree itself and instead takes it as an input. Is it possible to simultaneously generate the lineage tree structure alongside the transcriptional states? Recent work on generative modeling for phylogenetic inference provides early evidence that such a model may be feasible \cite{Baron2026-ye, Zhou2023-hl}. 

While lineage inpainting infers cell states on lineage trees, it does not model mechanisms underlying state transitions directly. It does not have any explicit notion of gene regulatory networks, cell signaling, or cell-cell interactions. Is it possible to build mechanistic information into the model, as recent work has explored for cell-level models \cite{Wang2026-eg}? This may require extending models to integrate with other measurement modalities, such as ATAC-seq \cite{Li2025-jk}, spatial transcriptomics \cite{Chen2021-np}, or readouts of signaling pathway activation \cite{Chen2024-yt, Hao2026-yj}. Mechanistic models may lead to stronger out-of-distribution generalization for testing counterfactuals.

Overall, we believe this work demonstrates the importance of models which move beyond individual cell trajectories to considering multicellular systems across multiple scales. The ideas developed here are more general than embryogenesis alone, and may be useful in the study of cancer, inflammation, infection, regeneration and other multicellular dynamical phenomena.

\section*{Code and data availability}

The code necessary to reproduce the results of this paper are available at \href{https://github.com/g-kl/lineage-inpainting}{this GitHub repository}\footnote{(\texttt{g-kl/lineage-inpainting})}. The data necessary to reproduce the results of this paper, including the Tbx5 mutant scRNA-seq data, are provided on Zenodo, at \texttt{10.5281/zenodo.23750040}.

\section*{Acknowledgments}

We thank Allon Klein, Yunyi Shen, and all members of the Weissman lab for thoughtful discussions. We thank the MIT Laboratory for Information and Decision Systems computational cluster (particularly B. Jones for support), and the Kempner Institute AI cluster for computational resources. We also thank L. Wittler, J. Feidler, A. Landsberger, M. Peetz, and K. Zill of the MPI Transgene Unit for their support and expertise, as well as the MPIMG Sequencing Core.

\section*{Funding}

G.G. was supported by the Tayebati Fellowship Program. P.R. was partially supported by National Science Foundation (NSF) grant DMS-2509011. L.W.K. was supported by a Helen Hay Whitney/HHMI fellowship and the Eunice Kennedy Shriver NICHD Pathway to Independence Award (NIH K99HD118574). J.S.W. was supported by the Howard Hughes Medical Institute, National Institutes of Health (NIH) Centers of Excellence in Genomic Science grant RM1HG009490, Chan Zuckerberg Initiative grant 2024-346405 (5022), and the Whitehead Innovation Initiative. Z.D.S. was supported by National Institutes of Health (NIH) Director's New Innovator Award DP2HD108774, the Mathers Foundation, the Chen Innovation Award, and the Max Planck Society. N.F. and M.Z. gratefully acknowledge support by the Biswas Computational Biology Initiative in partnership with Milken Institute and the Kempner Institute for the Study of Natural and Artificial Intelligence at Harvard University. N.F. is also supported by T32GM135117 from the National Institute of General Medical Sciences. S.A.K is supported by the David Abraham Fellowship.

\section*{Author contributions}

Conceptualization: GG, NF, WNC, LWK, JSW, SB, PR, ZDS.
Methodology: GG, NF, WNC, LWK, JSW, SB, PR, YK, MW.
Software: GG, NF, WNC, YK.
Investigation: GG, NF, WNC, MW, YK, LWK, SAK, AZ, MZ, ZDS, SB, PR, JSW.
Visualization: GG, NF.
Supervision: JSW, PR, SB, ZDS, MZ, LWK.
Writing -- original draft: GG, NF, WNC.
Writing -- review and editing: GG, NF, WNC, MW, YK, LWK, SAK, AZ, MZ, ZDS, SB, PR, JSW.

\section*{Competing interests}

J.S.W. declares outside interests in 5AM Ventures, Amgen, nChroma Bio, DEM Biosciences, KSQ Therapeutics, Maze Therapeutics, Tenaya Therapeutics, Tessera Therapeutics, Thermo Fisher Scientific, and Xaira Therapeutics. Z.D.S. is an academic co-founder of and scientific advisor to Harbinger Health.

\section*{AI use statement}
We acknowledge the use of large language models during the preparation of this manuscript. These models were used
to assist in: (1) the writing of the code, (2) the development of the proofs, and (3) the editing of the manuscript. The authors take full responsibility for the content, scientific claims, citations, and conclusions presented in this work.

\bibliographystyle{unsrtnat}
\bibliography{paperpile,nic,william}

\begin{thebibliography}{120}
\providecommand{\natexlab}[1]{#1}
\providecommand{\url}[1]{\texttt{#1}}
\expandafter\ifx\csname urlstyle\endcsname\relax
  \providecommand{\doi}[1]{doi: #1}\else
  \providecommand{\doi}{doi: \begingroup \urlstyle{rm}\Url}\fi

\bibitem[Whitman(1878)]{Whitman1878-fb}
Charles~Otis Whitman.
\newblock {The embryology of Clepsine}.
\newblock \emph{Journal of Cell Science}, S2-18\penalty0 (71):\penalty0 215--315, 1~July 1878.
\newblock ISSN 0021-9533,1477-9137.
\newblock \doi{10.1242/jcs.s2-18.71.215}.

\bibitem[Wilson(1892)]{Wilson1892-fp}
Edmund~B Wilson.
\newblock {The cell-lineage of Nereis. \textit{A contribution to the cytogeny of the annelid body}}.
\newblock \emph{Journal of Morphology}, 6\penalty0 (3):\penalty0 361--480, July 1892.
\newblock ISSN 0362-2525,1097-4687.
\newblock \doi{10.1002/jmor.1050060301}.

\bibitem[Wagner and Klein(2020)]{Wagner2020-xr}
Daniel~E Wagner and Allon~M Klein.
\newblock {Lineage tracing meets single-cell omics: opportunities and challenges}.
\newblock \emph{Nature reviews. Genetics}, 21\penalty0 (7):\penalty0 410--427, July 2020.
\newblock ISSN 1471-0056,1471-0064.
\newblock \doi{10.1038/s41576-020-0223-2}.

\bibitem[Liberali and Schier(2024)]{Liberali2024-uy}
Prisca Liberali and Alexander~F Schier.
\newblock {The evolution of developmental biology through conceptual and technological revolutions}.
\newblock \emph{Cell}, 187\penalty0 (14):\penalty0 3461--3495, 11~July 2024.
\newblock ISSN 0092-8674,1097-4172.
\newblock \doi{10.1016/j.cell.2024.05.053}.

\bibitem[Wagner et~al.(2018)Wagner, Weinreb, Collins, Briggs, Megason, and Klein]{Wagner2018-cc}
Daniel~E Wagner, Caleb Weinreb, Zach~M Collins, James~A Briggs, Sean~G Megason, and Allon~M Klein.
\newblock {Single-cell mapping of gene expression landscapes and lineage in the zebrafish embryo}.
\newblock \emph{Science (New York, N.Y.)}, 360\penalty0 (6392):\penalty0 981--987, 1~June 2018.
\newblock ISSN 0036-8075,1095-9203.
\newblock \doi{10.1126/science.aar4362}.

\bibitem[Briggs et~al.(2018)Briggs, Weinreb, Wagner, Megason, Peshkin, Kirschner, and Klein]{Briggs2018-po}
James~A Briggs, Caleb Weinreb, Daniel~E Wagner, Sean Megason, Leonid Peshkin, Marc~W Kirschner, and Allon~M Klein.
\newblock {The dynamics of gene expression in vertebrate embryogenesis at single-cell resolution}.
\newblock \emph{Science}, 360\penalty0 (6392), 1~June 2018.
\newblock ISSN 0036-8075,1095-9203.
\newblock \doi{10.1126/science.aar5780}.

\bibitem[Packer et~al.(2019)Packer, Zhu, Huynh, Sivaramakrishnan, Preston, Dueck, Stefanik, Tan, Trapnell, Kim, Waterston, and Murray]{Packer2019-xl}
Jonathan~S Packer, Qin Zhu, Chau Huynh, Priya Sivaramakrishnan, Elicia Preston, Hannah Dueck, Derek Stefanik, Kai Tan, Cole Trapnell, Junhyong Kim, Robert~H Waterston, and John~I Murray.
\newblock {A lineage-resolved molecular atlas of C. elegans embryogenesis at single-cell resolution}.
\newblock \emph{Science (New York, N.Y.)}, 365\penalty0 (6459):\penalty0 eaax1971, 20~September 2019.
\newblock ISSN 0036-8075,1095-9203.
\newblock \doi{10.1126/science.aax1971}.

\bibitem[Qiu et~al.(2022{\natexlab{a}})Qiu, Cao, Martin, Li, Welsh, Srivatsan, Huang, Calderon, Noble, Disteche, Murray, Spielmann, Moens, Trapnell, and Shendure]{Qiu2022-xl}
Chengxiang Qiu, Junyue Cao, Beth~K Martin, Tony Li, Ian~C Welsh, Sanjay Srivatsan, Xingfan Huang, Diego Calderon, William~Stafford Noble, Christine~M Disteche, Stephen~A Murray, Malte Spielmann, Cecilia~B Moens, Cole Trapnell, and Jay Shendure.
\newblock {Systematic reconstruction of cellular trajectories across mouse embryogenesis}.
\newblock \emph{Nature Genetics}, 54\penalty0 (3):\penalty0 328--341, 14~March 2022{\natexlab{a}}.
\newblock ISSN 1061-4036,1546-1718.
\newblock \doi{10.1038/s41588-022-01018-x}.

\bibitem[Qiu et~al.(2024)Qiu, Martin, Welsh, Daza, Le, Huang, Nichols, Taylor, Fulton, O'Day, Gomes, Ilcisin, Srivatsan, Deng, Disteche, Noble, Hamazaki, Moens, Kimelman, Cao, Schier, Spielmann, Murray, Trapnell, and Shendure]{Qiu2024-cq}
Chengxiang Qiu, Beth~K Martin, Ian~C Welsh, Riza~M Daza, Truc-Mai Le, Xingfan Huang, Eva~K Nichols, Megan~L Taylor, Olivia Fulton, Diana~R O'Day, Anne~Roshella Gomes, Saskia Ilcisin, Sanjay Srivatsan, Xinxian Deng, Christine~M Disteche, William~Stafford Noble, Nobuhiko Hamazaki, Cecilia~B Moens, David Kimelman, Junyue Cao, Alexander~F Schier, Malte Spielmann, Stephen~A Murray, Cole Trapnell, and Jay Shendure.
\newblock {A single-cell time-lapse of mouse prenatal development from gastrula to birth}.
\newblock \emph{Nature}, 626\penalty0 (8001):\penalty0 1084--1093, 14~February 2024.
\newblock ISSN 0028-0836,1476-4687.
\newblock \doi{10.1038/s41586-024-07069-w}.

\bibitem[Rashid et~al.(2017)Rashid, Kotton, and Bar-Joseph]{Rashid2017-wa}
Sabrina Rashid, Darrell~N Kotton, and Ziv Bar-Joseph.
\newblock {TASIC: determining branching models from time series single cell data}.
\newblock \emph{Bioinformatics (Oxford, England)}, 33\penalty0 (16):\penalty0 2504--2512, 15~August 2017.
\newblock ISSN 1367-4803,1367-4811.
\newblock \doi{10.1093/bioinformatics/btx173}.

\bibitem[Schiebinger et~al.(2019)Schiebinger, Shu, Tabaka, Cleary, Subramanian, Solomon, Gould, Liu, Lin, Berube, Lee, Chen, Brumbaugh, Rigollet, Hochedlinger, Jaenisch, Regev, and Lander]{Schiebinger2019-kc}
Geoffrey Schiebinger, Jian Shu, Marcin Tabaka, Brian Cleary, Vidya Subramanian, Aryeh Solomon, Joshua Gould, Siyan Liu, Stacie Lin, Peter Berube, Lia Lee, Jenny Chen, Justin Brumbaugh, Philippe Rigollet, Konrad Hochedlinger, Rudolf Jaenisch, Aviv Regev, and Eric~S Lander.
\newblock {Optimal-Transport Analysis of Single-Cell Gene Expression Identifies Developmental Trajectories in Reprogramming}.
\newblock \emph{Cell}, 176\penalty0 (4):\penalty0 928--943.e22, 7~February 2019.
\newblock ISSN 0092-8674,1097-4172.
\newblock \doi{10.1016/j.cell.2019.01.006}.

\bibitem[Tong et~al.(2020)Tong, Huang, Wolf, Van~Dijk, and Krishnaswamy]{Tong2020-ix}
Alexander Tong, Jessie Huang, Guy Wolf, David Van~Dijk, and Smita Krishnaswamy.
\newblock {TrajectoryNet: A Dynamic Optimal Transport Network for Modeling Cellular Dynamics}.
\newblock In \emph{{International Conference on Machine Learning}}, pages 9526--9536. PMLR, 21~November 2020.

\bibitem[Qiu et~al.(2022{\natexlab{b}})Qiu, Zhang, Martin-Rufino, Weng, Hosseinzadeh, Yang, Pogson, Hein, Hoi Joseph~Min, Wang, Grody, Shurtleff, Yuan, Xu, Ma, Replogle, Lander, Darmanis, Bahar, Sankaran, Xing, and Weissman]{Qiu2022-ms}
Xiaojie Qiu, Yan Zhang, Jorge~D Martin-Rufino, Chen Weng, Shayan Hosseinzadeh, Dian Yang, Angela~N Pogson, Marco~Y Hein, Kyung Hoi Joseph~Min, Li~Wang, Emanuelle~I Grody, Matthew~J Shurtleff, Ruoshi Yuan, Song Xu, Yian Ma, Joseph~M Replogle, Eric~S Lander, Spyros Darmanis, Ivet Bahar, Vijay~G Sankaran, Jianhua Xing, and Jonathan~S Weissman.
\newblock {Mapping transcriptomic vector fields of single cells}.
\newblock \emph{Cell}, 185\penalty0 (4):\penalty0 690--711.e45, 17~February 2022{\natexlab{b}}.
\newblock ISSN 0092-8674,1097-4172.
\newblock \doi{10.1016/j.cell.2021.12.045}.

\bibitem[Huguet et~al.(2022)Huguet, Magruder, Tong, Fasina, Kuchroo, Wolf, and Krishnaswamy]{Huguet2022-sm}
Guillaume Huguet, Daniel~Sumner Magruder, Alexander Tong, Oluwadamilola Fasina, Manik Kuchroo, Guy Wolf, and Smita Krishnaswamy.
\newblock {Manifold interpolating optimal-transport flows for trajectory inference}.
\newblock In S~Koyejo, S~Mohamed, A~Agarwal, D~Belgrave, K~Cho, and A~Oh, editors, \emph{{Advances in Neural Information Processing Systems 35}}, volume~35, pages 29705--29718, San Diego, California, USA, 2022. Neural Information Processing Systems Foundation, Inc. (NeurIPS).
\newblock \doi{10.52202/068431-2154}.

\bibitem[Tong et~al.(2023)Tong, Malkin, Fatras, Atanackovic, Zhang, Huguet, Wolf, and Bengio]{Tong2023-mf}
Alexander Tong, Nikolay Malkin, Kilian Fatras, Lazar Atanackovic, Yanlei Zhang, Guillaume Huguet, Guy Wolf, and Yoshua Bengio.
\newblock {Simulation-free Schr{\"{o}}dinger bridges via score and flow matching}.
\newblock \emph{arXiv [cs.LG]}, 7~July 2023.

\bibitem[Bunne et~al.(2024)Bunne, Schiebinger, Krause, Regev, and Cuturi]{Bunne2024-lq}
Charlotte Bunne, Geoffrey Schiebinger, Andreas Krause, Aviv Regev, and Marco Cuturi.
\newblock {Optimal transport for single-cell and spatial omics}.
\newblock \emph{Nature Reviews. Methods Primers}, 4\penalty0 (1):\penalty0 58, 14~August 2024.
\newblock ISSN 2662-8449,2662-8449.
\newblock \doi{10.1038/s43586-024-00334-2}.

\bibitem[Shen et~al.(2024)Shen, Berlinghieri, and Broderick]{Shen2024-ey}
Yunyi Shen, Renato Berlinghieri, and Tamara Broderick.
\newblock {Multi-marginal Schr{\"{o}}dinger Bridges with Iterative Reference Refinement}.
\newblock \emph{arXiv [stat.ML]}, 12~August 2024.
\newblock \doi{10.48550/arXiv.2408.06277}.

\bibitem[Weinreb et~al.(2018)Weinreb, Wolock, Tusi, Socolovsky, and Klein]{Weinreb2018-wy}
Caleb Weinreb, Samuel Wolock, Betsabeh~K Tusi, Merav Socolovsky, and Allon~M Klein.
\newblock {Fundamental limits on dynamic inference from single-cell snapshots}.
\newblock \emph{Proceedings of the National Academy of Sciences of the United States of America}, 115\penalty0 (10):\penalty0 E2467--E2476, 6~March 2018.
\newblock ISSN 0027-8424,1091-6490.
\newblock \doi{10.1073/pnas.1714723115}.

\bibitem[Weinreb et~al.(2020)Weinreb, Rodriguez-Fraticelli, Camargo, and Klein]{Weinreb2020-uk}
Caleb Weinreb, Alejo Rodriguez-Fraticelli, Fernando~D Camargo, and Allon~M Klein.
\newblock {Lineage tracing on transcriptional landscapes links state to fate during differentiation}.
\newblock \emph{Science}, 367\penalty0 (6479), 14~February 2020.
\newblock ISSN 0036-8075,1095-9203.
\newblock \doi{10.1126/science.aaw3381}.

\bibitem[Sulston et~al.(1983)Sulston, Schierenberg, White, and Thomson]{Sulston1983-dj}
J~E Sulston, E~Schierenberg, J~G White, and J~N Thomson.
\newblock {The embryonic cell lineage of the nematode Caenorhabditis elegans}.
\newblock \emph{Developmental Biology}, 100\penalty0 (1):\penalty0 64--119, November 1983.
\newblock ISSN 0012-1606,1095-564X.
\newblock \doi{10.1016/0012-1606(83)90201-4}.

\bibitem[Olivier et~al.(2010)Olivier, Luengo-Oroz, Duloquin, Faure, Savy, Veilleux, Solinas, Débarre, Bourgine, Santos, Peyriéras, and Beaurepaire]{olivier_cell_2010}
Nicolas Olivier, Miguel~A. Luengo-Oroz, Louise Duloquin, Emmanuel Faure, Thierry Savy, Israël Veilleux, Xavier Solinas, Delphine Débarre, Paul Bourgine, Andrés Santos, Nadine Peyriéras, and Emmanuel Beaurepaire.
\newblock Cell lineage reconstruction of early zebrafish embryos using label-free nonlinear microscopy.
\newblock \emph{Science}, 329\penalty0 (5994):\penalty0 967--971, August 2010.
\newblock ISSN 1095-9203.
\newblock \doi{10.1126/science.1189428}.

\bibitem[Kwak et~al.(2013)Kwak, Park, Jung, Hwang, Kwon, and Kee]{kwak_live_2013}
Jina Kwak, Ok~Kyu Park, Yoo~Jung Jung, Byung~Joon Hwang, Seung-Hae Kwon, and Yun Kee.
\newblock Live image profiling of neural crest lineages in zebrafish transgenic lines.
\newblock \emph{Molecules and Cells}, 35\penalty0 (3):\penalty0 255--260, March 2013.
\newblock ISSN 0219-1032.
\newblock \doi{10.1007/s10059-013-0001-5}.

\bibitem[Lawson et~al.(1991)Lawson, Meneses, and Pedersen]{lawson_clonal_1991}
Kirstie~A. Lawson, Juanito~J. Meneses, and Roger~A. Pedersen.
\newblock Clonal analysis of epiblast fate during germ layer formation in the mouse embryo.
\newblock \emph{Development}, 113\penalty0 (3):\penalty0 891--911, November 1991.
\newblock ISSN 0950-1991.
\newblock \doi{10.1242/dev.113.3.891}.
\newblock URL \url{https://doi.org/10.1242/dev.113.3.891}.

\bibitem[Kretzschmar and Watt(2012)]{kretzschmar_lineage_2012}
Kai Kretzschmar and Fiona~M. Watt.
\newblock Lineage tracing.
\newblock \emph{Cell}, 148\penalty0 (1-2):\penalty0 33--45, January 2012.
\newblock ISSN 1097-4172.
\newblock \doi{10.1016/j.cell.2012.01.002}.

\bibitem[Rodriguez-Fraticelli and Parreno(2026)]{rodriguez-fraticelli_charting_2026}
Alejo~E. Rodriguez-Fraticelli and Victoria Parreno.
\newblock Charting single-cell lineages with synthetic and natural barcodes.
\newblock \emph{Nature Reviews Genetics}, February 2026.
\newblock ISSN 1471-0056, 1471-0064.
\newblock \doi{10.1038/s41576-026-00943-5}.
\newblock URL \url{https://www.nature.com/articles/s41576-026-00943-5}.

\bibitem[Frieda et~al.(2017)Frieda, Linton, Hormoz, Choi, Chow, Singer, Budde, Elowitz, and Cai]{frieda_synthetic_2017}
Kirsten~L. Frieda, James~M. Linton, Sahand Hormoz, Joonhyuk Choi, Ke-Huan~K. Chow, Zakary~S. Singer, Mark~W. Budde, Michael~B. Elowitz, and Long Cai.
\newblock Synthetic recording and in situ readout of lineage information in single cells.
\newblock \emph{Nature}, 541\penalty0 (7635):\penalty0 107--111, January 2017.
\newblock ISSN 1476-4687.
\newblock \doi{10.1038/nature20777}.
\newblock URL \url{https://www.nature.com/articles/nature20777}.

\bibitem[McKenna et~al.(2016)McKenna, Findlay, Gagnon, Horwitz, Schier, and Shendure]{mckenna_whole-organism_2016}
Aaron McKenna, Gregory~M. Findlay, James~A. Gagnon, Marshall~S. Horwitz, Alexander~F. Schier, and Jay Shendure.
\newblock Whole-organism lineage tracing by combinatorial and cumulative genome editing.
\newblock \emph{Science}, 353\penalty0 (6298):\penalty0 aaf7907, July 2016.
\newblock \doi{10.1126/science.aaf7907}.
\newblock URL \url{https://www.science.org/doi/10.1126/science.aaf7907}.

\bibitem[Alemany et~al.(2018)Alemany, Florescu, Baron, Peterson-Maduro, and van Oudenaarden]{alemany_whole-organism_2018}
Anna Alemany, Maria Florescu, Chloé~S. Baron, Josi Peterson-Maduro, and Alexander van Oudenaarden.
\newblock Whole-organism clone tracing using single-cell sequencing.
\newblock \emph{Nature}, 556\penalty0 (7699):\penalty0 108--112, April 2018.
\newblock ISSN 1476-4687.
\newblock \doi{10.1038/nature25969}.
\newblock URL \url{https://www.nature.com/articles/nature25969}.

\bibitem[Spanjaard et~al.(2018)Spanjaard, Hu, Mitic, Olivares-Chauvet, Janjuha, Ninov, and Junker]{spanjaard_simultaneous_2018}
Bastiaan Spanjaard, Bo~Hu, Nina Mitic, Pedro Olivares-Chauvet, Sharan Janjuha, Nikolay Ninov, and Jan~Philipp Junker.
\newblock Simultaneous lineage tracing and cell-type identification using {CRISPR}-{Cas9}-induced genetic scars.
\newblock \emph{Nature Biotechnology}, 36\penalty0 (5):\penalty0 469--473, June 2018.
\newblock ISSN 1546-1696.
\newblock \doi{10.1038/nbt.4124}.

\bibitem[Raj et~al.(2018)Raj, Wagner, McKenna, Pandey, Klein, Shendure, Gagnon, and Schier]{raj_simultaneous_2018}
Bushra Raj, Daniel~E. Wagner, Aaron McKenna, Shristi Pandey, Allon~M. Klein, Jay Shendure, James~A. Gagnon, and Alexander~F. Schier.
\newblock Simultaneous single-cell profiling of lineages and cell types in the vertebrate brain.
\newblock \emph{Nature Biotechnology}, 36\penalty0 (5):\penalty0 442--450, June 2018.
\newblock ISSN 1546-1696.
\newblock \doi{10.1038/nbt.4103}.

\bibitem[Kalhor et~al.(2018)Kalhor, Kalhor, Mejia, Leeper, Graveline, Mali, and Church]{kalhor_developmental_2018}
Reza Kalhor, Kian Kalhor, Leo Mejia, Kathleen Leeper, Amanda Graveline, Prashant Mali, and George~M. Church.
\newblock Developmental barcoding of whole mouse via homing {CRISPR}.
\newblock \emph{Science}, 361\penalty0 (6405):\penalty0 eaat9804, August 2018.
\newblock \doi{10.1126/science.aat9804}.
\newblock URL \url{https://www.science.org/doi/10.1126/science.aat9804}.

\bibitem[Bowling et~al.(2020)Bowling, Sritharan, Osorio, Nguyen, Cheung, Rodriguez-Fraticelli, Patel, Yuan, Fujiwara, Li, Orkin, Hormoz, and Camargo]{bowling_engineered_2020}
Sarah Bowling, Duluxan Sritharan, Fernando~G. Osorio, Maximilian Nguyen, Priscilla Cheung, Alejo Rodriguez-Fraticelli, Sachin Patel, Wei-Chien Yuan, Yuko Fujiwara, Bin~E. Li, Stuart~H. Orkin, Sahand Hormoz, and Fernando~D. Camargo.
\newblock An {Engineered} {CRISPR}-{Cas9} {Mouse} {Line} for {Simultaneous} {Readout} of {Lineage} {Histories} and {Gene} {Expression} {Profiles} in {Single} {Cells}.
\newblock \emph{Cell}, 181\penalty0 (6):\penalty0 1410--1422.e27, June 2020.
\newblock ISSN 0092-8674.
\newblock \doi{10.1016/j.cell.2020.04.048}.
\newblock URL \url{https://www.sciencedirect.com/science/article/pii/S0092867420305547}.

\bibitem[Koblan et~al.(2025)Koblan, Yost, Zheng, Colgan, Jones, Yang, Kumar, Sandhu, Schnell, Sun, Ergen, Saunders, Zhuang, Allen, Yosef, and Weissman]{Koblan2025-ac}
Luke~W Koblan, Kathryn~E Yost, Pu~Zheng, William~N Colgan, Matthew~G Jones, Dian Yang, Arhan Kumar, Jaspreet Sandhu, Alexandra Schnell, Dawei Sun, Can Ergen, Reuben~A Saunders, Xiaowei Zhuang, William~E Allen, Nir Yosef, and Jonathan~S Weissman.
\newblock {High-resolution spatial mapping of cell state and lineage dynamics in vivo with {PEtracer}}.
\newblock \emph{Science (New York, N.Y.)}, 390\penalty0 (6770):\penalty0 eadx3800, 16~October 2025.
\newblock ISSN 0036-8075,1095-9203.
\newblock \doi{10.1126/science.adx3800}.

\bibitem[Choi et~al.(2022)Choi, Chen, Minkina, Chardon, Suiter, Regalado, Domcke, Hamazaki, Lee, Martin, Daza, and Shendure]{Choi2022-yv}
Junhong Choi, Wei Chen, Anna Minkina, Florence~M Chardon, Chase~C Suiter, Samuel~G Regalado, Silvia Domcke, Nobuhiko Hamazaki, Choli Lee, Beth Martin, Riza~M Daza, and Jay Shendure.
\newblock {A time-resolved, multi-symbol molecular recorder via sequential genome editing}.
\newblock \emph{Nature}, pages 1--10, 6~July 2022.
\newblock ISSN 0028-0836.
\newblock \doi{10.1038/s41586-022-04922-8}.

\bibitem[Yu et~al.(2026)Yu, Kim, Seidel, Acosta-Clark, Martin, O'Connor, Daza, Gasperini, Nathans, Lam, Gamo, Vijay~Kumar, Kuo, Lalanne, Simeonov, Pepper, Trapnell, Gray, Choi, Qiu, and Shendure]{Yu2026-bh}
Qi~Yu, Haedong Kim, Sophie Seidel, James~F Acosta-Clark, Beth~K Martin, Kyle O'Connor, Riza~M Daza, Molly Gasperini, Jenny~F Nathans, Maggie Lam, Elena Gamo, Shruthi Vijay~Kumar, Lauren Kuo, Jean-Beno\^{\i}t Lalanne, Kamen~P Simeonov, Marion Pepper, Cole Trapnell, Jesse~M Gray, Junhong Choi, Chengxiang Qiu, and Jay Shendure.
\newblock {A DNA Typewriter records the cell lineage history of a mouse, from zygote to late organogenesis}.
\newblock \emph{Science (New York, N.Y.)}, \penalty0 (eael0508):\penalty0 eael0508, 8~October 2026.
\newblock ISSN 0036-8075,1095-9203.
\newblock \doi{10.1126/science.ael0508}.

\bibitem[Colgan et~al.(2026)Colgan, Koblan, Villagrana, Hou, Wang, Gowri, Chandler, Sep{\'u}lveda, Ciftci, Smolyar, Young, Wittler, Markoulaki, Loh, Zhuang, Yosef, Smith, and Weissman]{colgan}
W.~N. Colgan, Luke~W. Koblan, JoAnne Villagrana, Tien-Chi~Jason Hou, Minming Wang, Gokul Gowri, Whitney Chandler, Leonardo~A. Sep{\'u}lveda, Didar Ciftci, Karina Smolyar, Alicia Young, Lars Wittler, Styliani Markoulaki, Kyle~M. Loh, Xiaowei Zhuang, Nir Yosef, Zachary~D. Smith, and Jonathan~S. Weissman.
\newblock {Comprehensive Lineage Tracing Maps the Landscape of Cell Fate Decisions in Mouse Embryogenesis}.
\newblock \emph{bioRxiv}, 2026.
\newblock \doi{10.64898/2026.05.07.722278}.
\newblock URL \url{https://www.biorxiv.org/content/10.64898/2026.05.07.722278v1}.

\bibitem[Tarkowski and Wr\'{o}blewska(1967)]{Tarkowski1967-bn}
Andrzej~K Tarkowski and Joanna Wr\'{o}blewska.
\newblock {Development of blastomeres of mouse eggs isolated at the 4- and 8-cell stage}.
\newblock \emph{Development (Cambridge, England)}, 18\penalty0 (1):\penalty0 155--180, 1~August 1967.
\newblock ISSN 0950-1991,1477-9129.
\newblock \doi{10.1242/dev.18.1.155}.

\bibitem[Gardner and Cockroft(1998)]{gardner_complete_1998}
R.~L. Gardner and D.~L. Cockroft.
\newblock Complete dissipation of coherent clonal growth occurs before gastrulation in mouse epiblast.
\newblock \emph{Development}, 125\penalty0 (13):\penalty0 2397--2402, July 1998.
\newblock ISSN 0950-1991.
\newblock \doi{10.1242/dev.125.13.2397}.

\bibitem[Tarkowski(1959)]{Tarkowski1959-ve}
A~K Tarkowski.
\newblock {Experiments on the development of isolated blastomeres of mouse eggs}.
\newblock \emph{Nature}, 184\penalty0 (4695):\penalty0 1286--1287, 24~October 1959.
\newblock ISSN 0028-0836,1476-4687.
\newblock \doi{10.1038/1841286a0}.

\bibitem[Jindal et~al.(2024)Jindal, Adil, Yamaguchi, Yang, Wang, Kamimoto, Rivera-Gonzalez, and Morris]{Jindal2024-ik}
Kunal Jindal, Mohd~Tayyab Adil, Naoto Yamaguchi, Xue Yang, Helen~C Wang, Kenji Kamimoto, Guillermo~C Rivera-Gonzalez, and Samantha~A Morris.
\newblock {Single-cell lineage capture across genomic modalities with CellTag-multi reveals fate-specific gene regulatory changes}.
\newblock \emph{Nature Biotechnology}, 42\penalty0 (6):\penalty0 946--959, June 2024.
\newblock ISSN 1087-0156,1546-1696.
\newblock \doi{10.1038/s41587-023-01931-4}.

\bibitem[Forrow and Schiebinger(2021)]{Forrow2021-vs}
Aden Forrow and Geoffrey Schiebinger.
\newblock {LineageOT is a unified framework for lineage tracing and trajectory inference}.
\newblock \emph{Nature Communications}, 12\penalty0 (1):\penalty0 4940, 16~August 2021.
\newblock ISSN 2041-1723,2041-1723.
\newblock \doi{10.1038/s41467-021-25133-1}.

\bibitem[Lange et~al.(2024)Lange, Piran, Klein, Spanjaard, Klein, Junker, Theis, and Nitzan]{Lange2024-ts}
Marius Lange, Zoe Piran, Michal Klein, Bastiaan Spanjaard, Dominik Klein, Jan~Philipp Junker, Fabian~J Theis, and Mor Nitzan.
\newblock {Mapping lineage-traced cells across time points with moslin}.
\newblock \emph{Genome Biology}, 25\penalty0 (1):\penalty0 277, 21~October 2024.
\newblock ISSN 1474-7596,1474-760X.
\newblock \doi{10.1186/s13059-024-03422-4}.

\bibitem[Klein et~al.(2024)Klein, Uscidda, Theis, and Cuturi]{Klein2024-mk}
Dominik Klein, Th\'{e}o Uscidda, Fabian Theis, and Marco Cuturi.
\newblock {GENOT: Entropic (Gromov) Wasserstein Flow Matching with Applications to Single-Cell Genomics}.
\newblock In A~Globerson, L~Mackey, D~Belgrave, A~Fan, U~Paquet, J~Tomczak, and C~Zhang, editors, \emph{{Advances in Neural Information Processing Systems 37}}, volume~37, pages 103897--103944, San Diego, California, USA, 2024. Neural Information Processing Systems Foundation, Inc. (NeurIPS).
\newblock \doi{10.52202/079017-3301}.

\bibitem[Ventre et~al.(2026)Ventre, Forrow, Gadhiwala, Chakraborty, Angel, and Schiebinger]{Ventre2026-bg}
Elias Ventre, Aden Forrow, Nitya Gadhiwala, Parijat Chakraborty, Omer Angel, and Geoffrey Schiebinger.
\newblock {Trajectory inference for a branching SDE model of cell differentiation via lineage tracing}.
\newblock \emph{Journal of Mathematical Biology}, 92\penalty0 (5), 17~April 2026.
\newblock ISSN 0303-6812,1432-1416.
\newblock \doi{10.1007/s00285-026-02375-5}.

\bibitem[Vandergrift et~al.(2026)Vandergrift, White, Polyanskiy, Rigollet, and Atanackovic]{Vandergrift2026-aa}
Matthew Vandergrift, Martha White, Yury Polyanskiy, Philippe Rigollet, and Lazar Atanackovic.
\newblock {Measure-to-measure regression with transformers}.
\newblock \emph{arXiv [cs.LG]}, 27~May 2026.
\newblock \doi{10.48550/arXiv.2605.28075}.

\bibitem[Fishman et~al.(2025)Fishman, Gowri, Yin, Gootenberg, and Abudayyeh]{Fishman2025-gq}
Nic Fishman, Gokul Gowri, Peng Yin, Jonathan~S Gootenberg, and Omar Abudayyeh.
\newblock {Generative Distribution Embeddings: Lifting autoencoders to the space of distributions for multiscale representation learning}.
\newblock In \emph{{The Thirty-ninth Annual Conference on Neural Information Processing Systems}}, 29~October 2025.

\bibitem[Ho et~al.(2020)Ho, Jain, and Abbeel]{Ho2020}
Jonathan Ho, Ajay Jain, and Pieter Abbeel.
\newblock {Denoising Diffusion Probabilistic Models}.
\newblock In \emph{Advances in Neural Information Processing Systems}, 2020.
\newblock URL \url{https://proceedings.neurips.cc/paper/2020/hash/4c5bcfec8584af0d967f1ab10179ca4b-Abstract.html}.

\bibitem[Lipman et~al.(2022)Lipman, Chen, Ben-Hamu, Nickel, and Le]{Lipman2022-qc}
Yaron Lipman, Ricky T~Q Chen, Heli Ben-Hamu, Maximilian Nickel, and Matt Le.
\newblock {Flow Matching for generative modeling}.
\newblock \emph{arXiv [cs.LG]}, 6~October 2022.
\newblock \doi{10.48550/arXiv.2210.02747}.

\bibitem[Albergo and Vanden-Eijnden(2023)]{Albergo2023}
Michael~S Albergo and Eric Vanden-Eijnden.
\newblock {Building Normalizing Flows with Stochastic Interpolants}.
\newblock In \emph{International Conference on Learning Representations}, 2023.
\newblock URL \url{https://arxiv.org/abs/2209.15571}.

\bibitem[Jones et~al.(2020)Jones, Khodaverdian, Quinn, Chan, Hussmann, Wang, Xu, Weissman, and Yosef]{Jones2020-jg}
Matthew~G Jones, Alex Khodaverdian, Jeffrey~J Quinn, Michelle~M Chan, Jeffrey~A Hussmann, Robert Wang, Chenling Xu, Jonathan~S Weissman, and Nir Yosef.
\newblock {Inference of single-cell phylogenies from lineage tracing data using Cassiopeia}.
\newblock \emph{Genome biology}, 21\penalty0 (1):\penalty0 92, 14~April 2020.
\newblock ISSN 1474-7596,1474-760X.
\newblock \doi{10.1186/s13059-020-02000-8}.

\bibitem[Prillo et~al.(2025)Prillo, Ravoor, Yosef, and Song]{Prillo2025-ik}
Sebastian Prillo, Akshay Ravoor, Nir Yosef, and Yun~S Song.
\newblock {ConvexML: Fast and accurate branch length estimation under irreversible mutation models, illustrated through applications to CRISPR/Cas9-based lineage tracing}.
\newblock \emph{Systematic biology}, \penalty0 (syaf054):\penalty0 syaf054, 8~August 2025.
\newblock ISSN 1063-5157,1076-836X.
\newblock \doi{10.1093/sysbio/syaf054}.

\bibitem[Felsenstein(1985)]{Felsenstein1985-ih}
Joseph Felsenstein.
\newblock {Phylogenies and the Comparative Method}.
\newblock \emph{The American Naturalist}, 125\penalty0 (1):\penalty0 1--15, January 1985.
\newblock ISSN 0003-0147,1537-5323.
\newblock \doi{10.1086/284325}.

\bibitem[Pagel(1999)]{Pagel1999-go}
M~Pagel.
\newblock {Inferring the historical patterns of biological evolution}.
\newblock \emph{Nature}, 401\penalty0 (6756):\penalty0 877--884, 28~October 1999.
\newblock ISSN 0028-0836,1476-4687.
\newblock \doi{10.1038/44766}.

\bibitem[Harmon(2018)]{Harmon2018-kk}
Luke~J Harmon.
\newblock \emph{{Phylogenetic comparative methods: Learning from trees}}.
\newblock Createspace Independent Publishing Platform, North Charleston, SC, 23~May 2018.
\newblock ISBN 9781719584463.

\bibitem[Xu et~al.(2026)Xu, Narasimhan, Gorstein, Perea, Shen, and Sol\'{\i}s-Lemus]{Xu2026-uy}
Tianyi Xu, Shrinaath Narasimhan, Evan Gorstein, Santiago Perea, Yunyi Shen, and Claudia Sol\'{\i}s-Lemus.
\newblock {Generating the unheard: Phylogeny-guided latent generation for ancestral sound reconstruction}.
\newblock \emph{arXiv [cs.SD]}, 14~September 2026.
\newblock \doi{10.48550/arXiv.2609.15240}.

\bibitem[Teo et~al.(2025)Teo, Bastide, and An{\'e}]{Teo2025-graphicalmodels}
Benjamin Teo, Paul Bastide, and C{\'e}cile An{\'e}.
\newblock Leveraging graphical model techniques to study evolution on phylogenetic networks.
\newblock \emph{Philosophical Transactions of the Royal Society B: Biological Sciences}, 380\penalty0 (1919):\penalty0 20230310, 2025.
\newblock \doi{10.1098/rstb.2023.0310}.

\bibitem[Vaswani et~al.(2017)Vaswani, Shazeer, Parmar, Uszkoreit, Jones, Gomez, Kaiser, and Polosukhin]{Vaswani2017}
Ashish Vaswani, Noam Shazeer, Niki Parmar, Jakob Uszkoreit, Llion Jones, Aidan~N Gomez, {\L}ukasz Kaiser, and Illia Polosukhin.
\newblock {Attention Is All You Need}.
\newblock In \emph{Advances in Neural Information Processing Systems}, 2017.
\newblock URL \url{https://proceedings.neurips.cc/paper/2017/hash/3f5ee243547dee91fbd053c1c4a845aa-Abstract.html}.

\bibitem[Turner(2023)]{Turner2023-df}
Richard~E Turner.
\newblock {An Introduction to Transformers}.
\newblock \emph{arXiv [cs.LG]}, 20~April 2023.
\newblock \doi{10.48550/arXiv.2304.10557}.

\bibitem[Press et~al.(2022)Press, Smith, and Lewis]{Press2022}
Ofir Press, Noah~A Smith, and Mike Lewis.
\newblock {Train Short, Test Long: Attention with Linear Biases Enables Input Length Extrapolation}.
\newblock In \emph{International Conference on Learning Representations}, 2022.
\newblock URL \url{https://arxiv.org/abs/2108.12409}.

\bibitem[Lugmayr et~al.(2022)Lugmayr, Danelljan, Romero, Yu, Timofte, and Van~Gool]{Lugmayr2022}
Andreas Lugmayr, Martin Danelljan, Andr{\'e}s Romero, Fisher Yu, Radu Timofte, and Luc Van~Gool.
\newblock {RePaint: Inpainting using Denoising Diffusion Probabilistic Models}.
\newblock In \emph{IEEE/CVF Conference on Computer Vision and Pattern Recognition}, 2022.
\newblock URL \url{https://openaccess.thecvf.com/content/CVPR2022/html/Lugmayr_RePaint_Inpainting_Using_Denoising_Diffusion_Probabilistic_Models_CVPR_2022_paper.html}.

\bibitem[Albergo et~al.(2024)Albergo, Goldstein, Boffi, Ranganath, and Vanden-Eijnden]{Albergo2024-pq}
Michael~Samuel Albergo, Mark Goldstein, Nicholas~Matthew Boffi, Rajesh Ranganath, and Eric Vanden-Eijnden.
\newblock {Stochastic Interpolants with Data-Dependent Couplings}.
\newblock In \emph{{International Conference on Machine Learning}}, pages 921--937. PMLR, 8~July 2024.

\bibitem[Tong et~al.(2026)Tong, Li, Jiang, Yuan, Liu, He, Liang, Han, and Guo]{Tong2026-og}
Dong Tong, Anqi Li, Quanquan Jiang, Qili Yuan, Xiaozhao Liu, Ximiao He, Jiejunyi Liang, Yunyun Han, and Zheng Guo.
\newblock {Intestinal stem cells count self-renewal divisions to switch multipotency}.
\newblock \emph{Nature}, pages 1--10, 29~July 2026.
\newblock ISSN 0028-0836,1476-4687.
\newblock \doi{10.1038/s41586-026-10814-y}.

\bibitem[Tran et~al.(2024)Tran, Askary, and Elowitz]{Tran2024-eu}
Martin Tran, Amjad Askary, and Michael~B Elowitz.
\newblock {Lineage motifs as developmental modules for control of cell type proportions}.
\newblock \emph{Developmental Cell}, 59\penalty0 (6):\penalty0 812--826.e3, 25~March 2024.
\newblock ISSN 1534-5807,1878-1551.
\newblock \doi{10.1016/j.devcel.2024.01.017}.

\bibitem[Lopez et~al.(2018)Lopez, Regier, Cole, Jordan, and Yosef]{Lopez2018}
Romain Lopez, Jeffrey Regier, Michael~B Cole, Michael~I Jordan, and Nir Yosef.
\newblock {Deep generative modeling for single-cell transcriptomics}.
\newblock \emph{Nature Methods}, 15\penalty0 (12):\penalty0 1053--1058, 2018.
\newblock \doi{10.1038/s41592-018-0229-2}.

\bibitem[Stone and Stainier(2019)]{Stone2019-hb}
Oliver~A Stone and Didier Y~R Stainier.
\newblock {Paraxial mesoderm is the major source of lymphatic endothelium}.
\newblock \emph{Developmental Cell}, 50\penalty0 (2):\penalty0 247--255.e3, 22~July 2019.
\newblock ISSN 1534-5807,1878-1551.
\newblock \doi{10.1016/j.devcel.2019.04.034}.

\bibitem[Pardanaud et~al.(1996)Pardanaud, Luton, Prigent, Bourcheix, Catala, and Dieterlen-Lievre]{pardanaud_two_1996}
L.~Pardanaud, D.~Luton, M.~Prigent, L.~M. Bourcheix, M.~Catala, and F.~Dieterlen-Lievre.
\newblock Two distinct endothelial lineages in ontogeny, one of them related to hemopoiesis.
\newblock \emph{Development}, 122\penalty0 (5):\penalty0 1363--1371, May 1996.
\newblock ISSN 0950-1991.
\newblock \doi{10.1242/dev.122.5.1363}.

\bibitem[Couly et~al.(1995)Couly, Coltey, Eichmann, and Le~Douarin]{couly_angiogenic_1995}
G.~Couly, P.~Coltey, A.~Eichmann, and N.~M. Le~Douarin.
\newblock The angiogenic potentials of the cephalic mesoderm and the origin of brain and head blood vessels.
\newblock \emph{Mechanisms of Development}, 53\penalty0 (1):\penalty0 97--112, September 1995.
\newblock ISSN 0925-4773.
\newblock \doi{10.1016/0925-4773(95)00428-9}.

\bibitem[Sahai-Hernandez et~al.(2023)Sahai-Hernandez, Pouget, Eyal, Svoboda, Chacon, Grimm, Gj\o{}en, and Traver]{Sahai-Hernandez2023-uy}
Pankaj Sahai-Hernandez, Claire Pouget, Shai Eyal, Ondrej Svoboda, Jose Chacon, Lin Grimm, Tor Gj\o{}en, and David Traver.
\newblock {Dermomyotome-derived endothelial cells migrate to the dorsal aorta to support hematopoietic stem cell emergence}.
\newblock \emph{eLife}, 12\penalty0 (e58300):\penalty0 e58300, 11~September 2023.
\newblock ISSN 2050-084X.
\newblock \doi{10.7554/eLife.58300}.

\bibitem[Alexander et~al.(2009)Alexander, Nolte, and Krumlauf]{Alexander2009-hox}
Tara Alexander, Christof Nolte, and Robb Krumlauf.
\newblock {Hox Genes and Segmentation of the Hindbrain and Axial Skeleton}.
\newblock \emph{Annual Review of Cell and Developmental Biology}, 25:\penalty0 431--456, 2009.
\newblock \doi{10.1146/annurev.cellbio.042308.113423}.

\bibitem[Fang et~al.(2022)Fang, Bell, Sapirstein, Asami, Leeper, Zack, Ji, and Kalhor]{Fang2022-tg}
Weixiang Fang, Claire~M Bell, Abel Sapirstein, Soichiro Asami, Kathleen Leeper, Donald~J Zack, Hongkai Ji, and Reza Kalhor.
\newblock {Quantitative fate mapping: A general framework for analyzing progenitor state dynamics via retrospective lineage barcoding}.
\newblock \emph{Cell}, 185\penalty0 (24):\penalty0 4604--4620.e32, 23~November 2022.
\newblock ISSN 0092-8674,1097-4172.
\newblock \doi{10.1016/j.cell.2022.10.028}.

\bibitem[Lewis and Wolpert(1976)]{Lewis1976-zo}
J~H Lewis and L~Wolpert.
\newblock {The principle of non-equivalence in development}.
\newblock \emph{Journal of Theoretical Biology}, 62\penalty0 (2):\penalty0 479--490, 21~October 1976.
\newblock ISSN 0022-5193,1095-8541.
\newblock \doi{10.1016/0022-5193(76)90132-6}.

\bibitem[van~de Laar et~al.(2016)van~de Laar, Saelens, De~Prijck, Martens, Scott, Van~Isterdael, Hoffmann, Beyaert, Saeys, Lambrecht, and Guilliams]{van-de-Laar2016-yd}
Lianne van~de Laar, Wouter Saelens, Sofie De~Prijck, Liesbet Martens, Charlotte~L Scott, Gert Van~Isterdael, Eik Hoffmann, Rudi Beyaert, Yvan Saeys, Bart~N Lambrecht, and Martin Guilliams.
\newblock {Yolk sac macrophages, fetal liver, and adult monocytes can colonize an empty niche and develop into functional tissue-resident macrophages}.
\newblock \emph{Immunity}, 44\penalty0 (4):\penalty0 755--768, 19~April 2016.
\newblock ISSN 1074-7613,1097-4180.
\newblock \doi{10.1016/j.immuni.2016.02.017}.

\bibitem[Rinn et~al.(2008)Rinn, Wang, Allen, Brugmann, Mikels, Liu, Ridky, Stadler, Nusse, Helms, and Chang]{Rinn2008-zv}
John~L Rinn, Jordon~K Wang, Nancy Allen, Samantha~A Brugmann, Amanda~J Mikels, Helen Liu, Todd~W Ridky, H~Scott Stadler, Roel Nusse, Jill~A Helms, and Howard~Y Chang.
\newblock {A dermal HOX transcriptional program regulates site-specific epidermal fate}.
\newblock \emph{Genes \& Development}, 22\penalty0 (3):\penalty0 303--307, 1~February 2008.
\newblock ISSN 0890-9369,1549-5477.
\newblock \doi{10.1101/gad.1610508}.

\bibitem[Metzis et~al.(2018)Metzis, Steinhauser, Pakanavicius, Gouti, Stamataki, Ivanovitch, Watson, Rayon, Mousavy~Gharavy, Lovell-Badge, Luscombe, and Briscoe]{metzis_nervous_2018}
Vicki Metzis, Sebastian Steinhauser, Edvinas Pakanavicius, Mina Gouti, Despina Stamataki, Kenzo Ivanovitch, Thomas Watson, Teresa Rayon, S.~Neda Mousavy~Gharavy, Robin Lovell-Badge, Nicholas~M. Luscombe, and James Briscoe.
\newblock Nervous {System} {Regionalization} {Entails} {Axial} {Allocation} before {Neural} {Differentiation}.
\newblock \emph{Cell}, 175\penalty0 (4):\penalty0 1105--1118.e17, November 2018.
\newblock ISSN 1097-4172.
\newblock \doi{10.1016/j.cell.2018.09.040}.

\bibitem[Jessell(2000)]{Jessell2000-vr}
T~M Jessell.
\newblock {Neuronal specification in the spinal cord: inductive signals and transcriptional codes}.
\newblock \emph{Nature Reviews. Genetics}, 1\penalty0 (1):\penalty0 20--29, October 2000.
\newblock ISSN 1471-0056,1471-0064.
\newblock \doi{10.1038/35049541}.

\bibitem[Forlani et~al.(2003)Forlani, Lawson, and Deschamps]{forlani_acquisition_2003}
Sylvie Forlani, Kirstie~A. Lawson, and Jacqueline Deschamps.
\newblock Acquisition of {Hox} codes during gastrulation and axial elongation in the mouse embryo.
\newblock \emph{Development}, 130\penalty0 (16):\penalty0 3807--3819, August 2003.
\newblock ISSN 0950-1991.
\newblock \doi{10.1242/dev.00573}.

\bibitem[Wellik(2007)]{wellik_hox_2007}
Deneen~M. Wellik.
\newblock Hox patterning of the vertebrate axial skeleton.
\newblock \emph{Developmental Dynamics: An Official Publication of the American Association of Anatomists}, 236\penalty0 (9):\penalty0 2454--2463, September 2007.
\newblock ISSN 1058-8388.
\newblock \doi{10.1002/dvdy.21286}.

\bibitem[Sierant et~al.(2025)Sierant, Jin, Bilguvar, Morton, Dong, Jiang, Lu, Li, L\'{o}pez-Gir\'{a}ldez, Tikhonova, Zeng, Lu, Choi, Zhang, Nelson-Williams, Knight, Zhao, Cao, Mane, Sedore, Gruber, Lek, Goldmuntz, Deanfield, Giardini, Mital, Russell, Gaynor, King, Wagner, Srivastava, Shen, Bernstein, Porter, Newburger, Seidman, Roberts, Yandell, Yost, Tristani-Firouzi, Kim, Chung, Gelb, Seidman, Brueckner, and Lifton]{Sierant2025-lx}
Michael~C Sierant, Sheng~Chih Jin, Kaya Bilguvar, Sarah~U Morton, Weilai Dong, Wei Jiang, Ziyu Lu, Boyang Li, Francesc L\'{o}pez-Gir\'{a}ldez, Irina Tikhonova, Xue Zeng, Qiongshi Lu, Jungmin Choi, Junhui Zhang, Carol Nelson-Williams, James~R Knight, Hongyu Zhao, Junyue Cao, Shrikant Mane, Stanley~C Sedore, Peter~J Gruber, Monkol Lek, Elizabeth Goldmuntz, John Deanfield, Alessandro Giardini, Seema Mital, Mark Russell, J~William Gaynor, Eileen King, Michael Wagner, Deepak Srivastava, Yufeng Shen, Daniel Bernstein, George~A Porter, Jr, Jane~W Newburger, Jonathan~G Seidman, Amy~E Roberts, Mark Yandell, H~Joseph Yost, Martin Tristani-Firouzi, Richard Kim, Wendy~K Chung, Bruce~D Gelb, Christine~E Seidman, Martina Brueckner, and Richard~P Lifton.
\newblock {Genomic analysis of 11,555 probands identifies 60 dominant congenital heart disease genes}.
\newblock \emph{Proceedings of the National Academy of Sciences of the United States of America}, 122\penalty0 (13):\penalty0 e2420343122, 1~April 2025.
\newblock ISSN 0027-8424,1091-6490.
\newblock \doi{10.1073/pnas.2420343122}.

\bibitem[Fijabi et~al.(2026)Fijabi, Teague, Freeburne, Khan, Johnson, Br{\"{u}}ckner, and Heemskerk]{Fijabi2026-sq}
Abdul-Bassit Fijabi, Seth Teague, Emily Freeburne, Hina~Aftab Khan, Craig Johnson, David~B Br{\"{u}}ckner, and Idse Heemskerk.
\newblock {Interpretable decoding of cell fate from a snapshot of combinatorial signaling}.
\newblock \emph{bioRxiv.org: The Preprint Server for Biology}, page 2026.05. 17.725652, 18~May 2026.
\newblock \doi{10.64898/2026.05.17.725652}.

\bibitem[Britsch et~al.(2001)Britsch, Goerich, Riethmacher, Peirano, Rossner, Nave, Birchmeier, and Wegner]{Britsch2001}
Stefan Britsch, Derk~E. Goerich, Dieter Riethmacher, Reto~I. Peirano, Moritz Rossner, Klaus-Armin Nave, Carmen Birchmeier, and Michael Wegner.
\newblock The transcription factor {Sox10} is a key regulator of peripheral glial development.
\newblock \emph{Genes \& Development}, 15\penalty0 (1):\penalty0 66--78, 2001.
\newblock \doi{10.1101/gad.186601}.

\bibitem[Wilm et~al.(2004)Wilm, James, Schultheiss, and Hogan]{Wilm2004}
Bettina Wilm, Richard~G. James, Thomas~M. Schultheiss, and Brigid L.~M. Hogan.
\newblock The forkhead genes, {Foxc1} and {Foxc2}, regulate paraxial versus intermediate mesoderm cell fate.
\newblock \emph{Developmental Biology}, 271\penalty0 (1):\penalty0 176--189, 2004.
\newblock \doi{10.1016/j.ydbio.2004.03.034}.

\bibitem[Lee et~al.(2005)Lee, Friedman, Fulmer, and Kaestner]{Lee2005}
Catherine~S. Lee, Joshua~R. Friedman, James~T. Fulmer, and Klaus~H. Kaestner.
\newblock The initiation of liver development is dependent on {Foxa} transcription factors.
\newblock \emph{Nature}, 435\penalty0 (7044):\penalty0 944--947, 2005.
\newblock \doi{10.1038/nature03649}.

\bibitem[Eppig et~al.(2015)Eppig, Blake, Bult, Kadin, Richardson, and {Mouse Genome Database Group}]{Eppig2015-bc}
Janan~T Eppig, Judith~A Blake, Carol~J Bult, James~A Kadin, Joel~E Richardson, and {Mouse Genome Database Group}.
\newblock {The Mouse Genome Database (MGD): facilitating mouse as a model for human biology and disease}.
\newblock \emph{Nucleic Acids Research}, 43\penalty0 (Database issue):\penalty0 D726--36, January 2015.
\newblock ISSN 1362-4962,0305-1048.
\newblock \doi{10.1093/nar/gku967}.

\bibitem[Lek et~al.(2016)Lek, Karczewski, Minikel, Samocha, Banks, Fennell, O'Donnell-Luria, Ware, Hill, Cummings, Tukiainen, Birnbaum, Kosmicki, Duncan, Estrada, Zhao, Zou, Pierce-Hoffman, Berghout, Cooper, Deflaux, DePristo, Do, Flannick, Fromer, Gauthier, Goldstein, Gupta, Howrigan, Kiezun, Kurki, Moonshine, Natarajan, Orozco, Peloso, Poplin, Rivas, Ruano-Rubio, Rose, Ruderfer, Shakir, Stenson, Stevens, Thomas, Tiao, Tusie-Luna, Weisburd, Won, Yu, Altshuler, Ardissino, Boehnke, Danesh, Donnelly, Elosua, Florez, Gabriel, Getz, Glatt, Hultman, Kathiresan, Laakso, McCarroll, McCarthy, McGovern, McPherson, Neale, Palotie, Purcell, Saleheen, Scharf, Sklar, Sullivan, Tuomilehto, Tsuang, Watkins, Wilson, Daly, MacArthur, and {Exome Aggregation Consortium}]{Lek2016-wm}
Monkol Lek, Konrad~J Karczewski, Eric~V Minikel, Kaitlin~E Samocha, Eric Banks, Timothy Fennell, Anne~H O'Donnell-Luria, James~S Ware, Andrew~J Hill, Beryl~B Cummings, Taru Tukiainen, Daniel~P Birnbaum, Jack~A Kosmicki, Laramie~E Duncan, Karol Estrada, Fengmei Zhao, James Zou, Emma Pierce-Hoffman, Joanne Berghout, David~N Cooper, Nicole Deflaux, Mark DePristo, Ron Do, Jason Flannick, Menachem Fromer, Laura Gauthier, Jackie Goldstein, Namrata Gupta, Daniel Howrigan, Adam Kiezun, Mitja~I Kurki, Ami~Levy Moonshine, Pradeep Natarajan, Lorena Orozco, Gina~M Peloso, Ryan Poplin, Manuel~A Rivas, Valentin Ruano-Rubio, Samuel~A Rose, Douglas~M Ruderfer, Khalid Shakir, Peter~D Stenson, Christine Stevens, Brett~P Thomas, Grace Tiao, Maria~T Tusie-Luna, Ben Weisburd, Hong-Hee Won, Dongmei Yu, David~M Altshuler, Diego Ardissino, Michael Boehnke, John Danesh, Stacey Donnelly, Roberto Elosua, Jose~C Florez, Stacey~B Gabriel, Gad Getz, Stephen~J Glatt, Christina~M Hultman, Sekar Kathiresan, Markku Laakso, Steven McCarroll,
  Mark~I McCarthy, Dermot McGovern, Ruth McPherson, Benjamin~M Neale, Aarno Palotie, Shaun~M Purcell, Danish Saleheen, Jeremiah~M Scharf, Pamela Sklar, Patrick~F Sullivan, Jaakko Tuomilehto, Ming~T Tsuang, Hugh~C Watkins, James~G Wilson, Mark~J Daly, Daniel~G MacArthur, and {Exome Aggregation Consortium}.
\newblock {Analysis of protein-coding genetic variation in 60,706 humans}.
\newblock \emph{Nature}, 536\penalty0 (7616):\penalty0 285--291, 18~August 2016.
\newblock ISSN 0028-0836,1476-4687.
\newblock \doi{10.1038/nature19057}.

\bibitem[Karczewski et~al.(2020)Karczewski, Francioli, Tiao, Cummings, Alf{\"{o}}ldi, Wang, Collins, Laricchia, Ganna, Birnbaum, Gauthier, Brand, Solomonson, Watts, Rhodes, Singer-Berk, England, Seaby, Kosmicki, Walters, Tashman, Farjoun, Banks, Poterba, Wang, Seed, Whiffin, Chong, Samocha, Pierce-Hoffman, Zappala, O'Donnell-Luria, Minikel, Weisburd, Lek, Ware, Vittal, Armean, Bergelson, Cibulskis, Connolly, Covarrubias, Donnelly, Ferriera, Gabriel, Gentry, Gupta, Jeandet, Kaplan, Llanwarne, Munshi, Novod, Petrillo, Roazen, Ruano-Rubio, Saltzman, Schleicher, Soto, Tibbetts, Tolonen, Wade, Talkowski, {Genome Aggregation Database Consortium}, Neale, Daly, and MacArthur]{Karczewski2020-ki}
Konrad~J Karczewski, Laurent~C Francioli, Grace Tiao, Beryl~B Cummings, Jessica Alf{\"{o}}ldi, Qingbo Wang, Ryan~L Collins, Kristen~M Laricchia, Andrea Ganna, Daniel~P Birnbaum, Laura~D Gauthier, Harrison Brand, Matthew Solomonson, Nicholas~A Watts, Daniel Rhodes, Moriel Singer-Berk, Eleina~M England, Eleanor~G Seaby, Jack~A Kosmicki, Raymond~K Walters, Katherine Tashman, Yossi Farjoun, Eric Banks, Timothy Poterba, Arcturus Wang, Cotton Seed, Nicola Whiffin, Jessica~X Chong, Kaitlin~E Samocha, Emma Pierce-Hoffman, Zachary Zappala, Anne~H O'Donnell-Luria, Eric~Vallabh Minikel, Ben Weisburd, Monkol Lek, James~S Ware, Christopher Vittal, Irina~M Armean, Louis Bergelson, Kristian Cibulskis, Kristen~M Connolly, Miguel Covarrubias, Stacey Donnelly, Steven Ferriera, Stacey Gabriel, Jeff Gentry, Namrata Gupta, Thibault Jeandet, Diane Kaplan, Christopher Llanwarne, Ruchi Munshi, Sam Novod, Nikelle Petrillo, David Roazen, Valentin Ruano-Rubio, Andrea Saltzman, Molly Schleicher, Jose Soto, Kathleen Tibbetts, Charlotte
  Tolonen, Gordon Wade, Michael~E Talkowski, {Genome Aggregation Database Consortium}, Benjamin~M Neale, Mark~J Daly, and Daniel~G MacArthur.
\newblock {The mutational constraint spectrum quantified from variation in 141,456 humans}.
\newblock \emph{Nature}, 581\penalty0 (7809):\penalty0 434--443, 27~May 2020.
\newblock ISSN 0028-0836,1476-4687.
\newblock \doi{10.1038/s41586-020-2308-7}.

\bibitem[Ivanovitch et~al.(2021)Ivanovitch, Soro-Barrio, Chakravarty, Jones, Bell, Mousavy~Gharavy, Stamataki, Delile, Smith, and Briscoe]{ivanovitch_ventricular_2021}
Kenzo Ivanovitch, Pablo Soro-Barrio, Probir Chakravarty, Rebecca~A. Jones, Donald~M. Bell, S.~Neda Mousavy~Gharavy, Despina Stamataki, Julien Delile, James~C. Smith, and James Briscoe.
\newblock Ventricular, atrial, and outflow tract heart progenitors arise from spatially and molecularly distinct regions of the primitive streak.
\newblock \emph{PLoS biology}, 19\penalty0 (5):\penalty0 e3001200, May 2021.
\newblock ISSN 1545-7885.
\newblock \doi{10.1371/journal.pbio.3001200}.

\bibitem[Bruneau et~al.(2001)Bruneau, Nemer, Schmitt, Charron, Robitaille, Caron, Conner, Gessler, Nemer, Seidman, and Seidman]{Bruneau2001-yk}
B~G Bruneau, G~Nemer, J~P Schmitt, F~Charron, L~Robitaille, S~Caron, D~A Conner, M~Gessler, M~Nemer, C~E Seidman, and J~G Seidman.
\newblock {A murine model of Holt-Oram syndrome defines roles of the T-box transcription factor Tbx5 in cardiogenesis and disease}.
\newblock \emph{Cell}, 106\penalty0 (6):\penalty0 709--721, 21~September 2001.
\newblock ISSN 0092-8674,1097-4172.
\newblock \doi{10.1016/s0092-8674(01)00493-7}.

\bibitem[Buckingham et~al.(2005)Buckingham, Meilhac, and Zaffran]{buckingham_building_2005}
Margaret Buckingham, Sigolène Meilhac, and Stéphane Zaffran.
\newblock Building the mammalian heart from two sources of myocardial cells.
\newblock \emph{Nature Reviews. Genetics}, 6\penalty0 (11):\penalty0 826--835, November 2005.
\newblock ISSN 1471-0056.
\newblock \doi{10.1038/nrg1710}.

\bibitem[Meilhac and Buckingham(2018)]{meilhac_deployment_2018}
Sigolène~M. Meilhac and Margaret~E. Buckingham.
\newblock The deployment of cell lineages that form the mammalian heart.
\newblock \emph{Nature Reviews. Cardiology}, 15\penalty0 (11):\penalty0 705--724, November 2018.
\newblock ISSN 1759-5010.
\newblock \doi{10.1038/s41569-018-0086-9}.

\bibitem[Jones et~al.(2026)Jones, Sun, Min, Colgan, Wang, T{\"{o}}r{\"{o}}k, Ribeiro, Xue, Cardoso, Rong, Tian, Weir, Chen, Koblan, Yost, Mathey-Andrews, D'Souza, Russell, Stickels, Balderrama, Rideout, Dai, Marrero, Kumar, Saqi, Herzberg, Izar, Chang, Lee, Jacks, Chen, Weissman, Yosef, and Yang]{Jones2026-gt}
Matthew~G Jones, Dawei Sun, Kyung Hoi~Joseph Min, William~N Colgan, Haoyu Wang, Tivadar T{\"{o}}r{\"{o}}k, Juliano Ribeiro, Jessie Xue, Erik~C Cardoso, Yao Rong, Luyi Tian, Jackson~A Weir, Victor~Z Chen, Luke~W Koblan, Kathryn~E Yost, Nicolas Mathey-Andrews, Edridge D'Souza, Andrew J~C Russell, Robert~R Stickels, Karol~S Balderrama, William~M Rideout, 3rd, Min Dai, Giovanni Marrero, Vipin Kumar, Anjali Saqi, Benjamin Herzberg, Benjamin Izar, Howard~Y Chang, Joo-Hyeon Lee, Tyler Jacks, Fei Chen, Jonathan~S Weissman, Nir Yosef, and Dian Yang.
\newblock {Spatiotemporal lineage tracing reveals the dynamic spatial architecture of tumor growth and metastasis}.
\newblock \emph{Nature Genetics}, 58\penalty0 (9):\penalty0 2398--2410, 2~September 2026.
\newblock ISSN 1061-4036,1546-1718.
\newblock \doi{10.1038/s41588-026-02739-z}.

\bibitem[Baron et~al.(2026)Baron, Amin, and Wilson]{Baron2026-ye}
Ethan Baron, Alan~N Amin, and Andrew~Gordon Wilson.
\newblock {Deep generative models for phylogenetic inference with complex evolutionary processes}.
\newblock \emph{arXiv [cs.LG]}, 13~October 2026.

\bibitem[Zhou et~al.(2023)Zhou, Yan, Layne, Whitammer, Zhang, Jain, Blanchette, and Bengio]{Zhou2023-hl}
Ming~Yang Zhou, Zichao Yan, Elliot Layne, Esmeralda~S Whitammer, Dinghuai Zhang, Moksh Jain, Mathieu Blanchette, and Yoshua Bengio.
\newblock {PhyloGFN: Phylogenetic inference with generative flow networks}.
\newblock In \emph{{The Twelfth International Conference on Learning Representations}}, 13~October 2023.

\bibitem[Wang et~al.(2026)Wang, Hu, Weiler, Mayes, Lange, Fountain, Haug, Wang, Xue, Sauka-Spengler, and Theis]{Wang2026-eg}
Weixu Wang, Zhiyuan Hu, Philipp Weiler, Sarah Mayes, Marius Lange, Daniel~M Fountain, Julianna~O Haug, Jingye Wang, Zhengyuan Xue, Tatjana Sauka-Spengler, and Fabian~J Theis.
\newblock {RegVelo: Gene-regulatory-informed dynamics of single cells}.
\newblock \emph{Cell}, 189\penalty0 (12):\penalty0 3773--3800.e44, 11~June 2026.
\newblock ISSN 0092-8674,1097-4172.
\newblock \doi{10.1016/j.cell.2026.04.022}.

\bibitem[Li et~al.(2025)Li, Jiang, Xu, Wu, Liu, Chen, Liao, Li, Wang, Guo, Zhang, Wen, Kee, and Tang]{Li2025-jk}
Mengyao Li, Zhenhuan Jiang, Xueqiang Xu, Xinglong Wu, Yun Liu, Kexuan Chen, Yuhan Liao, Wen Li, Xiao Wang, Yuqing Guo, Bo~Zhang, Lu~Wen, Kehkooi Kee, and Fuchou Tang.
\newblock {Chromatin accessibility landscape of mouse early embryos revealed by single-cell NanoATAC-{seq2}}.
\newblock \emph{Science (New York, N.Y.)}, 387\penalty0 (6741):\penalty0 eadp4319, 28~March 2025.
\newblock ISSN 0036-8075,1095-9203.
\newblock \doi{10.1126/science.adp4319}.

\bibitem[Chen et~al.(2021)Chen, Liao, Cheng, Ma, Wu, Lai, Qiu, Yang, Li, Xu, Hao, Wang, Lu, Chen, Liu, Huang, Lin, Li, Hong, Fu, Jiang, Peng, Liu, Shen, Liu, Li, Yuan, Zheng, Wang, Wang, Xiang, Han, Qin, Guo, C\'{a}noves, Thiery, Wu, Zhao, Li, Kuang, Hui, Wang, Lu, Wang, Liu, Ni, Zhang, Mu, Yin, Yang, Lisby, Cornall, Mulder, Uhlen, Esteban, Li, Liu, Xu, and Wang]{Chen2021-np}
Ao~Chen, Sha Liao, Mengnan Cheng, Kailong Ma, Liang Wu, Yiwei Lai, Xiaojie Qiu, Jin Yang, Wenjiao Li, Jiangshan Xu, Shijie Hao, Xin Wang, Huifang Lu, Xi~Chen, Xing Liu, Xin Huang, Feng Lin, Zhao Li, Yan Hong, Defeng Fu, Yujia Jiang, Jian Peng, Shuai Liu, Mengzhe Shen, Chuanyu Liu, Quanshui Li, Yue Yuan, Huiwen Zheng, Zhifeng Wang, Zhaohui Wang, Haitao Xiang, Lei Han, Baoming Qin, Pengcheng Guo, Pura Mu\~{n}oz C\'{a}noves, Jean~Paul Thiery, Qingfeng Wu, Fuxiang Zhao, Mei Li, Haoyan Kuang, Junhou Hui, Ou~Wang, Haorong Lu, Bo~Wang, Shiping Liu, Ming Ni, Wenwei Zhang, Feng Mu, Ye~Yin, Huanming Yang, Michael Lisby, Richard~J Cornall, Jan Mulder, Mathias Uhlen, Miguel~A Esteban, Yuxiang Li, Longqi Liu, Xun Xu, and Jian Wang.
\newblock {Spatiotemporal transcriptomic atlas of mouse organogenesis using DNA nanoball patterned arrays}.
\newblock \emph{bioRxiv}, page 2021.01.17.427004, 22~October 2021.
\newblock \doi{10.1101/2021.01.17.427004}.

\bibitem[Chen et~al.(2024)Chen, Choi, Li, Nathans, Martin, Yang, Hamazaki, Qiu, Lalanne, Regalado, Kim, Agarwal, Nichols, Leith, Lee, and Shendure]{Chen2024-yt}
Wei Chen, Junhong Choi, Xiaoyi Li, Jenny~F Nathans, Beth Martin, Wei Yang, Nobuhiko Hamazaki, Chengxiang Qiu, Jean-Beno\^{\i}t Lalanne, Samuel Regalado, Haedong Kim, Vikram Agarwal, Eva Nichols, Anh Leith, Choli Lee, and Jay Shendure.
\newblock {Symbolic recording of signalling and cis-regulatory element activity to {DNA}}.
\newblock \emph{Nature}, 632\penalty0 (8027):\penalty0 1073--1081, 17~August 2024.
\newblock ISSN 0028-0836,1476-4687.
\newblock \doi{10.1038/s41586-024-07706-4}.

\bibitem[Hao et~al.(2026)Hao, Liu, Barrett, Samadi, Zarezadeh, McGrath, Zernicka-Goetz, and Askary]{Hao2026-yj}
Kai Hao, Yunzheng Liu, Mykel Barrett, Zainalabedin Samadi, Amirhossein Zarezadeh, Yuka McGrath, Magdalena Zernicka-Goetz, and Amjad Askary.
\newblock {Genetic recording and in situ readout of single-cell signaling memory}.
\newblock \emph{Nature Chemical Biology}, pages 1--12, 11~March 2026.
\newblock ISSN 1552-4450,1552-4469.
\newblock \doi{10.1038/s41589-026-02168-3}.

\bibitem[Sz{\'e}kely and Rizzo(2013)]{Szekely2013}
G{\'a}bor~J Sz{\'e}kely and Maria~L Rizzo.
\newblock {Energy statistics: A class of statistics based on distances}.
\newblock \emph{Journal of Statistical Planning and Inference}, 143\penalty0 (8):\penalty0 1249--1272, 2013.
\newblock \doi{10.1016/j.jspi.2013.03.018}.
\newblock URL \url{https://pages.stat.wisc.edu/~wahba/stat860public/pdf4/Energy/JSPI5102.pdf}.

\bibitem[Song et~al.(2021{\natexlab{a}})Song, Sohl-Dickstein, Kingma, Kumar, Ermon, and Poole]{Song2021}
Yang Song, Jascha Sohl-Dickstein, Diederik~P Kingma, Abhishek Kumar, Stefano Ermon, and Ben Poole.
\newblock {Score-Based Generative Modeling through Stochastic Differential Equations}.
\newblock In \emph{International Conference on Learning Representations}, 2021{\natexlab{a}}.
\newblock URL \url{https://arxiv.org/abs/2011.13456}.

\bibitem[Ying et~al.(2021)Ying, Cai, Luo, Zheng, Ke, He, Shen, and Liu]{Ying2021-graphormer}
Chengxuan Ying, Tianle Cai, Shengjie Luo, Shuxin Zheng, Guolin Ke, Di~He, Yanming Shen, and Tie-Yan Liu.
\newblock {Do Transformers Really Perform Bad for Graph Representation?}
\newblock In \emph{Advances in Neural Information Processing Systems}, 2021.
\newblock URL \url{https://proceedings.neurips.cc/paper/2021/hash/f1c1592588411002af340cbaedd6fc33-Abstract.html}.

\bibitem[Li et~al.(2015)Li, Swersky, and Zemel]{Li2015}
Yujia Li, Kevin Swersky, and Richard Zemel.
\newblock {Generative Moment Matching Networks}.
\newblock In \emph{Proceedings of the 32nd International Conference on Machine Learning (ICML)}, volume~37 of \emph{PMLR}, pages 1718--1727, 2015.
\newblock URL \url{https://proceedings.mlr.press/v37/li15.pdf}.

\bibitem[Dziugaite et~al.(2015)Dziugaite, Roy, and Ghahramani]{Dziugaite2015}
Gintare~Karolina Dziugaite, Daniel~M Roy, and Zoubin Ghahramani.
\newblock {Training generative neural networks via Maximum Mean Discrepancy optimization}.
\newblock In \emph{Uncertainty in Artificial Intelligence (UAI)}, pages 258--267, 2015.
\newblock URL \url{https://auai.org/uai2015/proceedings/papers/230.pdf}.

\bibitem[Sohl-Dickstein et~al.(2015)Sohl-Dickstein, Weiss, Maheswaranathan, and Ganguli]{SohlDickstein2015}
Jascha Sohl-Dickstein, Eric~A Weiss, Niru Maheswaranathan, and Surya Ganguli.
\newblock {Deep Unsupervised Learning using Nonequilibrium Thermodynamics}.
\newblock In \emph{International Conference on Machine Learning}, 2015.
\newblock URL \url{https://proceedings.mlr.press/v37/sohl-dickstein15.html}.

\bibitem[Song et~al.(2021{\natexlab{b}})Song, Meng, and Ermon]{Song2020ddim}
Jiaming Song, Chenlin Meng, and Stefano Ermon.
\newblock {Denoising Diffusion Implicit Models}.
\newblock In \emph{International Conference on Learning Representations}, 2021{\natexlab{b}}.
\newblock URL \url{https://arxiv.org/abs/2010.02502}.

\bibitem[Li and He(2025)]{x1prediction}
Tianhong Li and Kaiming He.
\newblock {Back to Basics: Let Denoising Generative Models Denoise}.
\newblock \emph{arXiv preprint arXiv:2511.13720}, 2025.
\newblock URL \url{https://arxiv.org/abs/2511.13720}.

\bibitem[Tashiro et~al.(2021)Tashiro, Song, Song, and Ermon]{Tashiro2021-ko}
Yusuke Tashiro, Jiaming Song, Yang Song, and Stefano Ermon.
\newblock {CSDI: Conditional Score-based Diffusion Models for Probabilistic Time Series Imputation}.
\newblock \emph{Advances in Neural Information Processing Systems}, 34:\penalty0 24804--24816, 6~December 2021.

\bibitem[Feydy et~al.(2018)Feydy, S\'{e}journ\'{e}, Vialard, Amari, Trouv\'{e}, and Peyr\'{e}]{Feydy2018-ow}
Jean Feydy, Thibault S\'{e}journ\'{e}, Fran\c{c}ois-Xavier Vialard, Shun-Ichi Amari, Alain Trouv\'{e}, and Gabriel Peyr\'{e}.
\newblock {Interpolating between Optimal Transport and MMD using Sinkhorn divergences}.
\newblock \emph{arXiv [math.ST]}, 18~October 2018.
\newblock \doi{10.48550/arXiv.1810.08278}.

\bibitem[T{\"{u}}mpel et~al.(2009)T{\"{u}}mpel, Wiedemann, and Krumlauf]{Tumpel2009-cw}
Stefan T{\"{u}}mpel, Leanne~M Wiedemann, and Robb Krumlauf.
\newblock {Hox genes and segmentation of the vertebrate hindbrain}.
\newblock \emph{Current Topics in Developmental Biology}, 88:\penalty0 103--137, 2009.
\newblock ISSN 0070-2153,1557-8933.
\newblock \doi{10.1016/S0070-2153(09)88004-6}.

\bibitem[Parrish et~al.(2009)Parrish, Nolte, and Krumlauf]{Parrish2009-yn}
M~Parrish, C~Nolte, and R~Krumlauf.
\newblock {Hox genes expression}.
\newblock In \emph{{Developmental Neurobiology}}. Academic Press, 2009.
\newblock ISBN 9780123751676.

\bibitem[Minoux and Rijli(2010)]{Minoux2010-ji}
M~Minoux and F~M Rijli.
\newblock {Molecular mechanisms of cranial neural crest cell migration and patterning in craniofacial development}.
\newblock \emph{Development}, 137\penalty0 (16):\penalty0 2605--2621, 2010.

\bibitem[Philippidou and Dasen(2013)]{Philippidou2013-sq}
P~Philippidou and J~S Dasen.
\newblock {Hox genes: choreographers in neural development, architects of circuit organization}.
\newblock \emph{Neuron}, 80\penalty0 (1):\penalty0 12--34, 2013.

\bibitem[Grosswendt et~al.(2020)Grosswendt, Kretzmer, Smith, Kumar, Hetzel, Wittler, Klages, Timmermann, Mukherji, and Meissner]{Grosswendt2020-se}
Stefanie Grosswendt, Helene Kretzmer, Zachary~D Smith, Abhishek~Sampath Kumar, Sara Hetzel, Lars Wittler, Sven Klages, Bernd Timmermann, Shankar Mukherji, and Alexander Meissner.
\newblock {Epigenetic regulator function through mouse gastrulation}.
\newblock \emph{Nature}, 584\penalty0 (7819):\penalty0 102--108, 29~August 2020.
\newblock ISSN 0028-0836,1476-4687.
\newblock \doi{10.1038/s41586-020-2552-x}.

\bibitem[Kijima et~al.(2026)Kijima, Villagrana, Wang, Tornisiello, Dias, Hou, Tse, Wang, Kretzmer, Sumigray, Akta\c{s}, and Smith]{Kijima2026-fd}
Yusuke Kijima, JoAnne Villagrana, Minming Wang, Rosaria Tornisiello, Nicolas Dias, Jason Tien-Chi Hou, Kevin Tse, Jennifer Wang, Helene Kretzmer, Kaelyn Sumigray, Tu\u{g}\c{c}e Akta\c{s}, and Zachary~D Smith.
\newblock {Epigenetic regulators partition the genome by evolutionary and developmental constraint}.
\newblock \emph{bioRxiv}, page 2026.10.01.755350, 2~October 2026.
\newblock \doi{10.64898/2026.10.01.755350}.

\bibitem[Dias et~al.(2026)Dias, Lewinsohn, Colgan, Wang, Kijima, Villagrana, Hou, Gowri, Chau, Akta\c{s}, Sumigray, Weissman, Koblan, Wagner, and Smith]{Dias2026-vh}
Nicolas Dias, Daniel~P Lewinsohn, William~N Colgan, Minming Wang, Yusuke Kijima, JoAnne Villagrana, Tien-Chi~Jason Hou, Gokul Gowri, Adelina Chau, Tu\u{g}\c{c}e Akta\c{s}, Kaelyn Sumigray, Jonathan~S Weissman, Luke~W Koblan, Allon Wagner, and Zachary~D Smith.
\newblock {Folate deficiency disrupts key metabolic transitions within the developing neural ectoderm}.
\newblock \emph{bioRxiv: The Preprint Server for Biology}, page 2026.09.18.752622, 22~September 2026.
\newblock \doi{10.64898/2026.09.18.752622}.

\bibitem[Henrique et~al.(2015)Henrique, Abranches, Verrier, and Storey]{Henrique2015-hb}
Domingos Henrique, Elsa Abranches, Laure Verrier, and Kate~G Storey.
\newblock {Neuromesodermal progenitors and the making of the spinal cord}.
\newblock \emph{Development (Cambridge, England)}, 142\penalty0 (17):\penalty0 2864--2875, 1~September 2015.
\newblock ISSN 0950-1991,1477-9129.
\newblock \doi{10.1242/dev.119768}.

\bibitem[Chuang and Jegelka(2022)]{chuang2022tree}
Ching-Yao Chuang and Stefanie Jegelka.
\newblock {Tree Mover's Distance: Bridging Graph Metrics and Stability of Graph Neural Networks}.
\newblock In \emph{Advances in Neural Information Processing Systems}, volume~35, 2022.
\newblock URL \url{https://proceedings.neurips.cc/paper_files/paper/2022/hash/139ae969f49abd9a113981c1f7fce5ce-Abstract-Conference.html}.

\bibitem[O'Donnell(2014)]{o2014analysis}
Ryan O'Donnell.
\newblock \emph{Analysis of boolean functions}, volume~2.
\newblock Cambridge University Press, Cambridge, 2014.

\bibitem[Sejdinovic et~al.(2013)Sejdinovic, Sriperumbudur, Gretton, and Fukumizu]{sejdinovic2013equivalence}
Dino Sejdinovic, Bharath Sriperumbudur, Arthur Gretton, and Kenji Fukumizu.
\newblock Equivalence of distance-based and rkhs-based statistics in hypothesis testing.
\newblock \emph{The annals of statistics}, pages 2263--2291, 2013.

\bibitem[Gretton et~al.(2012)Gretton, Borgwardt, Rasch, Sch{\"o}lkopf, and Smola]{gretton2012kernel}
Arthur Gretton, Karsten~M Borgwardt, Malte~J Rasch, Bernhard Sch{\"o}lkopf, and Alexander Smola.
\newblock A kernel two-sample test.
\newblock \emph{The journal of machine learning research}, 13:\penalty0 723--773, 2012.

\bibitem[Wang(1979)]{wang1979dependent}
YH~Wang.
\newblock Dependent random variables with independent subsets.
\newblock \emph{The American Mathematical Monthly}, 86\penalty0 (4):\penalty0 290--292, 1979.

\end{thebibliography}

\section{Methods}

\newcommand{\R}{\mathbb{R}}
\newcommand{\calT}{\mathcal{T}}
\newcommand{\calX}{\mathcal{X}}

\subsection{Implementing lineage inpainting}

In this section, we outline in more detail the precise model fitting, model architecture, and sampling procedures we implement in lineage inpainting models. We first introduce the general impute-and-sample framework, then the details of our particular implementations of the imputation and sampling step. 

We describe three technical contributions: (1) a tree-aware attention mechanism inspired by phylogenetic comparative methods \cite{Felsenstein1985-ih}, (2) a divergence between distributions over complete lineage trees which we use as a loss function for imputation models \cite{Szekely2013}, and (3) a stochastic interpolant model for sampling from distributions over lineage trees \cite{Albergo2023}.

\subsubsection{The impute-and-sample framework}
\label{supp:framework}
Our goal is to fit a model which, given a lineage tree with a single snapshot of states observed, can sample from the distribution of possible complete trees where all states, past and future, are known (recall \textbf{Fig.~\ref{fig:schematic}a}). The data available are incomplete, in the sense that we only ever observe the states at the leaves of lineage trees (\textbf{Fig.~\ref{fig:schematic}b}). 

Many approaches to fitting such a model can be broken into two components which we refer to as imputation and sampling (\textbf{Fig.~\ref{fig:schematic}c}). The central idea is that the sampler is a generative model which must be fit using complete lineage trees where all node states are observed. Because our direct observations have states only at the leaves, the sampler must instead be fit on trees where the internal node states are somehow imputed by a separate procedure. 

Existing methods can be thought of under the umbrella of this impute-and-sample framework. In GENOT \cite{Klein2024-mk}, for example, the imputer is an entropic optimal-transport coupling between marginals at different timepoints, which defines a distribution over possible ancestor--descendant pairings; the sampler is a cell-level conditional flow-matching model trained to reproduce the cell--cell relationships encoded by this coupling.

The impute-and-sample framework is useful for laying out the design space of models. One can imagine many different imputation methods: discrete optimal transport methods, deterministic ancestral reconstructions (such as Felsenstein's method \cite{Felsenstein1985-ih}) or a learnable approach as we will describe below in Section~\ref{sec:coupling}. Similarly, one can imagine many different samplers: a large family of generative models for conditional sampling have been explored in recent years \cite{Tong2023-mf,Albergo2023, Song2021, Ho2020}. One could apply these directly and sample states independently, or one could try to sample states across an entire tree jointly, as we will describe below in Sec.~\ref{sec:diffusion}. In Supp.~\ref{supp:grid} we explore many possible pairings of imputers and samplers, and choose our implementation accordingly.

\subsubsection{A tree-aware transformer architecture for lineage trees}
\label{sec:architecture}

Both the imputation and sampling models we introduce below are built from a shared architectural backbone for processing lineage trees simultaneously. It is an extension of the transformer architecture \cite{Vaswani2017, Turner2023-df}, which maps a set of input tokens to a set of output tokens. Here, each node of a lineage tree corresponds to a token. Beyond this redefinition of the token, we make two core modifications: (1) a biased attention layer which encodes the structure of the lineage tree and (2) a tree-dependent additive displacement for each output token, which provides an inductive bias inspired by classical ideas in phylogenetics about modeling the evolution of states on a tree \cite{Felsenstein1985-ih}.

\paragraph{Encoding tree structure in a biased attention mechanism} Transformers are built around the attention mechanism, which, without an added positional embedding, is permutation equivariant. If we were to directly use a transformer (without a positional embedding) to operate on lineage trees, with each cell corresponding to a token, the structure of the lineage tree would be ignored. Instead, the cells would be processed as a ``bag of cells'' with no internal relational structure.

To encode the structure of the lineage tree in the transformer architecture, we instead use a biased attention layer (\textbf{Box~\ref{box:attention}}) where attention weights receive an additive bias corresponding to a scaled version of the patristic distance.

\paragraph{An additive displacement for output tokens} Rather than using the transformer to directly predict raw output tokens, we instead recursively average child states for ancestral inference, or copy the value of the observed ancestor for descendant inference. This allows the transformer to predict a residual on top of this displacement state for each token, encoding an inductive bias into the architecture. Our choice of this displacement draws inspiration from classical approaches in phylogenetics. Felsenstein's phylogenetic comparative method \cite{Felsenstein1985-ih} assumes that traits evolve along phylogenetic trees following a Brownian motion. That is, across every edge of the tree, the change in the trait is sampled from a Gaussian distribution with variance proportional to the length of the edge (the amount of time elapsed).%

Given the observed states at the leaves of the tree, the expected state of any unobserved node under the Brownian model is a weighted average of its observed relatives. Conditioning on its descendant subtree, the posterior mean of an ancestral node is a variance-weighted mean of its descendant states. This is Felsenstein's ancestral reconstruction, computed by a bottom-up recursion from the leaves of the tree~\cite{Felsenstein1985-ih}. We compute a simpler, unweighted mean at every node from the anchor states and the topology and take it as a per-node displacement $b_v$, a first guess at each cell before any learning. The network then predicts only each node's deviation from $b_v$, a residual far easier to learn than the state outright, and we add $b_v$ back at the output. Section~\ref{supp:architecture} gives the exact recursion.
 
One forward pass runs as follows. Each node's current state is linearly projected to a hidden token, then we add learned embeddings of the node's lineage time \(t_v\) and a marker distinguishing a clamped observed anchor from a node to be filled in. The tokens pass through a stack of attention layers, where the patristic-distance bias enters. As in graph transformers that encode pairwise or structural information inside attention~\cite{Ying2021-graphormer}, we add the patristic distance to each pair's attention logit before the softmax following ALiBi~\cite{Press2022}, scaled by a weight learned separately for each layer and head (Box~\ref{box:attention}). In practice, the model develops large negative weights in early layers that concentrate attention on close relatives, and progressively less negative weights in later layers with positive weights only in the last layers. A final head maps each node back to state space, where we add \(b_v\) back to the output.

The two models described below share this architectural backbone with one difference: the sampling model \(S_\phi\) takes one extra input, the stochastic interpolant time \(s\in[0,1]\) (Section~\ref{sec:diffusion}), added to each token through a learned embedding, whereas the imputation model \(I_\theta\) takes none, as it is a single step model. Table~\ref{tab:architecture} (Section~\ref{supp:architecture}) specifies the architecture used for the mouse-embryo analyses.

\subsubsection{Learning an imputation model}
\label{sec:coupling}

The imputation model $I_\theta$ uses the tree-aware transformer introduced above. The data consist of several lineage trees at different developmental times, with a collection of multiple trees observed at each time. For expository purposes, consider two collections, observed at capture times $c, c'$ such that $c<c'$. Let $\mathcal I$ and $\mathcal I'$ denote the indices of the trees observed at time $c$ and $c'$ respectively. For any $i\in \mathcal I\cup \mathcal I'$, let
\[
\calT_i = (V_i,E_i,t_i)
\]
denote the topology of the \(i\)th lineage tree, where \(t_{i,v}\) gives the developmental time of node \(v\), and let \(A_i\subseteq V_i\) denote the observed snapshot, with measured cell states \(x_{i,A_i}\). We use
\[
\mathcal X_i=(\mathcal T_i,x_i)
\]
to denote a state-decorated tree: the topology \(\calT_i\) along with the states $x_{i,v}$ at each node.

The imputation model takes the topology and the observed snapshot as input, with standard Gaussian noise on the unobserved nodes, and samples a completed tree
\[
\widehat{\calX}_i = (\calT_i,\widehat x_i) =I_\theta(\calT_i,A_i,x_{i,A_i}),
\]
where \(\widehat x_{i,v}=x_{i,v}\) for observed nodes and \(\widehat x_{i,v}\) is imputed for each unobserved ancestral node. We first use the imputation model to complete each tree in the two collections.

The challenge is that we never observe the states of both ancestors and descendants within a single tree. If we did we could just use that to evaluate the resulting filled in trees. Instead, we observe only the leaves of the trees at their respective observation times. To address this missing-data problem, we train the imputer indirectly, by requiring the completed trees to be distributionally consistent with one another. The key question is how to define and measure distributional consistency across these collections of lineage trees in practice. 

The most basic requirement for distributional consistency is matching time-marginals: at the same developmental time $t$, the distribution of cell states should agree across collections observed at different times \cite{Li2015,Dziugaite2015}. For each collection, we pool the states of all cells alive at time $t$ across its completed trees, forming a single empirical distribution of cell states. We then compare the two resulting distributions (one for each collection).

We compare empirical distributions using the energy distance \cite{Szekely2013}. For two sets of vectors $Z=(z_a)_{a=1}^m$ and $Z'=(z'_b)_{b=1}^n$, the empirical energy
distance is
\[
\widehat D_E(Z,Z')
=
\frac{2}{mn}
\sum_{a=1}^m
\sum_{b=1}^n
\lVert z_a-z'_b\rVert
-
\frac{1}{m^2}
\sum_{a=1}^m
\sum_{a'=1}^m
\lVert z_a-z_{a'}\rVert
-
\frac{1}{n^2}
\sum_{b=1}^n
\sum_{b'=1}^n
\lVert z'_b-z'_{b'}\rVert .
\]
For this marginal component, the set of vectors $Z$ is just the set of states at time $t$ pooled across all trees in $\mathcal I$, and $Z'$ the corresponding states pooled across all trees in $\mathcal I'$. Intuitively, the energy distance compares the average distance between states drawn from the two collections with the average distance between states within each collection. When the two sets of states are the same the distance will be zero, and as the between-collection differences exceed the within-collection heterogeneity the distance becomes large.

Matching only the temporal marginals (the population of cell states at each time) is not enough to learn developmental dynamics \cite{Weinreb2018-wy}. The marginal distribution constrains which states are present at each time, but not how those states are connected on the lineage tree. In fact, under a marginal-only objective, the imputer could learn an entirely uninformative coupling, and simply randomly sample the internal nodes independently from the cell-state distribution at each developmental time. This would perfectly match every time marginal while learning no meaningful ancestor--descendant coupling. 

Lineage tracing provides us with more constraints. We can exploit lineage information by requiring the two collections to agree not only in their time marginals, but also in the distributions of motifs across their lineage trees. In practice, we use an imputation loss which aims to match three types of distributions: (1) marginals at individual developmental times, (2) lineage paths across multiple times, and (3) common-ancestor triples relating descendants of a shared ancestor.

Our second loss component enforces that the distributions of individual cell trajectories agree across the two collections. For each completed tree, we sample a lineage path, recording the sequence of cell states encountered at the shared comparison times. We concatenate these states into a single vector---for example, \((x_{t_0},x_{t_1},x_{t_2})\). Pooling these paths across all trees in each collection generates two empirical distributions of paths, which we compare using the same energy distance. An imputer that samples ancestors independently of their descendants can have the correct marginals while failing to preserve dependence along the lineage. With more than two observation times, the same construction matches the joint distribution of states along the lineage.

Path matching constrains the ``depth'' relationship between states along a trajectory, but it omits the ``breadth'' relationship between descendants. This ``breadth'' dimension can also carry dependence. For example, a cell's descendants may be conditionally dependent given the ancestral state, so that knowing both the ancestor and one descendant is informative about other descendants (e.g., due to interaction between sister cells). A model trained only on the marginal and path components could produce the correct ancestor--descendant path distribution while ignoring this dependence. Our third component addresses this issue. We repeatedly sample an ancestral cell at time \(t_0\) and two of its descendants at a later time \(t_1\), concatenate their states, \((x_{\mathrm{ancestor}},x_{\mathrm{descendant}\,1},x_{\mathrm{descendant}\,2})\), pool these triples across the trees in each collection, and compare the resulting distributions using the energy distance. Because the two descendants have no canonical ordering, we choose each of the two possible orderings with probability $1/2$. %

Our imputation loss is the weighted sum of these three discrepancies: the state marginal, lineage paths, and common-ancestor triples. We compare each feature only at matching developmental times: separately at each time for marginals, across the shared sequence of times for paths, and separately for each ancestor--descendant time pair for common-ancestor triples. Let \(t_0<\cdots<t_K\leq c\) denote the developmental times at which features are compared in both collections. The loss is
\[
\begin{aligned}
\mathcal L_{\mathrm{imp}}(\theta)
={}&
\lambda_{\mathrm{marg}}
\sum_{k=0}^K
\widehat D_E\!\left(
    \operatorname{Marginal}
    (\{\widehat{\calX}_i\}_{i\in \mathcal I},t_k),
    \operatorname{Marginal}
    (\{\widehat{\calX}_i\}_{i\in \mathcal I'},t_k)
\right)
\\
&+
\lambda_{\mathrm{path}}
\widehat D_E\!\left(
    \operatorname{Path}
    (\{\widehat{\calX}_i\}_{i\in \mathcal I};t_0,\ldots,t_K),
    \operatorname{Path}
    (\{\widehat{\calX}_i\}_{i\in \mathcal I'};t_0,\ldots,t_K)
\right)
\\
&+
\lambda_{\mathrm{anc}}
\sum_{0\leq k<\ell\leq K}
\widehat D_E\!\left(
    \operatorname{CommonAncestor}
    (\{\widehat{\calX}_i\}_{i\in \mathcal I};t_k,t_\ell),
    \operatorname{CommonAncestor}
    (\{\widehat{\calX}_i\}_{i\in \mathcal I'};t_k,t_\ell)
\right).
\end{aligned}
\]
Here the path and common-ancestor operators pool the corresponding features across all completed trees in each collection. In practice, the lineage trees are too large to process as single training examples, so we estimate this objective using two collections of subtrees rooted at a shared internal time, sampled from two larger lineage trees. For the marginal component, we can leverage large scale observational data to compare each collection's imputed states with measured cells sampled from the non-lineage-traced developmental atlases as well (Section~\ref{sec:atlas}). Algorithm~\ref{alg:imputation-loss} summarizes the feature comparisons.

Each of these terms can be viewed as comparing a low-dimensional feature distribution induced by the empirical distribution of completed trees. In principle, we could instead compare the distributions of whole state-decorated trees directly. Whole lineage trees, however, are large, high-dimensional, may differ in topology, and we have a limited number of samples, making such a comparison difficult to estimate. Section~\ref{supp:theory} formalizes this connection in an idealized fixed-topology setting: it constructs a distributional distance on whole trees and characterizes the information and statistical efficiency lost or gained by replacing the whole-tree comparison with the marginal, path, and common-ancestor feature distributions used here.

\begin{algorithm}[t]
\caption{Imputation alignment loss}
\label{alg:imputation-loss}
\small
\begin{algorithmic}[1]

\Function{Marginal}{$\{\widehat{\calX}_i\}_{i\in \mathcal I},t$}
    \State $F \gets [\,]$
    \For{$i\in \mathcal I$}
        \For{each cell $v$ in $\widehat{\calX}_i$ alive at time $t$}
            \State append $\widehat x_{i,v}$ to $F$
        \EndFor
    \EndFor
    \State \Return $F$
\EndFunction

\Statex

\Function{Path}{$\{\widehat{\calX}_i\}_{i\in \mathcal I},(t_0,\ldots,t_K)$}
    \State $F \gets [\,]$
    \For{$i\in \mathcal I$}
        \For{each sampled lineage path across $(t_0,\ldots,t_K)$}
            \State let $v_k$ be its cell alive at $t_k$
            \State append
            $(\widehat x_{i,v_0},\ldots,\widehat x_{i,v_K})$
            to $F$
        \EndFor
    \EndFor
    \State \Return $F$
\EndFunction

\Statex

\Function{CommonAncestor}{$\{\widehat{\calX}_i\}_{i\in \mathcal I},t_0,t_1$}
    \State $F \gets [\,]$
    \For{$i\in \mathcal I$}
        \For{each sampled cell $u$ in $\widehat{\calX}_i$ alive at $t_0$}
            \State sample two descendants $v,w$ of $u$ alive at $t_1$
            \State randomly order $(v,w)$
            \State append
            $(\widehat x_{i,u},\widehat x_{i,v},\widehat x_{i,w})$
            to $F$
        \EndFor
    \EndFor
    \State \Return $F$
\EndFunction

\Statex
\Require Two collections of observed trees indexed by $\mathcal I$ and $\mathcal I'$, and imputer $I_\theta$

\State Complete each tree:
\[
\widehat{\calX}_i
=
I_\theta(\calT_i,A_i,x_{i,A_i}),
\qquad i\in \mathcal I\cup \mathcal I'.
\]

\State Let $t_0<\cdots<t_K$ be the developmental times shared by both collections.

\State $\mathcal L_{\rm marg}\gets 0$
\For{$k=0,\ldots,K$}
    \State $\mathcal L_{\rm marg}\gets\mathcal L_{\rm marg}
    +\widehat D_E\!\left(
        \Call{Marginal}{\{\widehat{\calX}_i\}_{i\in \mathcal I},t_k},
        \Call{Marginal}{\{\widehat{\calX}_j\}_{j\in \mathcal I'},t_k}
      \right)$
\EndFor

\State $\mathcal L_{\rm path}\gets
    \widehat D_E\!\left(
        \Call{Path}{\{\widehat{\calX}_i\}_{i\in \mathcal I},(t_0,\ldots,t_K)},
        \Call{Path}{\{\widehat{\calX}_j\}_{j\in \mathcal I'},(t_0,\ldots,t_K)}
    \right)$

\State $\mathcal L_{\rm anc}\gets0$
\For{$0\leq k<\ell\leq K$}
    \State $\mathcal L_{\rm anc}\gets\mathcal L_{\rm anc}
    +\widehat D_E\!\left(
        \Call{CommonAncestor}{\{\widehat{\calX}_i\}_{i\in \mathcal I},t_k,t_\ell},
        \Call{CommonAncestor}{\{\widehat{\calX}_j\}_{j\in \mathcal I'},t_k,t_\ell}
      \right)$
\EndFor

\State \Return
$\lambda_{\rm marg}\mathcal L_{\rm marg}
+\lambda_{\rm path}\mathcal L_{\rm path}
+\lambda_{\rm anc}\mathcal L_{\rm anc}$

\end{algorithmic}
\end{algorithm}

\subsubsection{Learning a generative sampling model}
\label{sec:diffusion}

The imputation model produces completed trees by filling in ancestral states in a single forward pass. We want a generative model that, given a tree observed at a single snapshot, can sample the unobserved states of each cell's ancestors and descendants. This requires a second sampling model, $S_\phi$, which we train by taking the complete trees produced by the imputation step and randomly choosing a developmental time to condition on, then masking out all the other nodes and learning to fill in the states conditional on this selected snapshot. The sampling model uses the same tree-aware transformer architecture but generates the masked states through a multi-step denoising procedure, conditioning on the retained snapshot throughout. This builds on diffusion, flow, and stochastic interpolant based generative modeling techniques which have seen great success in computer vision \cite{SohlDickstein2015, Ho2020, Song2021, Lugmayr2022, Lipman2022-qc, Albergo2023}.

For a completed tree $\widehat{\mathcal X}_i=(\mathcal T_i,\widehat x_i)$, we choose a random developmental time and treat the cells alive at that time as the observable snapshot. We hide the states $\widehat x_{i,M_i}$ of the remaining nodes $M_i\subset V_i$ and train the model to resample them from everything that remains: the topology, the node times, the observable states, and which nodes were hidden. It is worth noting that, unless we select the leaf snapshot, this conditions on imputed nodes.

We do this with a stochastic interpolant, using a transformer $S_\phi$ that takes as input one token per node (both masked and unmasked). An observable node's token carries its state $\widehat x_{i,v}$. A hidden node's token instead carries an interpolant between noise and its clean state,
\begin{equation}
x_{i,v}^{s}
=
b_{i,v}
+
\alpha_s\bigl(\widehat x_{i,v}-b_{i,v}\bigr)
+
\sigma_s\varepsilon_{i,v},
\qquad
\varepsilon_{i,v}\sim\mathcal N(0,\mathbf I).
\end{equation}
Here $\widehat x_{i,v}$ is the clean state and $b_{i,v}$ is the additive displacement introduced above, computed from the observable snapshot and the tree topology. The functions $\alpha_s$ and $\sigma_s$ control the relative amount of clean signal and Gaussian noise in $x_{i,v}^{s}$. Different choices of these schedules let us train both standard diffusion models and flow matching models \cite{Ho2020,Song2020ddim,Lipman2022-qc,Albergo2023}. We consider both the linear conditional flow matching (CFM) interpolant and the variance-preserving SDE (VP-SDE) interpolant:
\begin{equation}
\begin{aligned}
\text{CFM:}\qquad
&\alpha_s=s,
\qquad
\sigma_s=1-s,
\\
\text{VP-SDE:}\qquad
&\alpha_s
=
\exp\!\left(
-\frac12\int_0^{1-s}\beta(u)\,du
\right),
\qquad
\sigma_s=\sqrt{1-\alpha_s^2},
\end{aligned}
\end{equation}
where $\beta(u) = \beta_{\rm{min}} + u(\beta_{\rm{max}} - \beta_{\rm{min}})$ is the linear VP-SDE noise-rate schedule. Thus the interpolant is the clean state at $s=1$ and approximately Gaussian noise around the additive displacement at $s=0$.

We draw $\varepsilon_{i,v}$ independently at each hidden node, sample one interpolation time $s\sim\mathcal U[0,1]$ for the whole tree, and add an embedding of $s$ to every token. Each token also carries an embedding of its developmental node time $t_{i,v}$, which is distinct from the interpolation time $s$.

At each hidden node $S_\phi$ predicts the clean state directly (the clean-state parameterization \citep{x1prediction}), with the additive displacement added back to the network's predicted residual, and we minimize the squared error
\begin{equation}
\mathcal L_{\mathrm{samp}}
=
\mathop{\mathbb E}_{i, \varepsilon_i, s}
\left[
\sum_{v\in M_i}
\left\|
S_\phi(\mathcal T_i,M_i,x_i^s,s)_v
-
\widehat x_{i,v}
\right\|^2
\right],
\end{equation}
where $x_i^s$ collects the input states at all nodes: observable nodes retain their states $\widehat x_{i,v}$, while hidden nodes carry the interpolants defined above. The expectation averages over completed trees from our imputed data, the choice of observable snapshot, the interpolation time $s\sim\mathcal U[0,1]$, and the noise. Since the observable snapshot is selected at a random developmental time, the same training procedure teaches the model to sample in both temporal directions. This approach can be thought of as a tree analog of stochastic-interpolant methods for image inpainting and for imputing missing values in time series \cite{Albergo2024-pq, Tashiro2021-ko}.

To sample, we set the hidden nodes to Gaussian noise around their additive displacements at $s=0$ and integrate up to $s=1$, using the prediction of $S_\phi$ to step the interpolant toward data (\textbf{Fig.~\ref{fig:schematic}d}). We use a deterministic ODE sampler for CFM and a stochastic reverse-SDE sampler for VP-SDE, keeping the observable snapshot fixed throughout (Section~\ref{supp:architecture}). We can re-sample the distribution of states conditioned on the same snapshot by re-sampling the noise and re-running the denoising process.

We evaluate samplers using both the VP-SDE and CFM in Section~\ref{supp:grid}, and find that the VP-SDE is more performant. We use the VP-SDE based sampler in all main text analyses.

\subsubsection{Jointly training the imputation and sampling models}
\label{sec:training}
The imputation and sampling models are trained simultaneously rather than in separate stages. Each training step draws two observed trees whose capture times differ by half a day, $d, d+0.5$ so that the imputation loss can compare their completions over the times they share. We take up to $2{,}048$ subtrees from each, rooted at a common time chosen uniformly from $[d-0.5, d]$. Subtrees with more than \(128\) nodes are subsampled by randomly dropping leaves. At every training step, we impute a batch of observed trees and compute the imputation loss of Algorithm~\ref{alg:imputation-loss}, and the same completed trees are then used to train the sampler. The two losses are added and a single optimizer step updates both parameters \(\theta\) and \(\phi\) at once. 

Training the sampler on completed trees creates a tension. If we train directly on imputed states, the sampler may learn to reproduce imperfect samples from the imputation model. But using the imputed completion directly has an important advantage: the sampling loss can propagate back through the imputer. This gives the two models an autoencoder-like relationship, encouraging the imputer to produce completions from which the sampler can reconstruct the measured cells from a single snapshot. In practice, this additional signal makes the completions sharper and more lineage-coherent than the imputation loss alone (Section~\ref{supp:grid}).

We balance these objectives by alternating between two training steps. On one step, we leave the imputed completion unchanged, condition the sampler on an internal snapshot, and compute the loss over the resulting tree. The imputed internal states still enter the sampler as context, so gradients can flow through them into the imputer, while the sampler loss is only computed on the leaves. On the other step, we select the conditioning snapshot and leave those states at their imputed values, then replace the other imputed targets with their nearest measured cells at nearby developmental times. These replacements give the sampler real-cell targets at internal times. Although the nearest-neighbor targets are detached from the imputer, retaining the imputed conditioning snapshot preserves a gradient path back into the imputer. Alternating the two steps lets the sampler train on measured-cell targets throughout development while retaining the gradient flowing from sampler through the imputer.

The imputation loss weights its marginal, path and common-ancestor components \(1\), $0.25$ and $0.25$, and the marginal component is compared against up to \(4{,}096\) measured cells randomly sampled at each comparison time. We train for \(20{,}000\) steps with Adam and keep an exponential moving average.

\subsection{Model evaluation}
\label{sec:evaluation}

\subsubsection{Synthetic benchmarks}
\label{sec:synthetic}
\label{supp:synthetic}

Synthetic benchmarks provide a setting where the ground-truth dynamics are known. Every synthetic dataset described here is a population of trees, with each tree observed at a single time, matching the structure of the real data. Two of the four benchmark settings are fully synthetic and are built from what we call branch-modified stochastic differential equations (described below). The other two are semi-synthetic, built by wiring real mouse-embryo cell states into prescribed lineage motifs \cite{Tran2024-eu}.

\paragraph{Branch-modified SDEs.}
We define branch-modified stochastic differential equation as SDEs whose parameters vary across the branches of a lineage tree. In our synthetic datasets, a cell's latent state is one-dimensional and follows an Ornstein--Uhlenbeck process $dx=\theta(\mu-x)\,dt+\sigma\,dW$, reducing to a Brownian motion at $\theta=0$. As in a conventional branching process, cells divide by branching into two daughters, each of which starts from its parent's final state, so state is inherited continuously across divisions. At certain divisions the daughter's attractor $\mu$ is set to a value determined by its position in the lineage tree -- a branch modification of the SDE parameters. Through the branch modified SDEs, we are able to construct processes where a cell's future then depends not only on its current state, but also on its position in the lineage tree.

Trajectories are integrated by Euler-Maruyama: $x_{i+1}=x_i+\theta(\mu-x_i)\,\Delta t+\sigma\sqrt{\Delta t}\,\xi_i$ with $\xi_i\sim\mathcal N(0,1)$ and $\Delta t$ is the integration timestep. The diffusion is additive and state-independent.  To vary the ambient dimension we lift the one-dimensional latent to $\mathbb{R}^d$ by $y=xw+\varepsilon$, with $w$ a fixed random unit vector shared across trees and $\varepsilon\sim\mathcal N\!\bigl(0,(0.1/\sqrt d)^2I_d\bigr)$.

\paragraph{Synthetic topology benchmark (crossing versus non-crossing).}
Our first benchmark tests whether a model uses lineage structure. It involves a pair of processes -- referred to as crossing and coherent in \textbf{Fig.~\ref{fig:synthetic}} -- with similar marginals but distinct joint distributions. The trees have nine generations. The root starts at zero, and cells are Brownian ($\theta=0$, $\mu=0$) until the split generation $5$, where the two daughters of each cell switch to Ornstein--Uhlenbeck dynamics with $\theta=2.0$, $\sigma=0.1$, and attractors $-1$ and $+1$ respectively. At generation $7$ each family splits again. In the non-crossing process, the $-1$ family is mapped to $(-1.5,-0.5)$ and the $+1$ family is mapped to $(0.5,1.5)$, while the crossing process maps them to $(-1.5,0.5)$ and $(-0.5,1.5)$ respectively. Each branching interval (i.e., time between cell division) is $1.0$ with $100$ integration steps. The $40$ training trees are binary to depth $9$ with $512$ leaves each, observed at one of the times $5$ to $9$. We test various dimensionalities across $d\in\{1,4,16,64\}$ with a separately trained model at each dimension, and held-out evaluation uses a single further realization of the process observed at every time.

\paragraph{Synthetic context benchmark.}
This benchmark tests whether a model can leverage the states of related cells across a subtree. A progenitor commits stochastically to one of two fates early in the tree. For several generations afterwards the committed fate biases each descendant's state only weakly, by an amount that becomes small relative to the noise as $\sigma$ grows, so a single cell's state barely reveals its fate; only later do the two fates separate sharply. The bias is shared across a clone, so averaging over related cells accumulates it and exposes the latent fate.

The trees have ten generations. The root starts at zero, and cells are Brownian until the split generation $3$, where each cell draws one of two fates independently, with probability $0.75$ for the positive fate, and acquires the corresponding weak primed attractor $\mu\in\{+0.2,-0.2\}$ under Ornstein--Uhlenbeck dynamics with $\theta=2.0$. At generation $9$, the attractor of each primed family switches to $\mu\in\{+1.0,-1.0\}$. Each branch interval is $1.0$ with $200$ integration steps, the $200$ training trees are binary to depth $10$ with $1{,}024$ leaves each, observed at one of the times $5$ to $10$. The evaluated noise scales scan across $\sigma\in\{0.2,0.4,0.6,0.8\}$.

The context-scaling evaluation predicts the transition from time $9$ to time $10$ given $k\in\{1,2,4,8\}$ related cells as input. A clone of size $k$ shares a common ancestor at time $9-\log_2 k$, so larger $k$ means deeper shared ancestry. Results are averaged over five held-out process realizations.

\paragraph{Semi-synthetic construction.}
The two remaining benchmarks move from purely simulated states to real cell states synthetically organized in a lineage tree. This is done by constructing lineage trees out of small repeated cell type motifs \cite{Tran2024-eu}.

\paragraph{Random cell type motif benchmark.}

Our first semi-synthetic dataset consists of small motifs made up of random cell types. This can be thought of as taking a cell type level map (similar to that shown in \textbf{Fig.~\ref{fig:f3_summary}a}) and randomly wiring cell types together, then using this to simulate a developmental process. Not all random motifs will be reasonable (as shown \textbf{Fig.~\ref{fig:synthetic}c}) but this experiment reveals what kind of dynamics a model is able to recover.

We do not consider all cell types in our construction of cell type motifs. Instead, we restrict to ectodermal cell types and restrict to abundant cell types: the ten most abundant eligible E9.0 progenitor types and the ten most abundant eligible E9.5 child types. In this semi-synthetic experiment, motifs consist of two progenitor cell types and four descendant cell types. This gives a possible landscape $\binom{10}{2}\binom{10}{2}\binom{8}{2}=56{,}700$ motif candidates. From this landscape, we randomly sample 32 to test.

We label candidates as Felsenstein-like or not, and OT-like or not, depending on the nature of the ancestor-descendant coupling. We measure how OT-like a motif is by comparing its ancestor-descendant pairings with minimum-displacement alternatives, given the marginals. To do this, we consider all six possible cell type motifs consistent with the marginals, and compute the displacement between ancestor and descendant cell type means. The displacement advantage of the true coupling over the other five determines how OT-like it is. Similarly, we measure how Felsenstein-like a motif is using the means of the pairs of sister cell types (among descendants). Concretely, we consider the two possible couplings consistent with the tree topology, and compute the displacement between each ancestor cell type mean and its sister-pair mean. The displacement advantage over the swapped pairing determines how Felsenstein-like it is. We classify the upper sixth of candidates as OT-like and the upper half as Felsenstein-like by these respective advantages. In the sampled motifs, every OT-like motif is also Felsenstein-like, so the two classifications sort the motifs into three categories: OT-like, Felsenstein-like but not OT-like, and neither.

\paragraph{Cell-type-motif training and evaluation.}
For each motif we draw $200$ trees from the corresponding cell-type pools. Each seven-node tree contains an E8.5 root, one cell from each of the two E9.0 progenitor types, and one cell from each of the four E9.5 child types, connected according to the imposed assignment. The root type is the most abundant eligible E8.5 subtype. States for each cell-type and time slot are sampled independently, without replacement when the pool contains at least $200$ cells and with replacement otherwise. Each tree record has the E9.5 children as its observed snapshot. Separate singleton records supply $200$ E8.5 root observations and $200$ E9.0 observations, the latter alternating between the two progenitor types and hence containing $100$ cells of each.

Every method is trained separately for each evaluated motif. The imputer and sampler are trained jointly on the three-time motif records for $2{,}000$ steps. Each baseline is trained for $3{,}000$ steps on the E9.0-to-E9.5 transition: multimarginal flow matching and GENOT with Sinkhorn use the E9.0 and E9.5 marginals, while GENOT with fused Gromov--Wasserstein uses the lineage trees. Training and scoring use the same sampled trees, so this evaluation measures whether a fitted model reproduces the imposed coupling on the empirical populations used to fit the motif rather than generalization.

The E9.0 nodes are hidden in the training tree records, but at evaluation we clamp every method to their true states and score only forward E9.0-to-E9.5 sampling. Lineage inpainting generates the two-child subtree jointly, while the baselines instead receive each progenitor state twice and transport the two copies separately, with eight draws from each GENOT model's internal solver and one deterministic prediction from multimarginal flow matching after $20$ integration steps. The score discards sibling pairing and pools predictions across trees by progenitor type, so it measures recovery of the equally weighted mixture of the two child populations assigned to each progenitor.

Each progenitor-specific target population contains $400$ true child states. We classify each prediction by its nearest E9.5 child-type centroid and compare its type proportions with those of the assigned two-type mixture using the accuracy of Section~\ref{sec:metrics}. \textbf{Fig.~\ref{fig:synthetic}c} shows the $32$ accuracies, with one column per motif and the columns ordered by hierarchical clustering of the four methods' accuracy profiles; values at or below $0.5$ share the lower display color.

\paragraph{Semi-synthetic context benchmark.}
The benchmark of \textbf{Fig.~\ref{fig:synthetic}f} builds trees from real cell states in which a progenitor's type is recoverable only from its cousins. It chooses motifs from cell types with at least $50$ cells at each corresponding time. Each sampled motif consists of four distinct types: a root type at E8.5, two progenitor types $A$ and $C$ present at both E9.0 and E9.5, and a shared child type $B$ at E9.5. A tree is a root with two progenitors of the same type, $A$ or $C$, each carrying two children, and exactly one of the two progenitors differentiates, so that its children are both of type $B$ while the other's keep the progenitor's type. A $\{B,B\}$ pair therefore identifies the type of its progenitor only through its cousin pair. The observed snapshots are at E9.5 and E9.0.

For each draw we build $200$ trees, half with progenitors of type $A$ and half of type $C$, alternating which of the two sister progenitors differentiates, and sample every node's state from its type's pool as in the motif benchmark; $200$ singleton records at E8.5 and $200$ at E9.0, the latter alternating between the two progenitor types, supply the marginals.

At test-time the imputer and sampler are given the four children of each tree and impute both hidden progenitors at once, with eight draws and $20$ sampling steps, while a baseline pushes each child back to E9.0 on its own and is scored against that child's own progenitor (Section~\ref{sec:metrics}). Predictions are averaged over draws before scoring. Progenitor-type accuracy is the fraction of imputed ancestors whose nearest progenitor-type centroid, over the E9.0 pools of the draw, is the correct one, averaged over the two types so that chance is $0.5$.  We report it over all progenitors and split by whether a transition is within the same type or between two types, averaged over the ten draws.

\subsubsection{Evaluation metrics}
\label{sec:metrics}
We evaluate models fit to the lineage-traced atlas with six metrics, shown in \textbf{Fig.~\ref{fig:f3_summary}b}. Lower is better throughout. We will next describe these metrics and their concrete implementation.

Three metrics involve making some assumptions about developmental dynamics. \emph{Germ-layer crossing rate} assumes that cell state transitions do not cross germ layers (with the exception of epiblast). \emph{Lineage-oscillation rate} assumes that developmental trajectories do not often alternate between lineage-level cell types. \emph{Hox-code distortion} assumes that in certain cell types individual Hox expression persists between E9.0 and E9.5. Three data-based metrics compare the model's imputed cells to the real data at other times. \emph{Marginal distortion} asks how well models recapitulate cell state distribution snapshots. \emph{Clonal distortion} measures the cycle-consistency of a model: given a clade, the extent to which it can be recovered by simulating reverse in time and then forward in time. \emph{Linkage distortion} measures the extent to which lineage relationships between cell types is consistent in simulated data and real measurements. We define all six of these metrics in detail below.

On the synthetic topology and random-motif benchmarks, we classify generated states by their nearest latent attractor or child-type centroid. We compare the predicted and true fate or type proportions, $\widehat p_f$ and $p_f$, and report $1-\tfrac12\sum_f|\widehat p_f-p_f|$. We average this accuracy over source parents in the topology benchmark and over the two progenitor groups in the random-motif benchmark. The synthetic context benchmark classifies individual draws by their nearest fate attractor; the real-cell context benchmark classifies the mean of eight draws by its nearest progenitor-type centroid. \textbf{Supplemental Fig.~\ref{fig:supp_synthetic}d} also compares accuracy with energy distance, Gaussian MMD, and Sinkhorn divergence on the full-dimensional states \cite{Feydy2018-ow}.

\label{supp:metrics}
The metrics above, and those of the implementation grid (Section~\ref{supp:grid}), are defined in full below. The label-based metrics use the same inferred ancestor--descendant couplings that underlie the Sankey maps of Section~\ref{sec:sankey-transition-maps}, without any nearest-neighbor distance cutoff. The continuous distortion metrics compare generated states directly.

\paragraph{Germ-layer crossing rate.}
To score germ-layer crossings we take each cell's germ layer from the annotation table. Because not every node carries a label, we condense the tree: from each labeled node we descend through unlabeled intermediates to the nearest labeled descendant, forming one source--descendant edge. The metric is the fraction of these edges that cross a germ-layer boundary, excluding edges that start from the epiblast, pooled over eligible edges.

\paragraph{Lineage-oscillation rate.}
This metric catches a cell-type label that reappears after having seemingly been left behind, a biologically implausible flip back to an ancestral identity. Walking each tree from the root, we track the last two lineage-level cell-type labels along the current path; an internal node takes the label of the real cell it is matched to, and unlabeled nodes carry none. Each labeled node with two labeled ancestors closes a triple (grandparent, parent, node), and the triple \emph{oscillates} when the grandparent and node share a label that the parent does not, an $A \to B \to A$ flip along a single path. The oscillation rate is the fraction of eligible labeled-node triples that oscillate; lower is better.

\paragraph{Hox-code distortion.}
The Hox-code distortion metric tests persistence of individual Hox expression in cell types with stable expression patterns during E9.0 to E9.5. We curate twenty such (gene, cell type) pairs, provided in \textbf{Tab.~\ref{tab:hox_pairs}}, based on prior literature \cite{Alexander2009-hox, Tumpel2009-cw, Parrish2009-yn, Minoux2010-ji, Philippidou2013-sq}.

For each pair we map the E9.5 leaves of the type to their inferred ancestor nearest to E9.0 (at most 0.4 days away). We then compute the Pearson correlation $\rho_{\rm{Hox}}$ between the gene's expression in the observed cell assigned to each ancestor and in its descendant leaf. The Hox distortion averages $1 - \rho_{\rm{Hox}}$ over pairs.

\begin{table}[t]
\centering
\caption{Cell types with stable Hox gene expression between E9.0 and E9.5}
\label{tab:hox_pairs}
\begin{tabular}{ll}
\toprule
Cell type & Hox genes \\
\midrule
Hindbrain (r2--r4) & \texttt{Hoxa2}, \texttt{Hoxb2}, \texttt{Hoxb1} \\
Cranial motor neuron & \texttt{Hoxb1} \\
Cranial neural crest & \texttt{Hoxa2} \\
Pharyngeal arch mesenchyme & \texttt{Hoxa2}, \texttt{Hoxb2}, \texttt{Hoxa3}, \texttt{Hoxd3}, \texttt{Hoxd4} \\
Hindbrain (r5--r8) & \texttt{Hoxa3}, \texttt{Hoxd3}, \texttt{Hoxb4}, \texttt{Hoxd4} \\
Spinal cord (cervical) & \texttt{Hoxb4}, \texttt{Hoxd4} \\
Spinal cord interneuron & \texttt{Hoxb4}, \texttt{Hoxd4} \\
Somatic motor neuron & \texttt{Hoxb4}, \texttt{Hoxd4} \\
\bottomrule
\end{tabular}
\end{table}

\paragraph{Marginal distortion.}
For each adjacent pair of observed times we draw \(400\) cells observed at the later time, within a quarter day of it, take each one step back to its inferred ancestor, and compute the energy distance between those ancestors and \(400\) true cells at the earlier time; the reported value averages over the pairs of times.

\paragraph{Clonal distortion.}
We draw \(400\) subtrees, each an observed leaf together with its ancestor one embryo day earlier, within half a day, and the leaves below that ancestor, keeping subtrees with at least nine leaves. The model clamps the leaves and infers the ancestor, then clamps that ancestor alone and regenerates the leaves, with eight sampling steps in each direction; the baselines take each leaf one step back once and push the result forward five times. The metric is the energy distance between the regenerated and the measured leaves, averaged over subtrees.

\paragraph{Linkage distortion.}
Using lineage-level cell-type labels, we compute the fate separation (or ancestral linkage) matrix \cite{colgan, Fang2022-tg}. This matrix describes how closely related two cell types are on the lineage tree. We take observed trees at E9.5 and E9.0, and use the models to infer progenitor states 0.5 days earlier (E9.0 and E8.5 respectively). We compute one ancestral linkage matrix per replicate embryo and compare the inferred matrices with those of observed trees at the earlier timepoint. The linkage distortion is $\lvert 1-\overline\rho_{\mathrm{inferred,real}}/\overline\rho_{\mathrm{real,real}}\rvert$, where the numerator is the mean Pearson correlation between inferred and observed matrices, and the denominator is the mean correlation between observed matrices from different replicates.

\paragraph{Cell type label composition total-variation gap.}
At each scored time we form the label composition (normalized) of the generated population (the fraction of cells carrying each label), \(Q_\ell\), and compare it to the real composition \(P_\ell\) by total-variation distance, \(\tfrac12\sum_\ell\lvert P_\ell-Q_\ell\rvert\) over the union of labels, assigning zero probability to absent labels. We report this separately for cell-type and for lineage-level cell-type labels, averaged over times.

\subsubsection{Baselines}
\label{sec:baselines}
We compare against generative models of developmental dynamics that span how much lineage structure they use. Our multimarginal flow-matching baseline~\cite{Lipman2022-qc} learns from independent samples of consecutive time marginals and ignores the tree entirely. Its deterministic flow maps each input state to a single output. GENOT with Sinkhorn couplings~\cite{Klein2024-mk} learns the same kind of transition but couples consecutive snapshots by entropic optimal transport in state space, again without reference to the tree. GENOT with fused Gromov--Wasserstein couplings~\cite{Klein2024-mk,Lange2024-ts} fuses the state-space cost with a tree-distance cost, so the coupling it fits does use lineage structure, though only through this pairwise cost rather than as context the model attends over. To make this the fairest test of the optimal-transport approach, we scale the fused coupling to the full lineage-traced atlas---which a dense solve cannot achieve---by solving it clone by clone, in minibatches drawn from one source lineage and one target lineage, with a constant-time tree-distance oracle (GENOT-cbFGW, detailed below). On the synthetic benchmarks, where the trees are small, the fused coupling is solved on randomly paired trees (Section~\ref{sec:synthetic}). All three are trained on the same atlas and time grid as the lineage inpainting model, with optimal-transport couplings fit between adjacent grid times in both directions.

\paragraph{The optimal-transport couplings.}
\label{supp:baselines}
Here we describe our implementations of two optimal transport-based baselines. First is a gene-expression based model which gives a generative model consistent with entropic OT couplings, (GENOT-Sinkhorn) and a similar method using a fused Gromov--Wasserstein coupling (GENOT-FGW)~\cite{Klein2024-mk,Lange2024-ts}, which additionally encodes lineage tree structure through the Gromov term. Both choose the cross-time correspondence by a transport cost instead of learning it. A standard implementation of GENOT-FGW does not scale to massive datasets, so below, we provide a modified implementation for the lineage-traced atlas. We refer to this method as being clonally batched (GENOT-cbFGW).

\paragraph{Scaling GENOT to million-cell atlases.}

GENOT-Sinkhorn uses entropic optimal transport to first couple cell states in the scVI embedding, then fits a flow-matching model to sample transitions with the coupling \cite{Klein2024-mk}. This approach can be minibatched, and easily scales to very large datasets. However, it does not consider the lineage tree information.

The fused Gromov--Wasserstein coupling instead determines a coupling using two cost terms solved together as one entropic fused-GW problem: a displacement term (similar to the Sinkhorn described above), and a structure term that compares two intra-domain matrices of pairwise tree distances, one among the source cells and one among the target cells. We weight the two terms equally.

Naively running the fused-GW solver on the lineage traced atlas is not feasible. A dense fused-GW solver forms the full intra-lineage distance matrices, quadratic in the number of cells, which caps it at tens of thousands of cells. Minibatching is the natural fix, but this is not straightforward. The intra-domain cost compares the tree distance within a pair of source cells against the tree distance within a pair of target cells, and a tree distance exists only within one lineage tree, so a minibatch drawn at random across the atlas leaves the cross-clone entries undefined. We instead minibatch by lineage tree: each iteration pairs one source lineage with one target lineage, solves the fused-GW problem on cells sampled from each, so that both intra-domain matrices hold within-lineage distances, and feeds source--target pairs sampled from the coupling to one flow-matching step. The per-iteration cost is then fixed by the minibatch size rather than by the size of the atlas.

Even with this adjustment, another scalability issue remains: the minibatch still needs the pairwise tree distances among its sampled cells, and recomputing paths for every minibatch would reintroduce the dependence on lineage tree size. We instead precompute, once per lineage of \(N\) cells, with a constant-time tree-distance oracle using an \(\mathcal{O}(N\log N)\) preprocessing by a standard Euler-tour and sparse-table range-minimum construction. Each minibatch's distance matrix is then filled by a fixed number of constant-time queries, the number of cells enters only the one-time oracle building step, and the coupling can run at the scale of the atlas, over a million cells.

\paragraph{GENOT-Sinkhorn.}
Each training iteration samples a random adjacent transition between consecutive observation times (i.e., every 0.5 days) in either direction and draws a minibatch of $512$ source and $512$ target cells, then solves an entropic optimal transport problem between them in the scVI embedding with a squared-Euclidean cost and regularization \(\varepsilon=5\times10^{-2}\). We normalize the coupling to a distribution over pairs and sample \(512\) source-to-target pairs from it. We then take a single flow-matching step, conditioning the velocity field on the source state and the two endpoint times. The velocity field model uses a multilayer perceptron with three hidden layers of width $1024$. We train for \(5{,}000\) iterations with Adam at learning rate \(10^{-4}\), and at inference transport each source cell with \(5\) noise draws.

\paragraph{GENOT-cbFGW.}
The fused coupling uses an inner Sinkhorn solver at the same regularization \(\varepsilon=5\times10^{-2}\), with the displacement and structure terms weighted equally. Following the lineage-minibatching construction above, each iteration pairs one source lineage with one target lineage and solves the fused-GW problem on \(512\) cells sampled from each, with the within-lineage distance matrices filled by the constant-time oracle. The network and the remaining training hyperparameters match the Sinkhorn variant. On the synthetic benchmarks, where every tree is small, the fused coupling is solved once on a random pairing of trees instead (Section~\ref{sec:synthetic}).

\subsection{Data and preprocessing}
\label{sec:data}

\subsubsection{The lineage-traced atlas}
\label{sec:atlas}
We model a mouse-embryo atlas of roughly four million cells spanning embryonic day E7.5 to E10.0. About \(1.4\) million of these are lineage-traced cells from \citet{colgan}. The remainder are non-lineage-traced cells, both host cells from chimeric embryos in the same study and cells from a separate scRNA-seq atlas \cite{Qiu2024-cq}.

The trees are reconstructed with branch length (i.e. node time) estimates in \cite{colgan} and are considered here as observed: each lineage record is a rooted tree \(\calT\) in which every node already carries a continuous developmental time, derived from molecular clock-based branch length estimation \cite{colgan}. Because sequencing is destructive, a tree is measured at a single capture time, so its observed cells are exactly its leaves, the only nodes with a measured state \(x_v\). The earlier ancestors are the unobserved nodes of the tree the model completes; the cells that would have followed lie beyond the capture time, and every analysis that samples descendants does so on the topology of a tree observed at the later time (Section~\ref{sec:diffusion}).

We summarize a cell's state \(x_v\in\R^{50}\) by the \(50\)-dimensional scVI embedding of its transcriptome~\cite{Lopez2018}. We fit one scVI model to raw counts across the whole atlas with the data source (lineage-tracing study versus reference atlas) as a categorical covariate, so the embedding is shared and batch-corrected across the two sources; gene selection is described below and the model configuration is given in Section~\ref{supp:scvi}.

Cells are annotated hierarchically at four nested levels, from finest to coarsest: cell subtype, cell type, lineage, and germ layer. Observed leaves keep their atlas annotations; an imputed ancestor carries the annotation of the observed cell it is coupled to (Section~\ref{sec:sankey-transition-maps}), and cells of the untraced atlas are annotated by nearest-neighbor transfer from the traced cells (below). This hierarchy supplies the germ-layer and cell-type labels used by the evaluation metrics (Section~\ref{sec:metrics}).

\label{supp:data}
\paragraph{Gene selection.}

We start from the genes called highly variable in the lineage-traced data that are present in both data sources. To remove batch-related features that merely separate the two data sources, we normalize each cell to \(10^4\) counts, take logarithms, and remove the top \(250\) differentially expressed genes per source (Wilcoxon test between sources), leaving \(3{,}956\) genes.

\paragraph{Annotation transfer.}
Cells of the untraced atlas are labeled by an exact \(L_2\) nearest-neighbor search in the \(50\)-dimensional embedding against the annotated traced cells, using the \(k=15\) nearest neighbors. At each of the four label levels (cell subtype, cell type, lineage, and germ layer), we take a majority vote and call the level confident when the top label holds at least a \(0.7\) fraction of the neighbors; the transferred annotation is the finest confident level, or ``ambiguous'' when no level is confident. About 1\% of untraced-atlas cells (19{,}632 of 2{,}054{,}584) remain ambiguous.

\subsubsection{\texttt{Tbx5} crispant embryos}
\label{sec:crispant-data}
\paragraph{Animal strain details} 
CAST/B6 F1 males for ICSI were generated in-house by breeding C57BL/6J strain male mice with CAST/EiJ strain females. We used 8-week-old male CAST/B6 F1 as sperm donors for ICSI and 6--12-week-old B6D2F1/J (BDF1) female mice to provide oocytes (Janvier). 28--40g female CD-1 mice (Janvier) were used as pseudopregnant fosters after mating with CD-1 strain vasectomized males (8 weeks -1.5 years old). All procedures were performed in specialized specific-pathogen-free facilities and followed all relevant animal welfare guidelines and regulations. Animals were kept on a 12h light-dark cycle and provided a standard diet \textit{ad libitum}. Embryos were collected without considering sex. Interventional procedures necessary to conduct this work were approved by the Yale University Committee on the Use of Animals in Research and Teaching (IACUC, Protocol $\#2026-20357$) or by the LaGeSo Berlin (Protocol \#G0098/23).

\paragraph{Generation of Tbx5 Crispant F0 Embryos} 
Tbx5 mutant embryos were generated as previously described \cite{Grosswendt2020-se, Kijima2026-fd}. Briefly, B6D2F1 (BDF1) strain female mice (Janvier) were serially primed over a 46 hr interval with 5 IU CARD HyperOva FD (Cat. \#KYD-015-EX, Cosmos Bio) and 5 IU of human chorionic gonadotropin (Cat. \#HOR-250, Prospect Protein Specialists), followed by isolation of oocytes the following morning. Oocytes were then fertilized with CAST/B6 F1 hybrid sperm via intracytoplasmic sperm injection (ICSI), followed by cytoplasmic injection with 200 ng/$\mu$L Cas9 mRNA and a 100 ng/$\mu$L cocktail of 3x sgRNAs (protospacer 1, 5’-GTGGCTGAAGTTCCACGAAG-3’, exon 3; protospacer 2, 5’-GCTTTCACGATGTGTAATCG-3’, exon 6; protospacer 3, GACGGTACAGCTGCTATAGGA, exon 9) generated by in vitro transcription using either the MEGAshortscript T7 Transcription Kit (Cat. \#AM1354, Invitrogen) or HiScribe T7 High Yield RNA Synthesis Kit (Cat. \#E2050S, NEB). Embryos were then cultured to E3.0-E3.5 in 25 $\mu$L drops of advanced KSOM (Cat. \#MR-101-D, Millipore sigma) under embryo-tested mineral oil. High grade morulae and/or cavitated blastocysts were transferred into the uterine horns of pseudopregnant CD-1 strain females (Janvier) and collected at E9.5 for profiling. Embryos were removed from decidua and the extraembryonic tissues dissected away. The embryo clutch was then documented and processed to single cell by incubating in TrypLE (Cat. \#12604013, Invitrogen) at $37^\circ C$ with gentle pipetting, followed by centrifugation, filtration and washing in 0.2\% BSA/PBS. The final pool was then counted and volumes adjusted for a target recovery of $\sim 50{,}000$ cells per GEM well run following computational demultiplexing of CAST SNPs. Cells were then subjected to single-cell RNA sequencing (10x Genomics, Chromium GEM-X Single Cell 3' v4). Single cell libraries were prepared following the protocol and sequenced on the Aviti system (Element Biosciences). Successful editing was validated using the sequencing reads and the examination of editing sites following demultiplexing.

\paragraph{Data preprocessing} 

Reads were preprocessed using Cell Ranger (v9.0.1) and aligned to the mouse GRCm39 assembly and quantified using GENCODE vM37 gene annotations. Cells with at least $750$ counts, $400$ detected genes and at most $10\%$ mitochondrial counts were then integrated with $250{,}000$ E9.5 cells from the lineage-traced study, thinned to the crispants' sequencing depth, by an scVI model with the data source as batch covariate, two hidden layers of width $512$ and a $50$-dimensional latent space, trained for $2{,}000$ steps in batches of $2{,}048$ cells. Each crispant cell is annotated with the subtype voted by its $25$ nearest reference neighbors.

\paragraph{Demultiplexing}

Embryos were demultiplexed using genotypic information, similar to \cite{Grosswendt2020-se, Kijima2026-fd, Dias2026-vh}. Reads overlapping known SNPs were assigned to the B6 (reference) or CAST (alternative) allele. Each autosome was divided into $10$ contiguous bins holding equal numbers of SNP-overlapping reads. Allele-specific read counts were summed per bin for each cell, giving cell-by-bin B6 and CAST read-count matrices ($187$ non-empty bins). Using cells that passed quality control and had at least $100$ SNP-informative reads, we fit, by expectation--maximization a beta-binomial mixture model with one component per embryo and one per pair of embryos to account for doublets ($8$ and $28$ components respectively), and removed cells whose posterior probability of not being a singlet exceeded $0.05$.

Using these cells as putative doublets identified by genotype, we next identify doublets using a transcriptional state filter. We identify cells with at least $22$ of their $30$ nearest neighbors in principal-component space ($50$ components of log-normalized expression) labeled as genotype doublets, and remove these.

\subsection{Computational analysis methods}
\label{sec:analysis}

\subsubsection{Sankey transition maps}
\label{sec:sankey-transition-maps}
We summarize the cell-type level transitions learned by a model using a Sankey map. For each time window, we first fully inpaint the lineage records using the sampler of Section~\ref{sec:diffusion}. Each sampled internal state is then snapped to its nearest observed cell, allowing us to assign a cell type label. This can equivalently be thought of as a nearest-neighbor annotation of cell types. Cell type assignments whose nearest-neighbor distance exceeds a cutoff---the 80th percentile of all internal-node nearest-neighbor distances pooled across trees---are filtered out.

We then convert these completed trees into adjacent-time transition pairs. Walking from root to leaf, if a retained representative \(p\) at time \(t\) and a retained representative \(c\) at time \(t+\Delta t\) are both available, the path contributes one pair \((p,c,t,t+\Delta t)\). In the mouse transition maps, \(\Delta t=0.5\) embryo days is the natural time interval since that is how the data are sampled.

The pairs are counted by time-resolved labels \((\text{label},t)\), counted at germ-layer resolution at E7.5, lineage resolution at E8.0, and cell-type resolution from E8.5 on. An edge is retained if it has at least five supporting pairs and its shares of the target node's incoming pairs and of the source node's outgoing pairs sum to at least \(0.2\).

\textbf{Fig.~\ref{fig:f3_summary}a} is a rendering of this directed graph. The map spans the six observed times from E7.5 to E10.0 and keeps the \(90\) most abundant cell types. Node sizes scale with observed cell counts at each labeled time point, and ribbon widths are proportional to retained transition counts. The row ordering is chosen from the weighted type-type graph to reduce long crossings; it affects only the layout, not the inferred transition counts.

\subsubsection{Context scaling}
\label{sec:context-scaling}
The context-scaling evaluations of \textbf{Fig.~\ref{fig:f3_summary}c--f} ask how the accuracy of an inferred ancestor or descendant changes with the context, or size of the input clade. Both are cycle-consistency tests on subtrees of the atlas whose root falls between E7.0 and E9.0 and which carry between \(3\) and \(20\) leaves (at least \(6\) for descendant inference), with at most \(100\) focal nodes drawn per terminal cell type.

For ancestral inference we fix a focal leaf and reveal successively larger clades around it, moving up its ancestors as far as one and a half embryonic days above the leaf. At each level the clade's leaves are clamped and the model infers the subtree's root. That root alone is then clamped and every leaf of the subtree is regenerated below it, conditioned only on the root. The reconstruction quality metric is the energy distance between the regenerated and the measured leaves of the subtree. We measure context size by the time from the focal leaf back to the most recent common ancestor of the revealed leaves (grouped into quarter day bins).

For descendant inference the focal node is an internal node with at least three descendant leaves and at least one cousin. We fill every internal state of the subtree with the sampler from its clamped leaves and take the distinct sets of cells alive at \(60\) times swept through the subtree, each holding an ancestor of the focal node, at most a day above it, with some of its cousins. Clamping the imputed states of such a set and regenerating the leaves below asks how much the cousin information improves reconstruction of the focal node's descendants: the score is the energy distance between the regenerated and the measured descendants of the focal node, and context is the number of cousins in the set (up to $8$).

The GENOT baselines take no context, and as such are represented with flat lines. For \(400\) focal cells per panel, one leaf is taken back to the root's or the focal node's time with a single draw and forward again with five draws, and the regenerated cloud is scored against the same measured leaves. 

The per-type panels (\textbf{Fig.~\ref{fig:f3_summary}d,f}) show error by cell type, excluding blood and extraembryonic types, for cell types with at least \(10{,}000\) leaves and at least three evaluations at every context size. The highlighted cell types are those that gains most from context and the one that gains least, measured by the slope of a linear fit of energy distance against context.

\subsubsection{Convergence analysis}
\label{sec:convergence-analysis}
The convergence score measures the extent to which transcriptionally similar cells at a particular time have distinct inferred histories. Fix a descendant time $t_b$ and an ancestor time $t_a=t_b-\delta$. In the analyses reported in the Results, $\delta=0.5$ days and every observed time from E9.0 to E10.0 serves in turn as \(t_b\), with the cell-type comparison pooling over all timepoints. We complete each lineage record back to \(t_a\) with the sampler and represent each sampled ancestor by its nearest observed cell (Section~\ref{sec:sankey-transition-maps}).

For an observed leaf \(i\) at \(t_b\), let \(a(i)\) be the nearest-cell representative of the first ancestor on the path from \(i\) whose time is at or before \(t_a\). Leaves whose first such ancestor lies more than a quarter day before \(t_a\) are omitted. Let \(\mathcal{N}_k(i)\) be the \(k\) nearest other leaves to \(i\) at time \(t_b\) across the atlas, measured in the state embedding; the reported analyses use \(k=100\). The ancestor set around \(i\) is
\[
A_i = \{a(j): j\in \{i\}\cup \mathcal{N}_k(i)\},
\]
and the convergence score is
\[
S_i = \operatorname{tr}\operatorname{Cov}\{x_a: a\in A_i\}.
\]
We also record the cell-type composition of \(A_i\), giving a source-type mixture for each descendant cell. Because this spread, like the OT-versus-tree discrepancy below, grows wherever cells are sparsely sampled in the embedding, we normalize both by the local spacing of cells. With \(r_i\) the mean distance from \(i\) to its \(10\) nearest other cells of the same stage, we normalize the convergence score as \(S_i/r_i^2\). \textbf{Fig.~\ref{fig:f4_convergence}a} shows this score after averaging each cell with its \(50\) nearest neighbors and clipping to the 2nd and 98th percentiles.

For the GENOT comparison, the tree ancestor of \(i\) is its ancestor node whose time is closest to \(t_a\) (within half a day), and the OT-versus-tree discrepancy is the Euclidean distance between the observed cells that lineage inpainting and GENOT-Sinkhorn assign to that node, divided by \(r_i\). GENOT assigns a node by chaining its backward ancestor map from one randomly chosen descendant leaf. \textbf{Fig.~\ref{fig:f4_convergence}b} compares the two as cell-type means over the types with at least \(500\) leaves, excluding blood.

The endothelial fan-in of \textbf{Fig.~\ref{fig:f4_convergence}c} traces each E9.5 endothelial leaf to its ancestor nearest E8.5, within \(0.3\) day, and reports the cell-type composition of those ancestors, excluding progenitors that are already endothelial. The example subtrees have four to ten leaves, a single other lineage holding the majority of them and no lineage oscillating along any path (Section~\ref{sec:metrics}).

The program axes of \textbf{Fig.~\ref{fig:f4_convergence}d} are Fisher discriminant directions in the embedding: for each program a positional contrast, paraxial against lateral plate mesoderm for endothelium and anterior (hindbrain) against posterior (spinal cord) types for the neurons, and an identity contrast between the fate and either the non-endothelial mesoderm or the other neuron types, orthogonalized against the positional direction. The directions are fit on E9.5 cells of the lineage-traced atlas, at most \(3{,}000\) per cell type and embryo, centered within each embryo so that no contrast is an embryo offset, with a ridge of five percent of the mean variance on the pooled within-group covariance.

A trajectory is the path of assigned nodes from an E9.5 leaf of the fate up through its inferred ancestors. Its source is the nearest labeled strict ancestor within \(0.3\) day of E8.5; trajectories whose source lies outside the fate's germ layer, or already carries the fate, are dropped. Endothelial trajectories are grouped by source, paraxial or lateral plate; neuronal trajectories are grouped into five equal-count bins of the source's positional score. Within each group the two program scores are smoothed against time with a Gaussian kernel of bandwidth \(0.10\) day on a grid from E7.5 to E9.2, keeping grid points with an effective sample size of at least \(30\), and each curve ends at the mean of the group's E9.5 leaves (\textbf{Fig.~\ref{fig:f4_convergence}e}).

The memory analysis of \textbf{Fig.~\ref{fig:f4_convergence}f,g} compares two groups of trajectories at their sources and again at their leaves. For a fate with a fitted positional axis the groups are the two sides of that axis; for the sweep of \textbf{Fig.~\ref{fig:f4_convergence}g} they are the fate's two most abundant E8.5 source cell types within its germ layer, each with at least \(50\) trajectories and neither already carrying the fate. Per gene we average log-normalized expression over up to \(1{,}000\) distinct cells per group, keep genes with mean above \(0.1\) in at least one group, and take the absolute \(\log_2\) fold change between the groups at E8.5 and between the E9.5 leaves they reach, with a pseudocount of \(0.1\). A gene whose E8.5 fold change exceeds \(1\) is source-specific; it is a memory gene when its E9.5 fold change stays above \(0.5\) and an erased gene otherwise.

\subsubsection{Perturbation analysis}
\label{sec:perturbation-analysis}
We simulate a knockout by redistributing cells in the scVI state space, and running forward-in-time inference using the sampler. The core assumption is that removing a gene displaces individual cells but leaves the transition rules intact, so the model fit on unperturbed embryos can be applied to displaced cells. Below we describe the details of how we repurpose our models to simulate genetic perturbations.

To knock out a gene $g$ we replace a cell by the nearest observed cell that does not express $g$, and then let the model propagate that cell's descendants as it would in an unperturbed record. Every number reported is a difference between two rollouts of the same lineage records under the same model, one with the cell states redistributed and one without. It describes how the learned dynamics respond to the shifts in cell state, but does not address scenarios where perturbations shift cells to unnatural states (i.e., those not observed in the reference embryos).

Each perturbation prediction begins from trees observed at a later timepoint $T$ (e.g. E9.5). We first impute the internal nodes of the tree using the sampler, giving a reference state at every internal node. Then, starting from the internal nodes at a simulation start time (e.g., E7.5), we run forward-in-time inference with the model, restricting the model to parts of the state space where $g$ is not expressed. At the current time of the simulation, we first identify all cells which express $g$. We then project these $g$-on cells to their nearest $g$-off neighbors in the scVI state space. Then, with this remodeled snapshot, we run the lineage inpainting model forward. 

We then repeat this process at every successive time step of a pre-selected grid. We use timesteps of 0.5 days across all our simulations. The knockout is enforced at every time the simulation crosses, not only at the start: a lineage that the model pushes back into expressing $g$ is knocked out again. A node whose representative cell is moved is marked perturbed, and so is everything below it. Because the topology is held fixed, the analysis can only rearrange cells among the states and lineages seen in normal development: it says nothing about effects on proliferation, cell death, or tree topology.

The nearest-neighbor projection step can carry a node far from any cell that resembles it. In these cases, we take the view that model predictions will be unreliable. We therefore record, at every node that moved, how far it moved, relative to the local spacing of the observed cells around it. That is, the distance from its assigned cell to its own nearest neighbor. We consider a node off-manifold when it moved farther than a fixed multiple of that spacing, and a readout node is off-manifold when any node on its path is. In practice we use an off-manifold ratio of 1.1. When a node drifts off manifold, we treat the trajectory as ``frozen'' in its prior location in the state space.

The genome-wide screen knocks out each gene in turn in the E9.0 records of one replicate embryo, starting from E7.5, with 16 denoising steps. The heart screen of \textbf{Fig.~\ref{fig:f5_pert}d,e} takes the $21$ genes of a curated congenital heart disease panel~\cite{Sierant2025-lx} that are among the highly variable genes and runs a full knockout on the same records, with 64 denoising steps. Both are read out at E9.0.

For the map of \textbf{Fig.~\ref{fig:f5_pert}a} we summarize each knockout by the mean state embedding of its readout leaves minus the control's, keeping the genes with at least $200$ readout leaves, and embed those differences by UMAP with $15$ neighbors and a minimum distance of $0.1$. Each gene is colored by the lineage whose share of the readout it depletes most. The insets show, for the two transcription factors that most specifically deplete each of six lineages, the predicted nearest-neighbor distortion (defined below) at the E9.0 reference cells.

\textbf{Fig.~\ref{fig:f5_pert}b} compares the screen with curated mouse phenotypes. We take the Mouse Genome Database (MGD) annotations~\cite{Eppig2015-bc} of the screened genes and map each phenotype term to the cell type it implicates, giving one gene and cell type pair per annotation. A pair is scored by the fraction of that cell type's control leaves that are perturbed, over cell types with more than $50$ control cells. The background draws $10{,}000$ pairs, matched to the annotated pairs' cell-type distribution, from $40$ random genes carrying no annotation, and coarse anatomical contexts with more than $100$ pairs are shown separately. For \textbf{Fig.~\ref{fig:f5_pert}c} each screened gene with a human ortholog takes the minimum, over the ortholog's canonical protein-coding transcripts, of the loss-of-function observed-over-expected upper bound fraction~\cite{Lek2016-wm, Karczewski2020-ki}. Its predicted effect is the $L_2$ norm over cell types of the $\log_2$ ratio of knockout to control readout counts, with a pseudocount of one, and the $100$ most and $100$ least constrained genes are compared by a Mann--Whitney test.

The heart panels score each gene by how much of the E9.0 readout its knockout moved. The heart share is the percentage by which the knockout reduces the on-manifold readout cells of the heart cell types relative to the control, the total share is the same quantity over all readout cells, and the specificity is the $\log_2$ ratio of the two, each share floored at one percent. Timing uses the shared topology: a readout cell is shifted when the knockout assigned it a different cell than the control did, or took it off the manifold, and its onset is the time of the most ancestral node on its path that is still shifted. A gene's shift time is the mean onset over its shifted readout cells.

\paragraph{Nearest-neighbor distortion}
\label{sec:nn-distortion}
The nearest-neighbor distortion of a reference cell is the log ratio of its distance to the nearest cell of a perturbed embryo to its distance to the nearest cell of another reference embryo. The two draws have equal size and the ratio is averaged over five draws and centered on the scope's median. The predicted distortion draws the perturbed cells from the simulated knockout of a third embryo, neither the reference cell's own nor the denominator's. The measured distortion draws them from a perturbed embryo (Section~\ref{sec:crispant-data}), and each perturbed embryo is scored against the same reference denominator.

For \textbf{Fig.~\ref{fig:f5_pert}a}, distances are taken in the scVI state space within the E9.0 reference cells outside the blood and extraembryonic endoderm, and an off-manifold leaf of the genome-wide screen is dropped. The \texttt{Tbx5} comparisons take distances in the integrated latent space of Section~\ref{sec:crispant-data}, within the cells of the heart tube (first and second heart field, endocardial cells and epicardium). The \texttt{Tbx5} analyses use a ratio of $1.1$ and hold an off-manifold trajectory at its last on-manifold state. \textbf{Fig.~\ref{fig:f5_pert}l} visualizes both measured and predicted distortions, and \textbf{Fig.~\ref{fig:f5_pert}m} reports their Pearson correlation over its reference cells.

\subsubsection{Tbx5 trajectory and program analyses}
\label{sec:tbx5-analysis}
The \texttt{Tbx5} analyses of \textbf{Fig.~\ref{fig:f5_pert}f--i} use a knockout screen run as in Section~\ref{sec:perturbation-analysis} on the E9.5 trees of three replicate embryos, started with cells alive at E8.0 with 16 denoising steps, the off-manifold ratio \(1.1\) and off-manifold leaves frozen in their last on-manifold states.  The shift curves of \textbf{Fig.~\ref{fig:f5_pert}f} give, on a grid from E7.0 to E9.5 in steps of \(0.05\) day, the share of a subtype's leaves whose trajectory is shifted by that time, for subtypes with at least \(100\) readout leaves, naming the seven most shifted. Note that the simulation begins with cells alive at 8.0, which includes nodes with earlier node times which have child nodes with time $>8.0$.

A node expresses \texttt{Tbx5} when the log-normalized expression of its assigned cell exceeds \(0.10\). The example subtree of \textbf{Fig.~\ref{fig:f5_pert}g} is one founder's subtree of a control record, pruned to the forelimb field and heart leaves and the internal nodes above them; among the subtrees with at least three leaves of each fate, at most \(55\) leaves after pruning and a root below the expression threshold, we draw the one whose two fates are largest and most evenly balanced. For \textbf{Fig.~\ref{fig:f5_pert}h} we take every founder's subtree that produces both fates and record, along each leaf's path, the split, the time of the last ancestor whose descendants include both fates, and the onset, the first time \texttt{Tbx5} is expressed on the path; the panel shows for each fate the cumulative share of its trajectories past each event.

The posterior second heart field program of \textbf{Fig.~\ref{fig:f5_pert}i} is a linear discriminant direction, fit with a shrinkage estimate of the within-class covariance, between the lineage-traced cells carrying the atlas's anterior and posterior second heart field labels at the stages that carry both, E8.0 to E9.0, so no cell the panel scores enters the fit. Its features are the log-normalized expression of the non-mitochondrial genes detected in at least \(5\%\) of the fit cells, standardized on those cells. The E9.5 heart-tube reference cells of the first and second heart fields are standardized the same way, projected on the direction, and the projection standardized over them. The three transcription factors with the largest loadings on each side name the axis, and the panel plots the score against the predicted nearest-neighbor distortion of the knockout screen at each reference cell (Section~\ref{sec:nn-distortion}).

\beginsupplement
\renewcommand{\thesection}{S\arabic{section}}
\setcounter{algorithm}{0}
\setcounter{proposition}{0}
\setcounter{lemma}{0}
\setcounter{assumption}{0}
\renewcommand{\thealgorithm}{S\arabic{algorithm}}
\def\theHalgorithm{S\arabic{algorithm}}
\renewcommand{\theproposition}{S\arabic{proposition}}
\renewcommand{\thelemma}{S\arabic{lemma}}
\renewcommand{\theassumption}{S\arabic{assumption}}

\supplementtitle
\supplementtoc

\clearpage

\section{Supplementary results: supporting analyses}
\label{supp:supporting}
\subsection{Additional results on synthetic and semi-synthetic benchmarks}

\paragraph{Decomposing error across transitions.} 
In \textbf{Fig.~\ref{fig:synthetic}a,b}, we show that lineage inpainting models are able to leverage tree topology to correctly infer transitions. The metrics we show are only for the transition where differences between models are very distinct ($7 \to 8$). Here we show the rest of the transitions (\textbf{Supplemental Fig.~\ref{fig:supp_synthetic}a}).

Decomposing performance across transitions points toward the source of multimarginal FM's poor performance. Because this multimarginal FM baseline fits a deterministic cell-state ODE (akin to \cite{Qiu2022-ms, Tong2020-ix}), it is unable to model bifurcations. So at transitions where the output is unimodal from a certain state ($6\to7, \ 8\to9$), multimarginal FM performs well. But in the other transitions where there is a bifurcation ($5\to6, 7\to8$) it performs poorly.

\paragraph{Additional baselines (including non-generative approaches).} 
Throughout the main text, we benchmark only against generative baselines: multimarginal FM, GENOT-Sinkhorn and GENOT-FGW. Here, we also compare to two established non-generative baselines. In order to compare these to generative baselines, we use a nearest neighbor projection. To get the ancestor (or similarly, descendant) of an input cell, we identify the closest cell in the training dataset, and output its discretely coupled ancestor. In a similar spirit, we also compare to lineage inpainting implemented with two ``non-parametric'' imputation models. The first is mean pooling, or in other words the ``vanilla'' ancestral reconstruction which is used as the additive baseline in the tree-structured transformer. The second is what we refer to as NN-snapped imputer, which replaces the mean-pooled ancestral reconstruction with its nearest observed cell.

Including these baselines in the semi-synthetic experiments, we find that all four baselines perform more poorly than lineage inpainting with a learnable imputation function across both semi-synthetic benchmarks. WaddingtonOT \cite{Schiebinger2019-kc} and moslin \cite{Lange2024-ts} exhibit stronger performance for recovering the correct cell type coupling over their generative counterparts. This suggests that the poor performance of the GENOT models is not always attributable to the optimal transport coupling itself.

This has a clear explanation in the context-copying benchmark. GENOT models are fundamentally incapable of recovering this, as they never receive the cousin context (which determines the ancestral state) as input. Even provided with the ground truth coupling, a flow matching model fit to generate ancestral states conditional on a single descendant state will be unable to correctly match the coupling, since the true state depends on the state of the cell's cousins.

\paragraph{Consistency across performance metrics.} 
Throughout our synthetic evaluations, we use performance metrics which evaluate the extent to which generated states match the discrete clusters of the true states. In the semi-synthetic experiments, these correspond to cell types, and in the synthetic experiments, these correspond to the attractors of the Ornstein-Uhlenbeck equations. Here, in \textbf{Supplemental Fig.~\ref{fig:supp_synthetic}d}, we compare the fate accuracy to distributional divergences which do not discretize the state space. We find that broadly the two metrics are in agreement ($\rho \approx -0.63$ across multiple divergences). The non-parametric imputation models (mean pooling and nearest-neighbor snapping) are outliers: they can achieve strong fate accuracy without closely matching the correct distribution in state space. This is understandable: the hierarchical means can be close to the correct cell type, but far from any realistic cell.

\begin{figure}[tp]
  \centering
  \includegraphics[width=\linewidth]{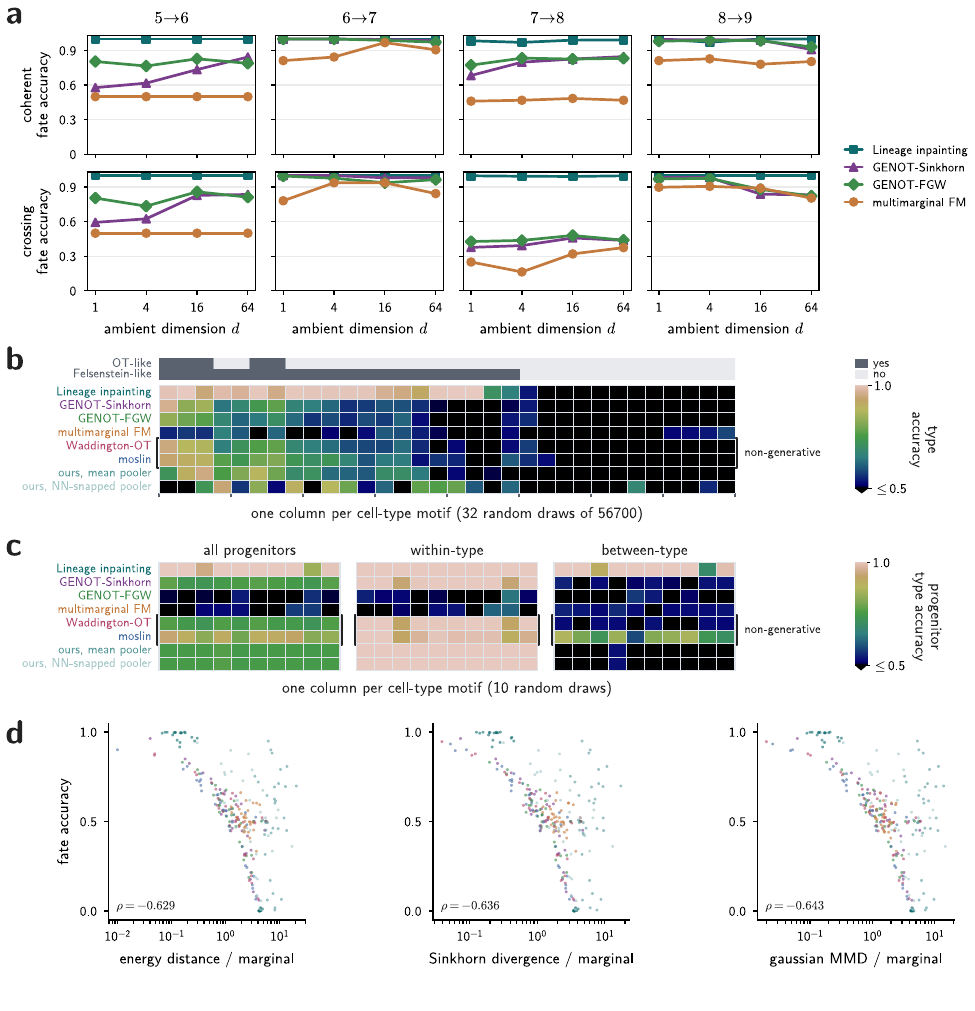}
  \caption{\textbf{Additional results on synthetic and semi-synthetic benchmarks.}
  \textbf{(a)} Fate accuracy of reconstructions across coherent and crossing dynamics from \textbf{Fig.~\ref{fig:synthetic}}, decomposed by transition.
  \textbf{(b)} Type accuracy of model predictions from random cell type motif benchmark of \textbf{Fig.~\ref{fig:synthetic}}, including additional baseline methods. WaddingtonOT and moslin are non-generative methods, here we evaluate using the discrete coupling.
  \textbf{(c)} Same as (b), but for the cell-type context motifs.
  \textbf{(d)} Comparison of fate accuracy metrics with distributional metrics, on the random cell type motif benchmarks. MMD refers to maximum mean discrepancy. Colors match panel (c).
  }
  \tocfigure{Additional results on synthetic and semi-synthetic benchmarks}
  \label{fig:supp_synthetic}
\end{figure}

\subsection{Inspecting the attention heads of a lineage inpainting model}

In the architecture introduced in \textbf{Box~\ref{box:attention}}, there is a parameter $\gamma$, which encodes (per head) the locality of the attention matrix. For negative $\gamma$, nodes attend more to their local neighborhood on the tree, while for positive $\gamma$ nodes attend more to distant relatives. After fitting a model on real data, we inspected the learned $\gamma$ values (\textbf{Supplemental Fig.~\ref{fig:supp_attn}a}). We find that there are both negative and positive $\gamma$ values, but the highest magnitude $\gamma$ (by far) is negative. We also find that the largest positive $\gamma$ values ($\gamma > 1$), which may correspond to heads processing global structure on the tree (long-range dependence), are exclusively present in the later layers of the transformer (layers 3 and 4).

When visually inspecting the attention matrices, we noticed a pattern that the model attends more to real cells (the input condition) during the beginning of the flow. This is sensible, as at early timesteps the unmasked nodes supply the observed states, while the masked nodes receive noise around the displacement state. At later time steps, attending between generated nodes is sensible to capture dependence between child states. For example in asymmetric divisions, such a mechanism would be necessary. We quantify this across heads in \textbf{Supplemental Fig.~\ref{fig:supp_attn}b}.

\begin{figure}[tp]
  \centering
   \includegraphics[width=\linewidth]{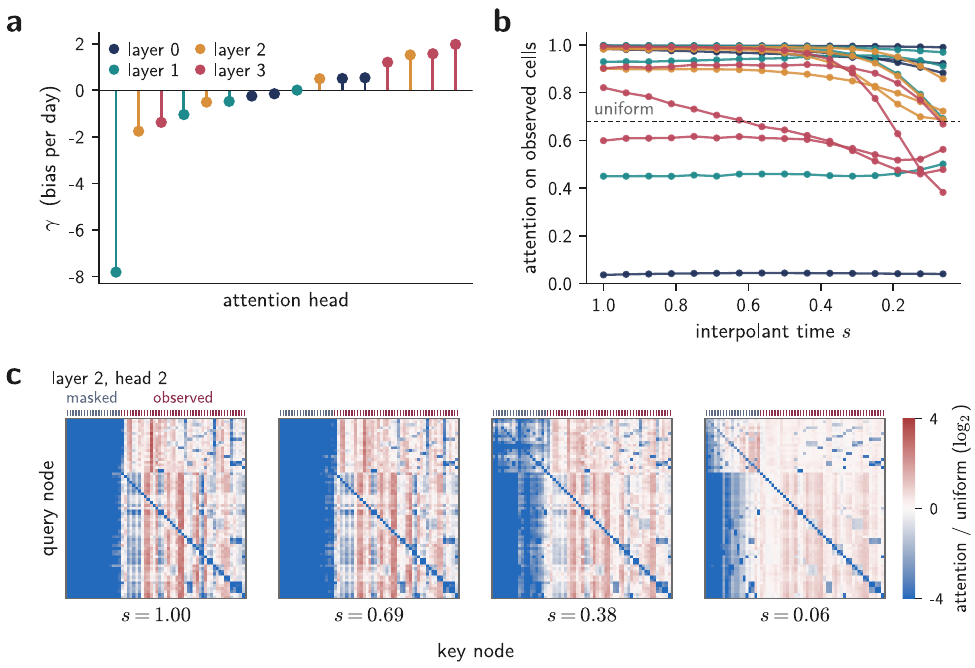}
  \caption{\textbf{Inspecting the learned attention heads.}
  \textbf{(a)} Learned patristic distance bias scale ($\gamma$), across each head in the sampler of a lineage inpainting model fit to \citet{colgan}. One dot per each of four heads across four layers.
  \textbf{(b)} Fraction of attention paid to observed (input condition) nodes, across interpolant times, across each head. Hue matches layers from panel (a). 
  \textbf{(c)} For an example input subtree, the normalized attention matrix of attention head 2 from layer 2, at different interpolant times. Columns represent key nodes, and have ticks hued by observed vs. masked.
  }
  \tocfigure{Inspecting the learned attention heads}
  \label{fig:supp_attn}
\end{figure}

\subsection{Additional results on model evaluation on the lineage-traced atlas}

All metrics in \textbf{Fig.~\ref{fig:f3_summary}b} are computed on the training data --- that is, the same trees that the models were fit using. To test whether these results hold on unseen trees, we evaluated a set of models trained on a dataset holding out 11 of the 56 clone-specific trees from training: two clones per stage from E7.5 to E9.5, each pair taken from different embryos, and one of the six E10.0 clones. Together these make up 23\% of all leaves.

We trained lineage inpainting on the remaining 45 trees, with the same settings as the model in the main text. On the same split we trained GENOT-Sinkhorn with $\varepsilon \in \{0.01, 0.05, 0.2\}$ and GENOT-cbFGW with fused penalty $\in \{0.1, 1, 10\}$. The main text uses $\varepsilon = 0.05$ and a fused penalty of 1. We then computed the same metrics of \textbf{Fig.~\ref{fig:f3_summary}b}, only over the 11 held out trees. As a point of comparison, we also scored the models from the main text, which were trained on all 56 trees, on the same held-out trees.

On held-out trees, lineage inpainting has the lowest clonal distortion, marginal distortion, lineage oscillation, germ-layer crossing and Hox-code distortion (\textbf{Supplemental Fig.~\ref{fig:supp_holdout}}).   Compared with the model trained on all trees, most metrics are similar on the held-out model, except linkage distortion, which rises from 0.13 to 0.28. This may be due to the higher variance in the metric.

Across all the GENOT parameter choices tested, model metrics remain relatively similar, with the exception of increasing the fused penalty of GENOT-cbFGW to 10. In this case, models appear to improve substantially over the main text models, but are still less performant than the lineage inpainting models. 

\begin{figure}[tp]
  \centering
  \includegraphics[width=\linewidth]{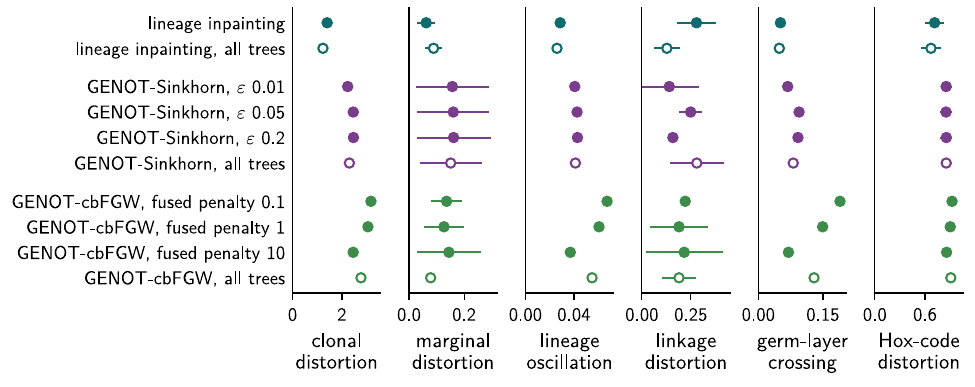}
  \caption{\textbf{Model evaluation on held-out lineage trees.}
  The metrics of \textbf{Fig.~\ref{fig:f3_summary}b}, computed over the 11 clone trees that were held out. Lower is better for every metric. Filled circles are models trained on the remaining 45 trees, while open circles are the models from the main text. Points show the mean, and bars show the 95\% interval over subtrees (clonal), stage pairs (marginal), stage transitions (linkage) or (cell type, Hox gene) pairs (Hox-code).
  }
  \tocfigure{Model evaluation on held-out lineage trees}
  \label{fig:supp_holdout}
\end{figure}

\subsection{Comparing lineage inpainting inferred transitions with prior literature}

In \citet{Qiu2024-cq}, the authors begin from a large set of $>10^3$ possible cell state transitions based on transcriptional similarity as measured using a $\sim12$ million cell scRNA-seq atlas, and manually curate these based on prior support in literature, resulting in a proposed cell type transition graph with $452$ edges. This graph remarkably recapitulates the vast majority (though presumably not all) of known developmental state transitions from zygote to birth. Here, we compare the cell state transitions learned by lineage inpainting -- purely from a single lineage-traced atlas -- to those annotated using prior literature in \citet{Qiu2024-cq}.

We first assess the extent to which the lineage inpainting model is consistent with this literature-informed map at the cell type level. Beginning from the aggregated cell type transition map in \textbf{Fig.~\ref{fig:f3_summary}a}, we ask if each of the edges also appears in the map of \citet{Qiu2024-cq}. Of the 272 total edges, 22 involve cell types with no clear correspondence to the annotations of \citet{Qiu2024-cq}. Of the remaining, 195 edges could be directly mapped to edges in \citet{Qiu2024-cq} or were ``identity'' edges between the same cell type at two timepoints. The remaining 55 edges were inconsistent between the two maps. Of the mismatches, 22 are transitions between similar cell types (which we denote ``boundary''), and could arise from labeling artifacts. For example certain cell types may not have clear boundaries between them, and rather a continuous gradient in transcriptional states. We identify the similar cell types as those that either descend from the same parent node or are classified as ``spatial continuity'' edges in \citet{Qiu2024-cq}. 

One apparent failure mode of the lineage inpainting models is that they can occasionally assign transitions in the wrong direction. For example, the model suggests that axial progenitors can arise from presomitic mesoderm, which is highly unlikely. We identify 13 such cases, where an edge exists in \citet{Qiu2024-cq}, but is found in the reverse direction in the lineage inpainting derived map. However, these incorrect ``reverse'' transitions are not particularly dramatic in propensity. he most prevalent reversed transition (hindlimb field $\to$ somatic mesoderm) appears only 893 times across $\sim 10^6$ inferred transitions.

Some mismatches fall into none of these categories. All have low propensity, appearing less than $2\times10^3$ times, and several appear plausible but unannotated in \citet{Qiu2024-cq}. Some are due to differences in resolution. For example, the top mismatch predicted by lineage inpainting is epiblast $\to$ paraxial mesoderm. This edge does not exist in \citet{Qiu2024-cq}. Instead, they have the more finely annotated path: epiblast $\to$ primitive streak $\to$ nascent mesoderm $\to$ paraxial mesoderm.

Others appear to be due to possible errors in the map of \citet{Qiu2024-cq}. For example, the second most abundant mismatch is epiblast $\to$ axial progenitor. In the map of \citet{Qiu2024-cq}, the axial progenitors are labeled as caudal lateral epiblast and neuromesodermal progenitors (CLE and NMPs), and their origin is given as surface ectoderm $\to$ CLE and NMPs. This disagrees with prior literature which suggests that axial progenitors arise from the epiblast \cite{Henrique2015-hb}.

Much of the disagreement between the two cell type level maps is concentrated in low-propensity transitions. Evaluating agreement between lineage inpainting inferred transitions and the \citet{Qiu2024-cq} map at the single-cell level, we find that 93\% of transitions are consistent (\textbf{Supplemental Fig.~\ref{fig:supp_summary}}). Of the remaining transitions, the majority come from boundary transitions, which are very sensitive to annotation choices.

\begin{figure}[tp]
  \centering
  \includegraphics[width=\linewidth]{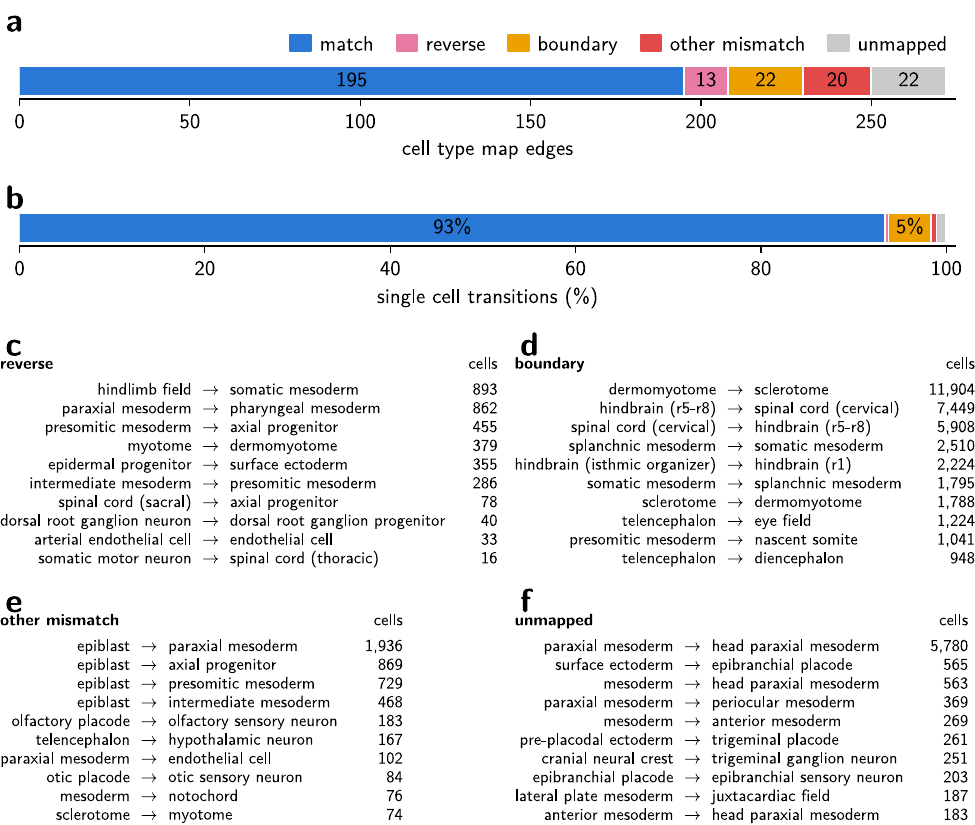}
  \caption{\textbf{Comparing lineage inpainting models with \citet{Qiu2024-cq} manually curated transitions.}
  \textbf{(a)} Agreement with \citet{Qiu2024-cq} transition map, at the aggregate level of edges between cell types. \textit{Match} denotes an edge in both graphs, including identity edges between the same cell type label at different times, \textit{reverse} denotes an edge in \citet{Qiu2024-cq} which lineage inpainting infers the opposite direction, \textit{boundary} denotes an edge between cell types which are labeled as spatially adjacent in \citet{Qiu2024-cq}, or nodes which share an ancestor in \citet{Qiu2024-cq}, \textit{unmapped} edges involve a cell type with no corresponding annotation in \citet{Qiu2024-cq}.
  \textbf{(b)} Same as (a), but now measured at the single cell transition level rather than the aggregate cell type level.
  \textbf{(c) - (f)} List of top 10 most frequent transitions from each non-match category.
  }
  \tocfigure{Comparing lineage inpainting inferred transitions with Qiu et al. 2024 curated transitions}
  \label{fig:supp_summary}
\end{figure}

\subsection{Additional results on transcriptional convergence analysis}

\paragraph{Convergence and OT distortion scores across cell types.} In \textbf{Fig.~\ref{fig:f4_convergence}b}, we quantify the convergence score and OT discrepancy averaged over cells with a particular cell type label. Here, in \textbf{Supplemental Fig.~\ref{fig:supp_conv}a}, we tabulate the cell types with the highest and lowest convergence scores. The top ten most convergent cell types are all either neurons or endothelial cells. The bottom ten cell types do not have as clear a pattern, but several are broader progenitors which are more prevalent at earlier timepoints, such as forebrain and spinal cord (see \textbf{Fig.~\ref{fig:f3_summary}a} for abundance over time).

\paragraph{Sensitivity to normalization.} When computing the OT discrepancy and convergence score in \textbf{Fig.~\ref{fig:f4_convergence}a,b}, we use normalization for the local spacing between cells. That is, for each cell, we divide its ``raw'' OT discrepancy (ancestor L2) or convergence score (covariance matrix trace over ancestors of transcriptional neighbors) by the cell's own nearest-neighbor distance. The goal is to ensure that the convergence scores and OT discrepancy are not simply capturing local transcriptional variation between cell types. Here, we show that the general conclusions of this figure do not rely on this particular normalization choice. In \textbf{Supplemental Fig.~\ref{fig:supp_conv}b}, we show that without normalization, the OT discrepancy and convergence score are more tightly correlated, and endothelium and neurons remain the most convergent cell types.

\paragraph{Core genes of transcriptional programs.} Across \textbf{Fig.~\ref{fig:f4_convergence}d,e} we illustrate program reuse in convergence events, relying on transcriptional programs defined for (1) the paraxial to lateral plate axis, (2) the anterior to posterior axis, (3) the endothelial cell identity axis, (4) the spinal cord interneuron identity axis and (5) the somatic motor neuron identity axis. We define all of these programs using the same procedure (linear discriminant analysis over the cell type labels). In \textbf{Supplemental Fig.~\ref{fig:supp_conv}c} we show the expression of the genes with highest magnitude loadings, across the relevant cell types. Many are classical cell type markers.

\paragraph{Threshold sensitivity for memory analysis.} In \textbf{Fig.~\ref{fig:f4_convergence}f,g}, we quantify the number of genes which are ``memory'' genes which retain a signature of origin cell types compared to genes which are ``erased''. To define a gene as persistently differentially expressed, we have to define cutoffs. Here, we ask if our qualitative observation -- that endothelium retains fewer memory genes than neurons -- is sensitive to the choice of cutoff between memory and erased. In \textbf{Supplemental Fig.~\ref{fig:supp_conv}d}, we find that regardless of choice of cutoffs, the ordering of cell types by erased fraction remains consistent.

\begin{figure}[tp]
  \centering
  \includegraphics[width=\linewidth]{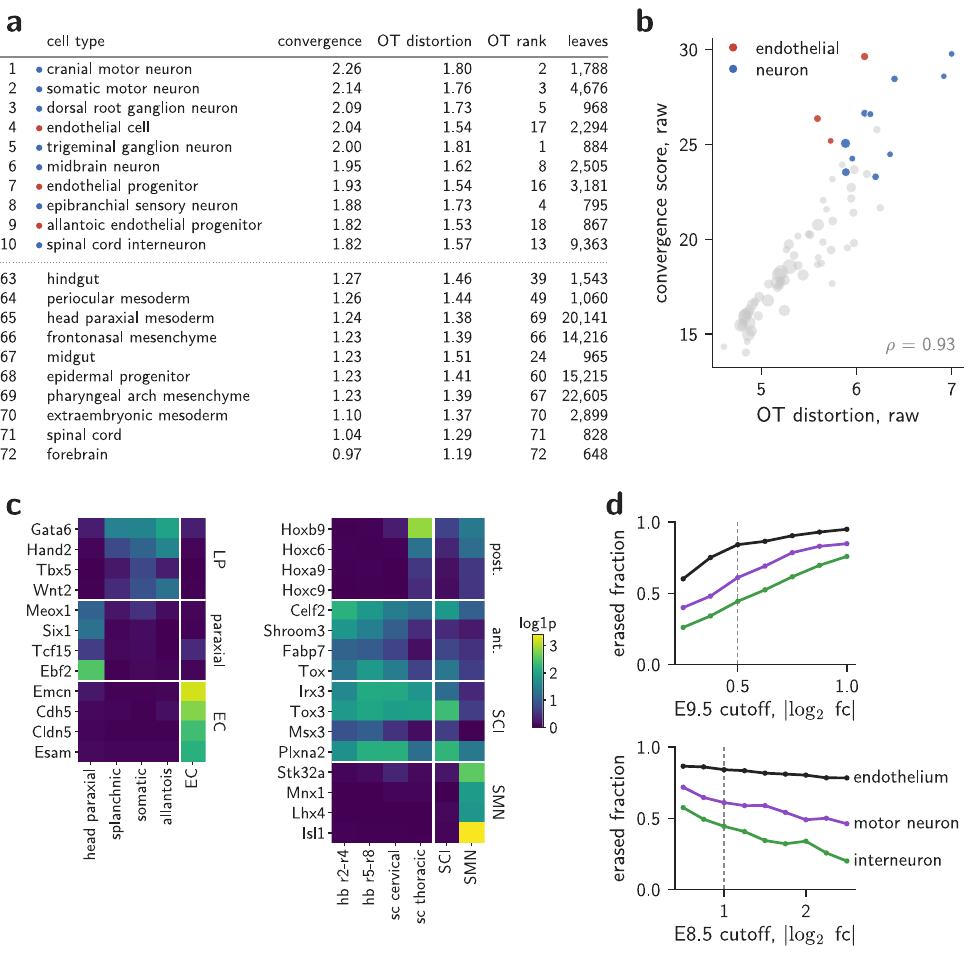}
  \caption{\textbf{Additional results on transcriptional convergence analysis.}
  \textbf{(a)} Table of top 10 and bottom 10 cell types as ranked by convergence score. Convergence and OT discrepancy refer to normalized scores from \textbf{Fig.~\ref{fig:f4_convergence}b}. \textit{Leaves} refers to number of observed cells of each type from which scores are computed.
  \textbf{(b)} Same as \textbf{Fig.~\ref{fig:f4_convergence}b}, except where scores have not been normalized by local nearest-neighbor distances.
  \textbf{(c)} Mean expression (per cell type) of each of the top four genes per transcriptional program. \textit{LP} refers to lateral plate direction, \textit{paraxial} refers to paraxial mesoderm direction, \textit{post.} refers to posterior direction, \textit{ant.} refers to anterior direction. \textit{EC} = endothelial cell, \textit{SCI} = spinal cord interneuron, \textit{SMN} = somatic motor neuron. 
  \textbf{(d)} Fraction of divergent genes considered erased, per cell type, across different divergence (E8.5) and memory (E9.5) cutoffs.
  }
  \tocfigure{Additional results on transcriptional convergence analysis}
  \label{fig:supp_conv}
\end{figure}

\subsection{Additional results on perturbation effect comparisons with MGD annotations}

In \textbf{Fig.~\ref{fig:f5_pert}c}, we show that lineage inpainting predictions of perturbation effects in mice are significantly higher for genes with high mutational constraint in humans compared to those that are unconstrained. It has been previously observed by \citet{Lek2016-wm} that the bulk gene expression (in humans) is typically higher for mutationally constrained genes. Here we study the relationship between gene expression within cell types and lineage inpainting model predictions.

We first take all (gene, cell type) pairs present in the MGD annotations analyzed in \textbf{Fig.~\ref{fig:f5_pert}b} and ask if the mean expression of a gene within a cell type is related to the lineage inpainting predicted effect of knocking out a gene on a cell type. We find that the predicted depletion of a cell type by a gene knockout is indeed correlated with the expression of that gene within that cell type (\textbf{Supplemental Fig.~\ref{fig:supp_mgi}a}). 

However, there are a number of genes with low or no expression in a cell type whose knockout is predicted to deplete that cell type. For each (gene $g$, cell type $c$) pair, we rank the cell types by mean expression of the gene $g$, and compute the rank of $c$. We repeat the same ranking process based on predicted depletion rather than expression, and identify the pairs with the highest discrepancy (\textbf{Supplemental Fig.~\ref{fig:supp_mgi}b}).

Inspecting these model predictions, we find that the gene -- while not expressed in the cell type itself -- is expressed in lineage inpainting's inferred ancestors of the cell type. We visualize three examples in \textbf{Supplemental Fig.~\ref{fig:supp_mgi}c}, which show a pattern where a transcription factor is transiently expressed in the progenitors of a cell type.

\begin{figure}[tp]
  \centering
  \includegraphics[width=\linewidth]{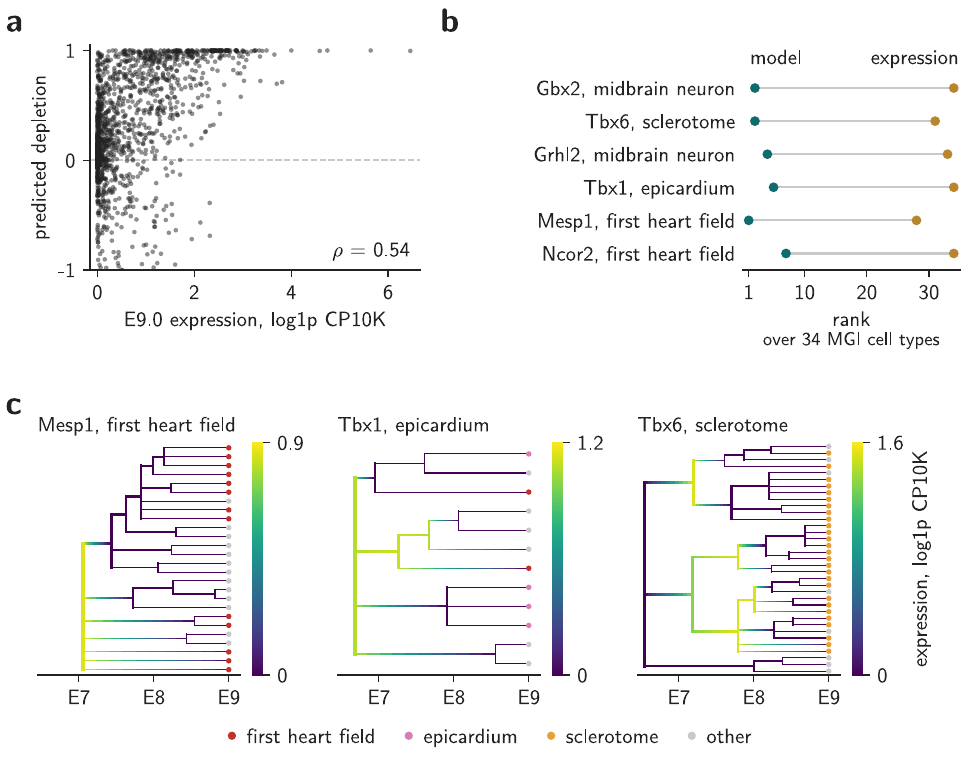}
  \caption{\textbf{Additional results on perturbation effect comparisons with MGD annotations.}
  \textbf{(a)} Across all MGD annotated (gene, cell type) pairs, the predicted depletion of the cell type under gene knockout against the mean expression of the gene within the cell type. 
  \textbf{(b)} The top six transcription factors with low cell type mean expression, and high predicted depletion. Chosen by difference between rank across 34 cell types with corresponding MGD annotations.
  \textbf{(c)} Representative trees with inferred expression along edges for cell type gene pairs with low expression and high predicted depletion.
  }
  \tocfigure{Additional results on perturbation effect comparisons with MGD annotations}
  \label{fig:supp_mgi}
\end{figure}

\begin{figure}[tp]
  \centering
  \includegraphics[width=\linewidth]{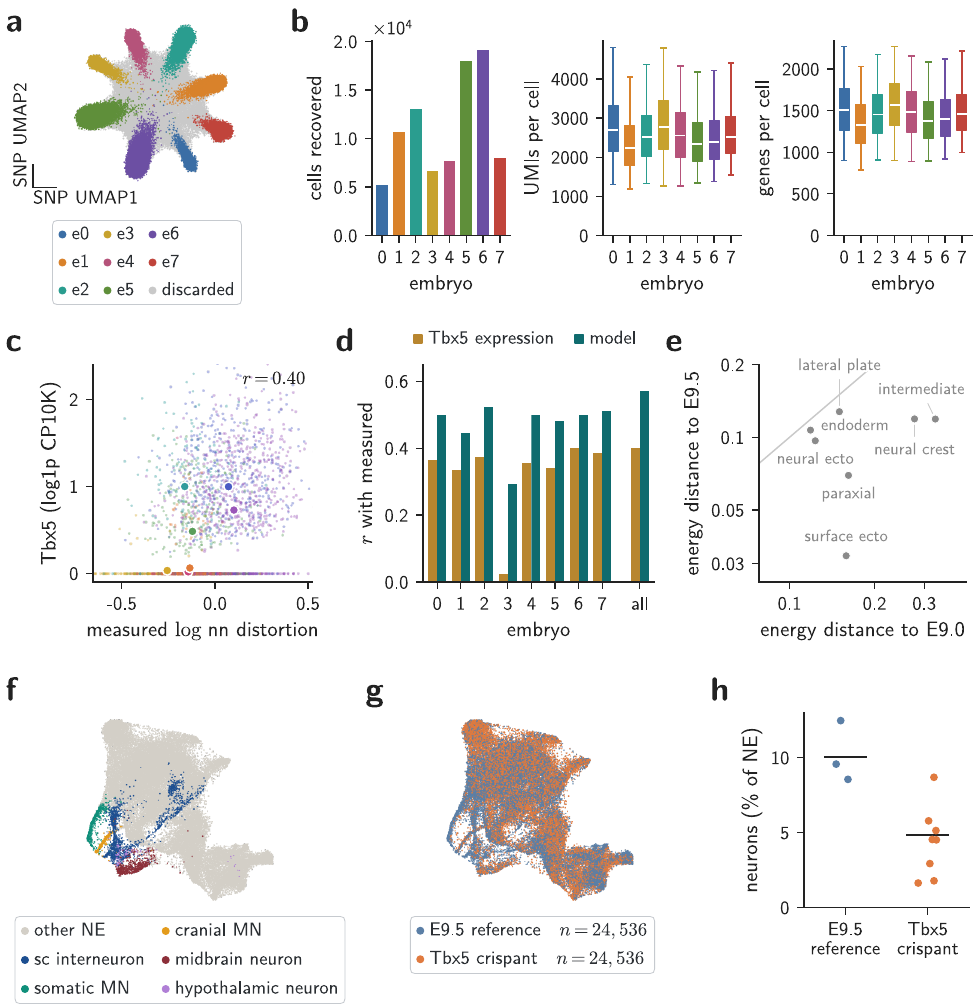}
  \caption{\textbf{Additional results on Tbx5 mutant profiling.}
  \textbf{(a)} UMAP of cells, prior to quality filtering, based on SNP profile, hued by inferred embryo identity.
  \textbf{(b)} Number of cells, counts per cell, and genes per cell, across embryos. UMI = unique molecular identifier.
  \textbf{(c)} Correlation of measured nearest-neighbor distortion (per reference E9.5 cell) and Tbx5 expression.
  \textbf{(d)} Comparison of model predicted distortion agreement and reference cell Tbx5 expression correlation, across embryos.
  \textbf{(e)} Energy distance of Tbx5 crispant cells to E9.0 reference and E9.5 reference cells of corresponding coarse cell type label.
  \textbf{(f)} UMAP of neuroectoderm cells from E9.5 reference and Tbx5 crispant. Neurons hued, all other cells grey.
  \textbf{(g)} Same as (f), hued by dataset rather than cell type.
  \textbf{(h)} Fraction of neuroectoderm labeled as a neuron, across reference and crispant embryos.
  }
  \tocfigure{Additional results on Tbx5 mutant profiling}
  \label{fig:supp_tbx5}
\end{figure}

\subsection{Additional results on Tbx5 crispant profiling}

In our experimental measurement of Tbx5 knockout effects, we profiled eight individual embryos. We demultiplexed the cells of the embryos using their SNP profiles as naturally occurring barcodes, similar to previous work \cite{Kijima2026-fd, Dias2026-vh}. Many cells have SNP barcodes that fall into a distinct cluster, while others are ambiguous or resemble doublets (Methods, \textbf{Supplemental Fig.~\ref{fig:supp_tbx5}a}). After this demultiplexing and SNP-based doublet removal pipeline, we recovered between $5\cdot 10^3$ and $2\cdot10^4$ cells per embryo (\textbf{Supplemental Fig.~\ref{fig:supp_tbx5}b}). The embryos had similar transcript sensitivity, with median unique molecular identifiers (UMIs) of roughly 2500 and roughly 1500 genes detected per cell across all eight embryos. 

Similar to in \textbf{Supplemental Fig.~\ref{fig:supp_mgi}a}, we again asked the extent to which perturbation effect (here measured by $\log$ nearest-neighbor distortion) is related to the depletion of a cell state. We compare the expression of \texttt{Tbx5} in each reference cell state to its measured $\log$ nearest-neighbor distortion. We find that they are loosely correlated with $r=0.40$ (\textbf{Supplemental Fig.~\ref{fig:supp_tbx5}c}). In aggregate, this is less predictive than the model predictions shown in \textbf{Fig.~\ref{fig:f5_pert}m}, which has $r=0.57$. This trend is consistent across all 8 embryos individually (\textbf{Supplemental Fig.~\ref{fig:supp_tbx5}d}).

While we measure clear effects for \texttt{Tbx5} perturbation in the lateral plate mesoderm as expected, we also found global effects which were not predicted by the model, nor evident from prior literature. We quantify these developmental delays per tissue, by measuring distributional distance of cell states to the E9.0 and E9.5 reference populations. While all tissues are more similar to E9.5 than E9.0, some tissues appear close to the boundary, suggesting a developmental delay (\textbf{Supplemental Fig.~\ref{fig:supp_tbx5}e}). This delay is particularly clear in the neuroectoderm. The \texttt{Tbx5} crispant data contains very few neurons compared to the E9.5 reference, while the rest of the neuroectoderm is well represented (\textbf{Supplemental Fig.~\ref{fig:supp_tbx5}f--h}).

\section{Supplementary results: choosing an implementation of lineage inpainting}
\label{supp:grid}
The figures below compare 28 implementations of lineage inpainting to understand which modeling choices are most effective in practice. We considered models which vary across the following design choices:

\begin{enumerate}
  \item The tree features used in the alignment loss.
  \item The choice to predict raw states or the Felsenstein additive displacement baseline.
  \item Whether we use the imputer or the sampler for ancestral state predictions.
  \item The conditional generative model we use for the sampler. We test a one-step model (essentially re-using the imputer as a sampler), as well as a deterministic flow matching model and a diffusion model (see Section~\ref{sec:diffusion} for details on these two interpolants).
  \item How we allow the gradient to flow through the model, and whether sampler is trained directly on imputer outputs, or on real observed cells.
\end{enumerate}

\paragraph{Tree features in the alignment loss.}
We begin by evaluating whether the different components of our imputer loss lead to measurable improvements in the sampling models we get. We trained the model with an imputer loss that only tries to match the marginal at each time, then adding the path and common ancestor components respectively before adding both. We find that incorporating the path and common ancestor features jointly improves the main coupling metrics at a small cost to the marginal performance, while either feature alone has metric-specific tradeoffs (\textbf{Supplemental Fig.~\ref{fig:grid_aligner}}).

\begin{figure}[tp]
  \centering
  \includegraphics[width=\linewidth]{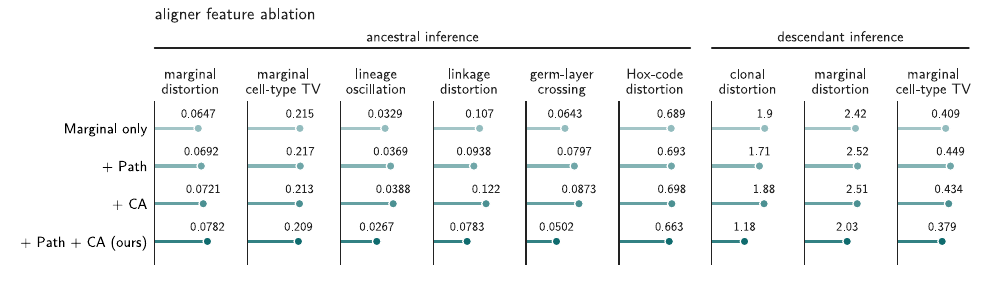}
  \caption{\textbf{The higher-order components of the alignment loss tighten the lineage structure.} The alternating VP-SDE model trained with the state marginal alone, with the path added, with the common-ancestor triple added, and with all three loss components together.}
  \label{fig:grid_aligner}
\end{figure}

\paragraph{Felsenstein displacement.}
The next question we address is whether our inductive bias driven by the additive baseline has any impact on the models we learn. In \textbf{Supplemental Fig.~\ref{fig:grid_tree}} we demonstrate that every metric which measures the quality of the coupling between timepoints is improved by incorporating the Felsenstein inductive bias. For many models the marginals also improve under the baseline. This demonstrates the value of the additive baseline in learning a reasonable coupling.
 
\begin{figure}[tp]
  \centering
  \includegraphics[width=\linewidth]{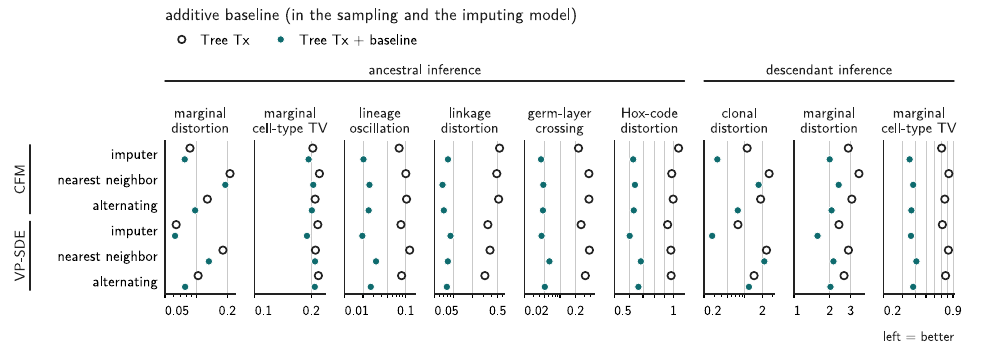}
  \caption{\textbf{The additive baseline improves the coupling metrics.} Using the full grid of models we compare the plain tree transformer and the tree transformer using the additive baseline in Sec. \ref{sec:architecture}. The additive baseline improves essentially all metrics across all models.}
  \label{fig:grid_tree}
\end{figure}

\paragraph{Sampler versus imputer.}
While the sampling model is the only model which can make descendant predictions, in principle we could use either the sampler or the imputer to make ancestral predictions. We next ask, in practice, whether the imputer or the sampler produce better results across our metrics. In \textbf{Supplemental Fig.~\ref{fig:grid_sampler_vs_imputer}} we use both the imputer and the sampler to compute our primary imputer metrics, and plot each imputer against its corresponding sampler. Across the lineage oscillation, distortion, and the germ layer crossing the sampler models are consistently better than the imputers (with the exception of some one-step models). On the Hox code distortion the results are more mixed, but overall it is clear that the samplers are better models on our metrics.
 
\begin{figure}[tp]
  \centering
  \includegraphics[width=\linewidth]{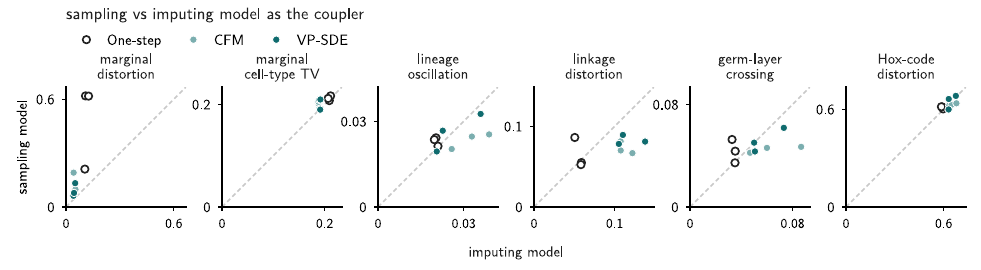}
  \caption{\textbf{Comparing the imputer and sampler models on ancestral inference metrics.} Using the full grid of models we compute all our metrics for both the imputer and the sampler, plotting them against each other. Points below the diagonal demonstrate better performance as a sampler than as an imputer. We use markers to indicate the three different sampler models we test: the one-step model, the CFM, and VP-SDE.}
  \label{fig:grid_sampler_vs_imputer}
\end{figure}

\paragraph{Sampler family.}
We next ask what type of conditional generative model is most effective as a sampler. We consider three families: a one-step model (mirroring the architecture of the imputer), and two stochastic interpolant models (CFM and VP-SDE models described in Section~\ref{sec:diffusion}). In \textbf{Supplemental Fig.~\ref{fig:grid_sampler}} we show that the one step model is worse than the stochastic interpolants across most methods. The CFM and VP-SDE models have very similar performance on the coupling metrics. The VP-SDE models consistently perform better on the distributional metrics such as the marginal distortion.
 
\begin{figure}[tp]
  \centering
  \includegraphics[width=\linewidth]{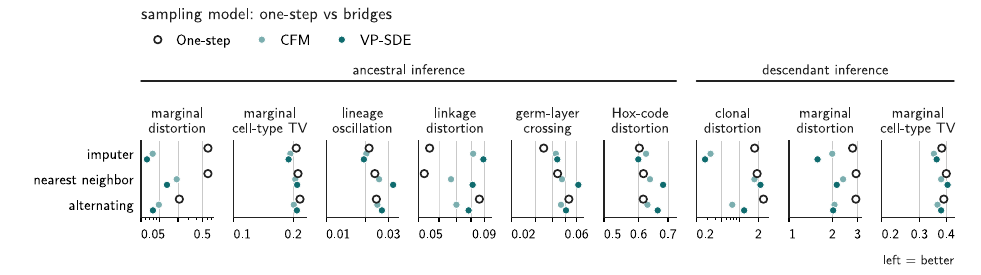}
  \caption{\textbf{Iterative samplers sharpen the distributional metrics over the one-step model.} Each row is an imputer and a pairing; open circles show the one-step model, and colored circles show flow matching (CFM) and variance-preserving diffusion (VP-SDE). Left is better in every panel.}
  \label{fig:grid_sampler}
\end{figure}

\paragraph{Gradient flow and nearest neighbor snapping.}
Next we can consider how the gradient flow through training affects the resulting model. In particular we consider three settings: first, we simply use the imputed states to directly train the sampling model; second, we replace the imputed states with their nearest neighbors (except the snapshot we condition on, which we leave as generated by the imputer); finally we test an alternating version, where we switch between the imputer and the imputer-nearest-neighbor. The disadvantage of the imputer-only model is that the sampling model is trained almost entirely on ``synthetic'' data; on the other hand, the nearest neighbor model does not let the gradients flow through the imputer fully. The alternating model aims for the best of both these worlds. In \textbf{Supplemental Fig.~\ref{fig:grid_stopgrad}} we compare these three settings, including stop-gradient variants of the raw and nearest-neighbor settings. We find that the gradient flows help in almost all cases, and that overall the alternating model maintains the performance of the imputer-only model while enabling training on real data.
 
\begin{figure}[tp]
  \centering
  \includegraphics[width=\linewidth]{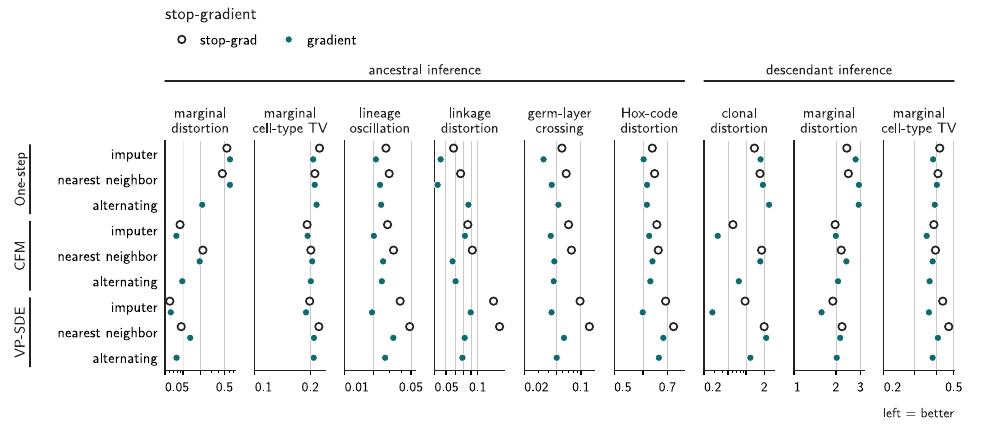}
  \caption{\textbf{Comparing model performance against different training setups and gradient flows.} Using the models from the grid we can compare three methods to leverage the imputer in the sampling model. We compare training on the raw imputed trees, the nearest neighbors, and a training procedure that alternates between the imputer and the nearest neighbors. The raw and nearest-neighbor settings are shown with and without gradients flowing from the sampler to the imputer; the alternating setting retains this gradient.}
\label{fig:grid_stopgrad}
\end{figure}

\paragraph{Model selection.}
Reviewing our results across these different axes, we selected the alternating VP-SDE model with the additive baseline because this seemed to produce the best overall results while enabling training the sampler model on real data.

\paragraph{Comparing the transition maps learned by different models.}
Finally we can examine to what degree these choices affect the cell type transition graph we infer from the model. Using the same procedure as in \textbf{Fig.~\ref{fig:f3_summary}a}, we infer a cell type level transition map for each model. We then compare transition maps using their cosine similarity (\textbf{Supplemental Fig.~\ref{fig:grid_overview}}). We find that the CFM and VP-SDE models using the Felsenstein additive baseline produce qualitatively similar cell-type transitions across our methods. This implies that although there is real variation in metrics across methods, the larger scale results are stable across these architectural choices.

\begin{figure}[tp]
  \centering
  \includegraphics[width=\linewidth]{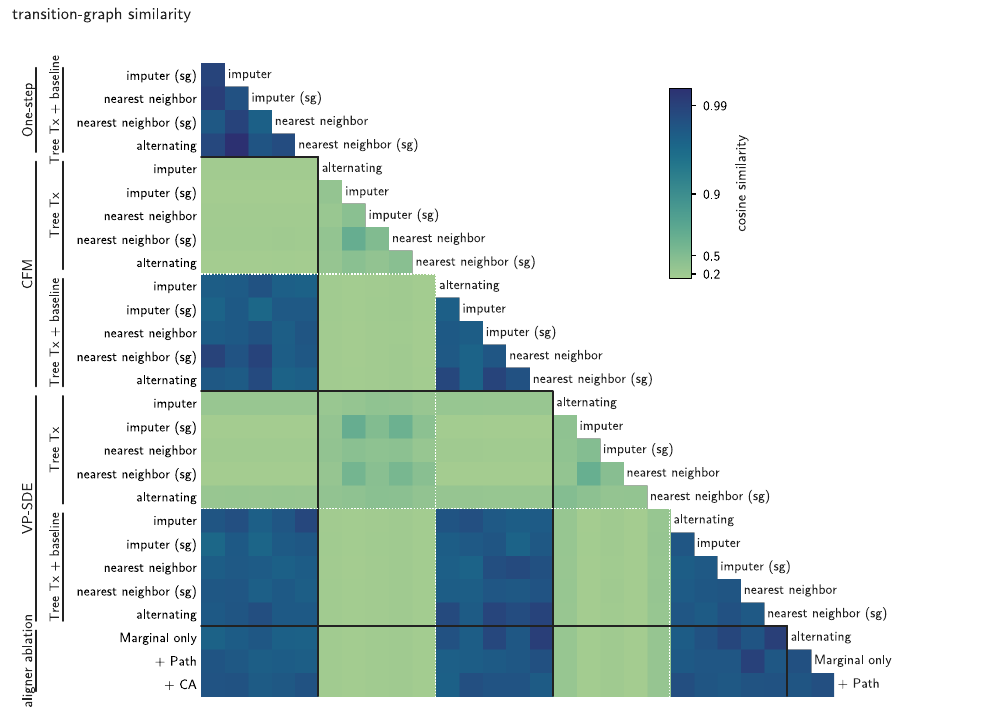}
  \caption{\textbf{Iterative models with the additive baseline infer similar cell-type transition graphs.} We compute the cell-type transition graph as in the main text figures for each model. We then convert these graphs into a vector, where each dimension of the vector represents the edge weight of an edge in the underlying graphs. We compute the cosine similarity between these vectors to represent how similar the induced cell-type transitions are across models. Missing edges have zero weight, and vectors use the union of time-resolved edges. The color scale transforms cosine similarity $c$ as $-\log_{10}\max(1-c,10^{-6})$.}
  \label{fig:grid_overview}
\end{figure}

\FloatBarrier

\section{Supplementary results: theoretical analysis of alignment loss}
\label{supp:theory}

To learn an imputation model, we compare distributions of cell states
and lineage features pooled across collections of completed trees. In particular, we use a distributional divergence to compare the marginal distributions of states, the distribution of states along ancestral paths, and the distribution of states along ancestor-descendant-descendant triples. The goal of this section is to illustrate the principles underlying this procedure.

We start by considering how we can compare two state-decorated trees when we do not know the correspondences between nodes, showing that this reduces to a distributional comparison for any pair of trees. We then show how to lift that construction over pairs of trees to pairs of distributions over trees. This gives us a principled whole-tree discrepancy and connects the energy distance we focus on with analogous divergences, like the Wasserstein-1 distance. 

The final section here concretely traces the trade-offs between this whole-tree loss and our three component loss. The central idea is that, because we have few trees and because the whole-tree loss is very high-dimensional, we can actually do a better job at measuring the discrepancy between two collections of trees by restricting ourselves to these lower dimensional components. This statistical efficiency comes at a price: two whole-tree distributions can agree on all the features but disagree on ``long-range'' dependencies that the features do not measure. In the binary setting, we make this tradeoff exact by characterizing the set of distributions over trees left unresolved by the features and comparing the bias and variance of the resulting estimate of the whole-tree energy. We find regimes in which the projected feature statistic has lower MSE for the whole-tree energy despite this; as tree depth grows, the computational advantage of the feature representation increases.

\subsection{Comparing unordered trees}

Consider two state-decorated trees
\[
\mathcal X=(\mathcal T_{k,d},x),
\qquad
\mathcal X'=(\mathcal T_{k,d},x'),
\]
with the same \(k\)-ary topology \(\mathcal T_{k,d}\), with \(k\geq2\) and depth \(d\). Here \(\mathcal T_{k,d}\) specifies the tree topology and \(x\) and \(x'\) give the states attached to its nodes. If the children were ordered, we could flatten the node states in that order and compare the resulting vectors using Euclidean distance. We assume that this order is not available; instead we compare trees using only their node states and parent--child relationships, treating any permutation of siblings as observationally equivalent.

Throughout this section, we therefore work with trees ``identified up to child order.'' If we knew the ordering of the children between trees we could compare that ordering directly, and this might well be different from the permutation-invariant distances we focus on here. But this kind of ordering is not available in our data, so this is, in some sense, the best we can do.

We will start with an intuitive way to compare our unordered trees, by trying every possible pairing and picking the one that minimizes the distances. We will show that this is actually a metric on the space of unordered trees with fixed topology. We will use this as a point of comparison for our actual loss. Let \(\Pi_{k,d}\) denote the set of orderings we can get by permuting the children of each internal node, and let
\(\operatorname{vec}_{\pi}(\mathcal{X})\) denote the flattened representation of \(\mathcal{X}\) under ordering \(\pi\). A natural distance is
\[
    d_{\min}(\mathcal{X},\mathcal{X}')
    =
    \min_{\pi,\pi'\in\Pi_{k,d}}
    \left\|
        \operatorname{vec}_{\pi}(\mathcal{X})
        -
        \operatorname{vec}_{\pi'}(\mathcal{X}')
    \right\|.
\]
Because only the relative ordering of the two trees matters, we may fix any ordering of the first tree and equivalently write
\[
    d_{\min}(\mathcal{X},\mathcal{X}')
    =
    \min_{\pi\in\Pi_{k,d}}
    \left\|
        \operatorname{vec}(\mathcal{X})
        -
        \operatorname{vec}_{\pi}(\mathcal{X}')
    \right\|.
\]
Two trees that differ only by child ordering are treated as the same unordered tree. This construction is related to the hierarchical tree distance underlying Tree Mover's Distance \citep[Definition 4]{chuang2022tree}: with unit depth weights and squared Euclidean node costs, that recursion gives $d_{\min}^2$ for matching complete rooted topologies, and an outer square root gives $d_{\min}$.

Our loss actually comes from a second, slightly less obvious, construction. Rather than choose the single relative ordering that makes two trees as close as possible, we represent each tree by the whole distribution of all its possible ordered representations. We take a tree \(\calX\) and sample child orderings uniformly at random, defining the distribution of the resulting flattened ordered vectors as:
\[
    P_{\calX,\mathrm{ord}}
    =
    \operatorname{Law}_{\Pi}
    \left\{
        \operatorname{vec}_{\Pi}(\calX)
    \right\},
    \qquad
    \Pi\sim\operatorname{Unif}(\Pi_{k,d}).
\]

The minimum-order distance and energy distance both compare the same distributions. The Wasserstein distance selects an optimal coupling, whereas energy distance uses between- and within-distribution expectations. We can make this precise:

\begin{proposition}[Randomized-order metrics on unordered trees]
\label{prop:randomized-order-metric}
On state-decorated trees with topology \(\calT_{k,d}\),
\[
    P_{\calX,\mathrm{ord}}=P_{\calX',\mathrm{ord}}
    \quad\Longleftrightarrow\quad
    \calX=\calX'
    \text{ up to child order}.
\]
Moreover,
\[
    d_{\min}(\calX,\calX')
    =
    W_1(P_{\calX,\mathrm{ord}},P_{\calX',\mathrm{ord}})
    \qquad\text{and}\qquad
    \sqrt{
        D_E(P_{\calX,\mathrm{ord}},P_{\calX',\mathrm{ord}})
    }
\]
are metrics on unordered trees.
\end{proposition}

\begin{proof}
First suppose that \(\calX\) and \(\calX'\) differ only by a fixed child ordering. Uniformly reordering \(\calX'\) is equivalent to applying this ordering to \(\calX\) and then applying another uniform random ordering. Composing a uniform random ordering with any fixed ordering preserves uniformity, so \(P_{\calX,\mathrm{ord}}=P_{\calX',\mathrm{ord}}\).

Conversely, suppose that \(P_{\calX,\mathrm{ord}}=P_{\calX',\mathrm{ord}}\). Fix any ordered representation of \(\calX\). Every ordering is sampled with positive probability so this has positive probability under \(P_{\calX,\mathrm{ord}}\). Since the laws are equal the same holds under \(P_{\calX',\mathrm{ord}}\). Thus some ordering of \(\calX'\) must produce exactly the same flattened representation. Hence the trees agree up to child order.

For any metric \(\Delta\) on probability laws over ordered tree vectors, write \(d_{\Delta}(\calX,\calX')=\Delta(P_{\calX,\mathrm{ord}},P_{\calX',\mathrm{ord}})\). This implies that the map \(\calX\mapsto P_{\calX,\mathrm{ord}}\) is injective on unordered trees, and composing any metric \(\Delta\) with this map preserves nonnegativity, symmetry, and separation. It also preserves the triangle inequality because
\[
\begin{aligned}
    d_{\Delta}(\calX,\calX'')
    &=
    \Delta(P_{\calX,\mathrm{ord}},P_{\calX'',\mathrm{ord}}) \\
    &\leq
    \Delta(P_{\calX,\mathrm{ord}},P_{\calX',\mathrm{ord}})
    +
    \Delta(P_{\calX',\mathrm{ord}},P_{\calX'',\mathrm{ord}}) \\
    &=
    d_{\Delta}(\calX,\calX')
    +
    d_{\Delta}(\calX',\calX'').
\end{aligned}
\]
Thus \(d_{\Delta}\) is a metric.

All that is left is to check the Wasserstein representation of \(d_{\min}\). The Euclidean distance between ordered representations of \(\calX\) and ordered representations of \(\calX'\) is at least \(d_{\min}(\calX,\calX')\), so every coupling of their ordering laws has at least this expected cost. For the reverse inequality, let \(\pi^\star\) attain the minimum in the definition of \(d_{\min}\). If \(\Pi\) is uniform on \(\Pi_{k,d}\), then pairing $\operatorname{vec}_{\Pi}(\calX)$ with $\operatorname{vec}_{\Pi\circ\pi^\star}(\calX')$ defines a coupling of the two uniform ordering laws. Applying the same ordering to both vectors preserves their Euclidean distance, so this coupling has constant cost \(d_{\min}(\calX,\calX')\) and attains the lower bound.

Each ordering law has finite support. Wasserstein distance is a metric on Euclidean probability laws with finite first moments, as is the square root of energy distance~\cite{Szekely2013}. The general claim therefore shows that both \(d_{\min}\) and the square root of energy distance are metrics.
\end{proof}

\subsection{Comparing collections of trees}

The preceding metrics compare two fixed trees. We now consider the distributional problem, where we need to compare two collections of trees. Given the sets of trees
\[
    \calX_1,\ldots,\calX_n,
    \qquad
    \calX'_1,\ldots,\calX'_n,
\]
each tree has the ordering law \(P_{\calX,\mathrm{ord}}\) introduced in the preceding section, so the two collections induce mixtures over flattened ordered vectors:
\[
    \frac{1}{n}\sum_{i=1}^n P_{\calX_i,\mathrm{ord}},
    \qquad
    \frac{1}{n}\sum_{i=1}^n P_{\calX'_i,\mathrm{ord}}.
\]
We can sample one tree from a collection and then order its children uniformly at random. Repeating this process we can compare the trees as in the previous section, but we can also ``lift'' our distributional distances to the setting where we have collections of trees: here, the Wasserstein distance chooses an optimal coupling between the two ordered-vector mixtures, whereas energy distance compares them through between- and within-mixture expectations.

We can write \(P\) and \(Q\) for two generic laws on state-decorated unordered trees with topology \(\calT_{k,d}\) and define
\[
    P_{\mathrm{ord}}
    =
    \operatorname{Law}\!\left\{
        \operatorname{vec}_{\Pi}(\calX)
    \right\},
    \qquad
    \calX\sim P,
    \quad
    \Pi\sim\operatorname{Unif}(\Pi_{k,d}),
\]
where \(\calX\) and \(\Pi\) are independent, and define \(Q_{\mathrm{ord}}\) analogously. Variation among trees is part of the distributions being compared; the variation in child order is a technical device we are using to handle the arbitrariness of child order. This lets us write a generic result for the distributional comparison of two tree laws which works for either Wasserstein or energy distance.

\begin{proposition}[Randomized-order comparison of tree laws]
\label{prop:random-tree-ordering}
For laws \(P\) and \(Q\) on state-decorated unordered trees with topology \(\calT_{k,d}\),
\[
    P_{\mathrm{ord}}=Q_{\mathrm{ord}}
    \quad\Longleftrightarrow\quad
    P=Q.
\]

Consequently, on tree laws whose ordered-vector laws have finite Euclidean first moments,
\[
    W_1(P_{\mathrm{ord}},Q_{\mathrm{ord}})
    \qquad\text{and}\qquad
    \sqrt{
        D_E(P_{\mathrm{ord}},Q_{\mathrm{ord}})
    }
\]
are metrics on tree laws.
\end{proposition}

\begin{proof}
If an ordered vector is drawn from \(P_{\mathrm{ord}}\), reconstructing the fixed topology and forgetting child order returns an unordered tree drawn from \(P\). Hence \(P_{\mathrm{ord}}=Q_{\mathrm{ord}}\) implies \(P=Q\), while the converse follows directly from the definition. The map \(P\mapsto P_{\mathrm{ord}}\) is therefore injective, so the pullback argument from Proposition~\ref{prop:randomized-order-metric} shows that any metric on ordered-vector laws pulls back to a metric on tree laws.

Finally, Wasserstein distance is a metric on Euclidean probability laws with finite first moments, as is the square root of energy distance~\cite{Szekely2013}. The general claim therefore applies to both. 
\end{proof}

The randomized-order representation gives a faithful whole-tree target, which we can treat as the ideal object that the practical alignment loss is trying to approximate. In our setting, however, we have a limited number of whole-tree samples, making this high-dimensional law difficult to estimate accurately. As trees grow, the whole-tree comparison asks each observation to support more joint directions, and the sampling variance can become huge. The alignment loss instead turns every tree into sets of fixed-size generation, path, and common-ancestor features. This effectively converts the high-dimensional problem into a collection of lower-dimensional problems, improving statistical efficiency but making some kinds of ``long-range'' dependence impossible to measure. The next subsection studies this tradeoff using binary states on binary trees.

\subsection{From whole-tree laws to low-order features}

To make this tradeoff exact, we now restrict our attention to complete binary trees with states encoded as \(-1\) and \(+1\). This is essentially a generalized Ising model on trees, where we allow the interactions to vary over arbitrarily long ranges on the tree. The binary setting is useful since it allows us to write all of our analysis as functions of the full probability mass table, which is finite and discrete. 

Our three feature families break up each tree in complementary directions. The generation marginal keeps the state distribution at each time, the path triple captures a  ``depth'' summary along each ancestral lineage, and the common-ancestor triple is a ``breadth'' summary across branches below every shared ancestor. We pool these feature samples at each time, replacing the whole-tree comparison with repeated comparisons of these low-dimensional feature distributions. 

These features let us break the large tree into smaller pieces, effectively reducing the variance of our divergence. But the central issue is that the feature laws cannot distinguish all whole-tree laws. The goal of this section is not to show that these features identify the whole-tree law (they almost certainly do not) but to quantify how much whole-tree information they resolve relative to their variance gains. We start by working through a simple example to show how this pooling works. We then explore a seven-node tree in \textbf{Supplemental Fig.~\ref{fig:supp-ising-phase}}, which lets us carefully analyze our feature laws against the whole-tree loss, weighing the statistical efficiency against the loss of information from the featurization.

To make the pooling concrete, let \(V\) be the node set of for one of our completed trees \(\calT_{k,d}\), with \(|V|=N\). We write $D_E(P,Q)=D_E(P_{\mathrm{ord}},Q_{\mathrm{ord}})$ for whole-tree energy. For a tree law \(R\), let \(p_R(x)\) be the mass assigned by \(R_{\mathrm{ord}}\) to the binary assignment \(x\in\{-1,1\}^V\).

The smallest example is a three-node ``fork'' with a parent \(0\) and two children \(1\) and \(2\). On this fork we can consider a two-node path from parent to child; the path has two eligible placements, one for each child, and the fork has a single time stratum, so pooling reduces to averaging these two placements. 

We can think about two concrete laws so we can understand how this feature compares to the whole-tree (whole fork) distance. We will let \(P\) be a law where the parent and both children are independent fair signs. We will compare this to the Bernstein triple \(Q\), where we sample the parent and one child independently, but the second child is always the product of the parent and first child. In this construction every pair is independent but the product of all three signs is always \(+1\). Equivalently, \(p_P(x)=2^{-3}\) and \(p_Q(x)=2^{-3}(1+x_0x_1x_2)\); \textbf{Supplemental Fig.~\ref{fig:supp-ising-running-example}a} shows the two laws, and \textbf{Supplemental Fig.~\ref{fig:supp-ising-running-example}b} their masses on the eight ordered assignments.

We visualize these distributions via their probability mass tables in \textbf{Supplemental Fig.~\ref{fig:supp-ising-running-example}b}. But it is hard to understand what distributions are or are not distinguishable from a given feature using that representation. In our binary setting there is a cleaner perspective. Because every state is \(\pm1\), the product of the states over a node set \(A\subseteq V\) is also just \(\pm1\). We can describe the distribution of this sign by its mean, since the product is \(1\) with some probability and \(-1\) otherwise. It is useful to define the product moment
\[
    m_R(A)
    =
    \mathbb E_R\left(\prod_{i\in A}X_i\right)
    =
    \sum_x p_R(x)\prod_{i\in A}x_i,
    \qquad X\sim R_{\mathrm{ord}}.
\]
These moments are a standard basis for boolean functions \cite{o2014analysis}: with the convention \(m_R(\emptyset)=1\), we can convert back to the probability table using
\[
    p_R(x)
    =
    2^{-N}
    \sum_{A\subseteq V}
    m_R(A)\prod_{i\in A}x_i,
\]
which for a single node is just \(\{1\pm m_R(\{v\})\}/2\). Each moment measures the parity of a set of nodes, and together the moments carry every combination of nodes. \textbf{Supplemental Fig.~\ref{fig:supp-ising-running-example}c} visualizes this change of coordinates for the two fork examples: every nonconstant moment of \(P\) is zero, while \(Q\) has a single nonzero nonconstant moment, the three-node parity \(m_Q(\{0,1,2\})=1\), which we can invert back to \(p_Q(x)=2^{-3}(1+x_0x_1x_2)\). The moment coordinates make the Energy distance a bit easier to understand: in Proposition~\ref{prop:binary-energy-expansion} we show that the whole-tree energy is a positively weighted sum of squared moment differences and each feature distance keeps a subset of these moments. So which laws the features can distinguish becomes a question of which moments they retain.

\begin{proposition}
\label{prop:binary-energy-expansion}
For every \(N\geq1\), there are positive constants \(\gamma_{N,r}\), \(1\leq r\leq N\), such that
\begin{equation}
\label{eq:binary-energy-expansion}
    D_E(P,Q)
    =
    \sum_{\substack{A\subseteq V\\ |A|\geq1}}
    \gamma_{N,|A|}\{m_P(A)-m_Q(A)\}^2.
\end{equation}
For each \(A\), the squared moment difference is the energy distance between
the laws of \(\prod_{i\in A}X_i\) under \(P_{\mathrm{ord}}\) and
\(Q_{\mathrm{ord}}\), and \(\gamma_{N,|A|}\) supplies its positive weight.
\end{proposition}

The two fork distributions make this concrete. The whole-fork energy comparison reduces to the single weighted squared moment difference \(\gamma_{3,3}(0-1)^2\), which is normalized to one in the figure. Each parent--child placement already agrees under the two laws, so pooling leaves the path comparison at \(0\): the path feature cannot distinguish \(P\) and \(Q\) so it only partially identifies the whole-tree law. \textbf{Supplemental Fig.~\ref{fig:supp-ising-running-example}d} shows the two comparison values.

\begin{figure}[!t]
    \centering
    \includegraphics[page=1,width=0.98\linewidth]{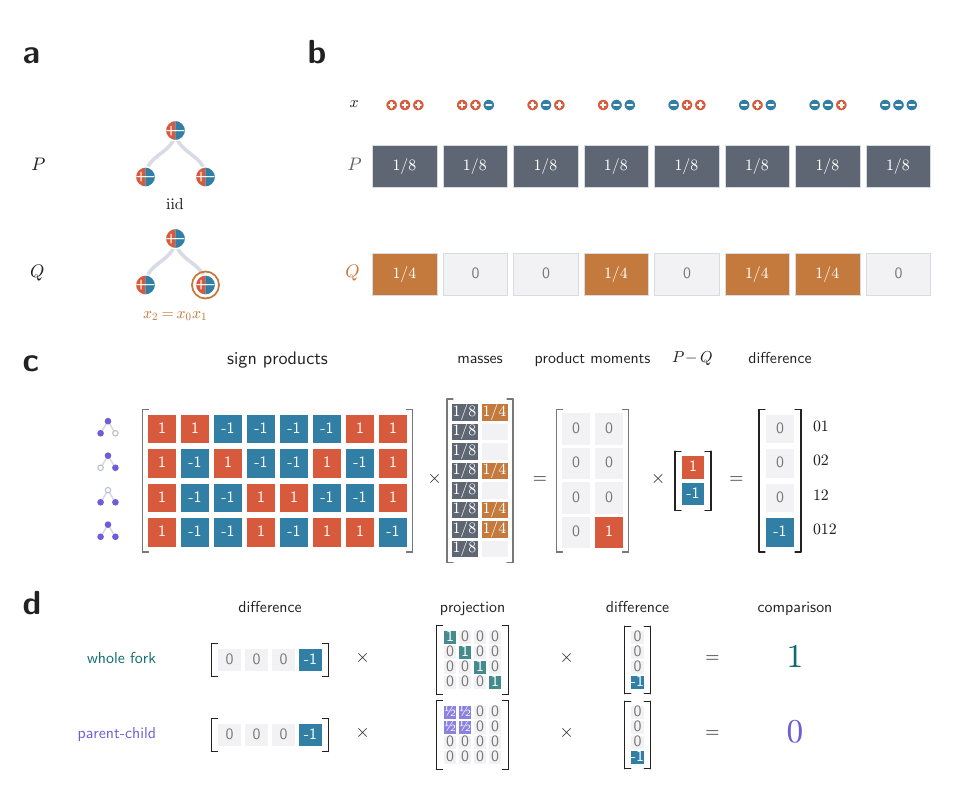}
    \caption{\textbf{A path table can leave the whole-tree law partially identified.} Panel a contrasts the independent fork law \(P\) with the Bernstein triple \(Q\), and panel b shows their masses on the eight ordered assignments. Panel c maps those masses to the product moment representation. The distributions are the same except for the triple-product moment. Panel d normalizes this single squared difference: the whole-fork comparison is \(1\), while the parent--child comparison is \(0\) because every pair marginal agrees.}
    \tocfigure{A path table can leave the whole-tree law partially identified}
    \label{fig:supp-ising-running-example}
\end{figure}

The actual feature loss compares the feature tables one stratum at a time, adding an energy-distance term for each table. Equation~\eqref{eq:binary-energy-expansion} lets us re-write this as a sum of weighted moment terms. Up to the choice of weights, this is equivalent to orthogonally projecting the whole-tree weighted moment vectors onto the subspace of combinations retained by a given set of features. The resulting discrepancy $D_F$ is a slightly different objective from the main feature loss because its weights align with the whole-tree loss's weights rather than the feature loss's weights, but this tweak lets us write the ``omitted'' energy \(D_E(P,Q)-D_F(P,Q)\) as exactly the nonnegative part of the whole-tree energy that the selected feature discrepancy misses.

So far we have focused on the ``population'' version of these discrepancies. To compare the variance from having finite trees we need the per-tree version of these vectors. For a state-decorated tree \(\mathcal X\) and nonempty node set \(A\), average the weighted sign product over all child orderings,
\[
    \overline\phi_A(\mathcal X)
    =
    \frac{\sqrt{\gamma_{N,|A|}}}{|\Pi_{k,d}|}
    \sum_{\pi\in\Pi_{k,d}}
    \prod_{v\in A}
    [\operatorname{vec}_\pi(\mathcal X)]_v,
\]
and let \(\overline{\boldsymbol\phi}(\mathcal X)\) collect these entries. Here we are computing the average over orderings exactly, and for \(\mathcal X\sim R\) this is unbiased for the weighted moments, \(\mathbb E_R\{\overline\phi_A(\mathcal X)\}=\sqrt{\gamma_{N,|A|}}\,m_R(A)\). Equation~\eqref{eq:binary-energy-expansion} then gives that \(D_E(P,Q)\) is the squared distance between the mean vectors \(\mathbb E_P\{\overline{\boldsymbol\phi}(\mathcal X)\}\) and \(\mathbb E_Q\{\overline{\boldsymbol\phi}(\mathcal X)\}\) \cite{sejdinovic2013equivalence}, and \(D_F(P,Q)\) is the same squared distance after projecting each mean vector onto the retained combinations.

Echoing the two sample V-statistic we use for the Energy distance in practice \cite{gretton2012kernel}, the squared distance between empirical mean vectors is built from inner products: average \(\langle\overline{\boldsymbol\phi}(\mathcal X_i),\overline{\boldsymbol\phi}(\mathcal X_j)\rangle\) over all pairs of trees within each sample, including equal indices, and subtract twice the average over pairs across the two samples. We will write this estimator as \(\widehat D_F\) for the projected vectors, computed from \(n\) independent complete trees under each law, with fixed topology, and let \(V_F(P,Q;n)\) be its variance; \(\widehat D_E\) and \(V_E(P,Q;n)\) are the whole-tree versions.

This machinery lets us analyze the tradeoff between sampling variance and partial identification. The difference between the feature and whole-tree loss is just the sets of combinations that the feature loss cannot measure. Two laws can differ in the unmeasured combinations while agreeing on
all selected feature tables, as the parity construction above showed. We can actually compute the worst-case discrepancy for any pair of laws. For a law \(Q\), let
\[
    \mathcal C_F(Q)=\{R:D_F(Q,R)=0\},
\]
the set of laws with the same selected feature tables as \(Q\). Every law in this set differs from \(Q\) only in ``invisible'' combinations, so the loss cannot distinguish \(Q\) from any of these distributions.

The bias--variance decomposition of the feature statistic makes the tradeoff exact. Let \(\boldsymbol\Sigma_R\) denote the covariance of the projected per-tree vector under \(R\), and write \(b_F(P,Q;n)=\operatorname{tr}(\boldsymbol\Sigma_P+\boldsymbol\Sigma_Q)/n\) for its sampling bias:

\begin{proposition}
\label{prop:feature-table-bias-variance}
For \(n\geq1\), \(\mathbb E(\widehat D_F)=D_F(P,Q)+b_F(P,Q;n)\), and for every pair of laws \(P\) and \(Q\),
\begin{equation}
\label{eq:feature-map-mse}
    \operatorname{MSE}\{\widehat D_F;D_E(P,Q)\}
    \leq\max_{R\in\mathcal C_F(Q)}\{D_E(P,R)-D_F(P,R)-b_F(P,Q;n)\}^2
      +V_F(P,Q;n),
\end{equation}
with equality when \(\mathcal C_F(Q)=\{Q\}\).
\end{proposition}

Conceptually the maximum bounds the squared bias over laws the loss cannot distinguish from \(Q\), holding the sampling bias under \(P,Q\) fixed. When \(Q\) is point-identified, \(\mathcal C_F(Q)=\{Q\}\) and equation~\eqref{eq:feature-map-mse} is the standard bias--variance decomposition, with squared bias \(\{D_E(P,Q)-D_F(P,Q)-b_F(P,Q;n)\}^2\). When the features are the whole-tree, \(D_F=D_E\) and \(\widehat D_F=\widehat D_E\), so only the sampling bias remains and the MSE is \(b_E(P,Q;n)^2+V_E(P,Q;n)\), where \(b_E\) is defined using the unprojected per-tree vectors. When \(Q\) matches the tables of \(P\), \(\mathcal C_F(Q)=\mathcal C_F(P)\) and \(D_F(P,\cdot)\) vanishes on it, so the maximum is \(\max_{R\in\mathcal C_F(P)}\{D_E(P,R)-b_F(P,Q;n)\}^2\) while \(V_F(P,Q;n)\) remains the variance under the actual pair; this is the case we evaluate below.

To work this tradeoff out with our actual features we need an example a bit larger than the fork, so we focus on the seven-node example in \textbf{Supplemental Fig.~\ref{fig:supp-ising-phase}}: the complete binary tree of depth two, with root \(0\), children \(1\) and \(2\), and grandchildren \(3,4,5,6\). \textbf{Supplemental Fig.~\ref{fig:supp-ising-phase}a} shows one root--child--grandchild path table and three common-ancestor tables: root with two children, root with two grandchildren, and child with two grandchildren. The marginal tables are also included, but their moments are already fixed by these joint tables. On this tree we can describe the compatible set exactly, numerically search for the worst-case discrepancy, and then weigh that worst case against the variance reduction of these features through Proposition~\ref{prop:feature-table-bias-variance}.

Let \(P_0\) be seven independent fair signs, with \(p_{P_0}(x)=2^{-7}\). The nice thing about the binary setting is that we can exactly describe the compatible set. After child-order symmetry, a whole-tree law is a point in a \(41\)-dimensional simplex. In moment coordinates, the feature tables (marginal, path, and common-ancestor) fix \(13\) combinations and leave the remaining \(28\) free. For any law \(P\), the compatible set \(\mathcal C_F(P)\) is the intersection of the simplex with the \(28\)-dimensional affine space through \(P\) spanned by the omitted combinations.

We can make this a bit more interpretable by focusing on just the four grandchildren. The features observe only one- and two-node marginals on these grandchildren, so they tell us that the four signs are pairwise independent and fair. Any distribution over the grandchildren is a point in the fifteen-dimensional probability simplex, but the distributions compatible with the observed trees occupy one slice of the simplex. These features fix the one- and two-node moments, so only the only freedom is in the two grandchild combinations, the third-order moment (which, somewhat counterintuitively, is the same for every triple due to child-order symmetry), and the four-node parity, so the compatible grandchild laws form a two-parameter family in which the mass of an assignment depends only on its number of positive signs. This gives us a smaller simplex that we can use to visualize the compatible set.

To get this down to a tractable, visual representation, \textbf{Supplemental Fig.~\ref{fig:supp-ising-phase}c} groups the assignments into four patterns, all signs equal, exactly one positive, exactly one negative, and two of each sign; the four pattern totals are nonnegative and sum to one, so they form a three-dimensional simplex, the image of the fifteen-dimensional one under this grouping. On a compatible law the four totals are
\[
\begin{aligned}
    \Pr(\text{all equal})&=\tfrac18\{1+m_{3456}\},&
    \Pr(\text{one positive})&=\tfrac14\{1+2m_{345}-m_{3456}\},\\
    \Pr(2{+}2)&=\tfrac38\{1+m_{3456}\},&
    \Pr(\text{one negative})&=\tfrac14\{1-2m_{345}-m_{3456}\},
\end{aligned}
\]
so the grouping loses nothing on the compatible family: \(m_{3456}\) is recovered from the all-equal total and \(m_{345}\) from the difference between the one-positive and one-negative totals. \textbf{Supplemental Fig.~\ref{fig:supp-ising-phase}d} draws this simplex. Pairwise independence pins the totals to the plane \(\Pr(2{+}2)=3\Pr(\text{all equal})\), which cuts a triangle out of the simplex. Nonnegativity of the individual sign-configuration probabilities restricts this triangle to the compatible quadrilateral, with independence \(P_0\) in its interior.

Now we can interpret the worst-case grandchild law. The fork law \(Q\) is a three-node case of a general parity distribution called a Bernstein construction \cite{wang1979dependent}. The construction is simple: we sample \(k-1\) nodes as independent fair signs and then take the \(k\)th node to be their product. This gives a distribution where any subset of \(k-1\) nodes is independent but the full set of \(k\) nodes is not. This dependence would be invisible to any feature table involving at most \(k-1\) nodes. Computing the worst-case grandchild law compatible with our features (the point in the simplex intersected with the compatible set furthest from the grandchild marginal of \(P_0\)), we end up with the Bernstein law \(B\) with
\[
    p_B(x_3,x_4,x_5,x_6)
    =2^{-4}\{1+x_3x_4x_5x_6\}.
\]
The full-tree law \(Q_0\) below has this grandchild marginal, with additional dependence involving the upper nodes.

To get some sense for the relationship between the process that actually generates the trees and the compatible set we can consider a concrete family of distributions that index the strength of local dependence. We let \(P_\rho\) be the distribution we get by choosing a fair root sign and keeping the parent sign as we go down the tree with probability \((1+\rho)/2\), independently across children:
\begin{equation}
\label{eq:depth-two-ising-law}
    p_{P_\rho}(x)
    =2^{-7}\prod_{a\to b}(1+\rho x_ax_b),
    \qquad 0\leq\rho\leq1,
\end{equation}
where the product contains the six edges. The endpoint \(\rho=0\) is the independent law we have been considering, while \(\rho=1\) is deterministic conditional on the root, with each child directly inheriting its parent's sign. 

In \textbf{Supplemental Fig.~\ref{fig:supp-ising-phase}e}, for each \(P_\rho\) we compute and visualize the compatible set in this style. This shows how the compatible set changes with \(\rho\), and in particular we can see that as \(\rho\) increases, the compatible set eventually contracts to give identification at the deterministic limit. For each compatible set, we also define the worst compatible law,
\[
    Q_\rho\in\arg\max_{R\in\mathcal C_F(P_\rho)}D_E(P_\rho,R),
\]
and approximate this maximization numerically by searching over the polytope. Then \(D_F(P_\rho,Q_\rho)=0\) and \(D_E(P_\rho,Q_\rho)\) is the largest population discrepancy found by the search. We evaluate the finite-sample MSE at this pair; it need not maximize the shifted squared bias in equation~\eqref{eq:feature-map-mse}.

\begin{figure}[p]
    \centering
    \includegraphics[page=2,width=0.94\linewidth]{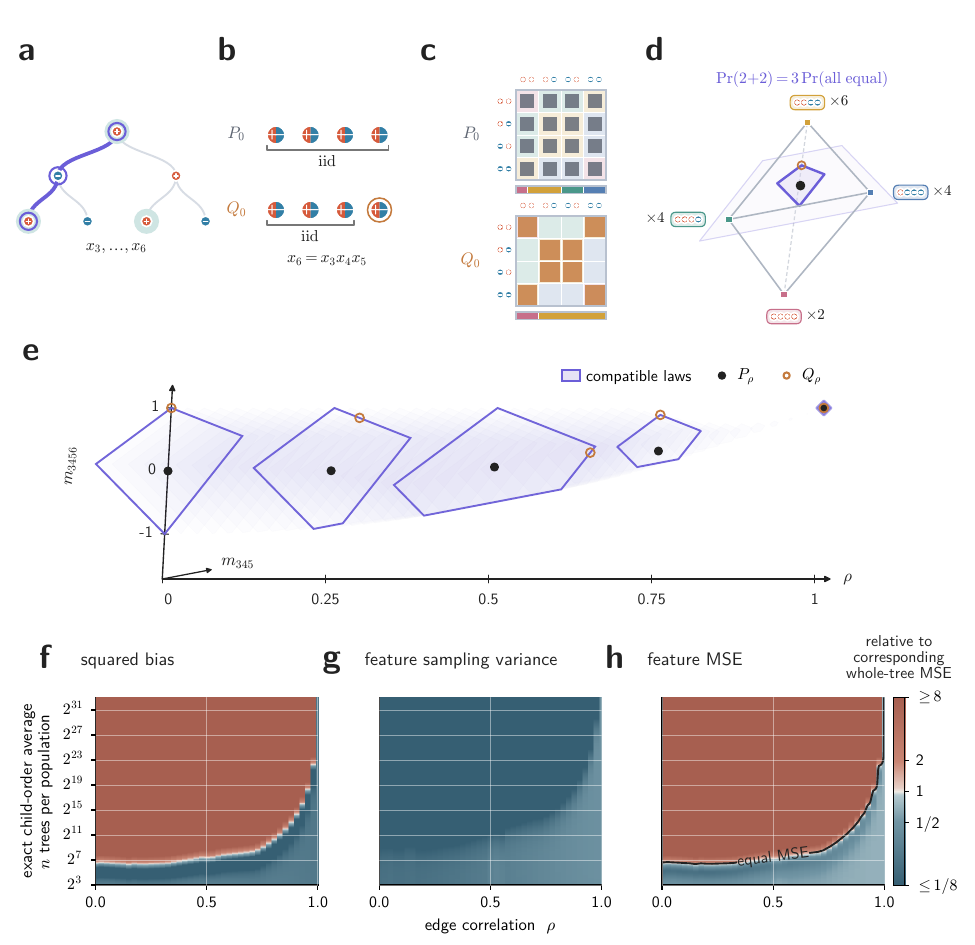}
    \caption{\textbf{Fixed-size feature tables trade sampling variance for partial identification.} \textbf{(a)} Purple marks show the path table, and teal marks show the three common-ancestor tables; marginal tables are redundant with these joint tables. \textbf{(b--d)} At \(\rho=0\), \(P_0\) and \(Q_0\) agree on the selected tables although their four-grandchild laws differ; panel d shows the resulting compatible set. \textbf{(e)} Purple polygons project the compatible sets onto two omitted product-moment combinations, with \(P_\rho\) in black and the numerical worst case \(Q_\rho\) in orange. \textbf{(f--h)} For \((P_\rho,Q_\rho)\), the heatmaps show squared bias for the whole-tree energy, feature sampling variance, and their sum for the V-statistic at a fixed number \(n\) of independent trees per population, averaging child orderings exactly. Colors are relative to whole-tree MSE, and the black contour marks equal MSE.}
    \tocfigure{Fixed-size feature tables trade sampling variance for partial identification}
    \label{fig:supp-ising-phase}
\end{figure}

\textbf{Supplemental Fig.~\ref{fig:supp-ising-phase}f--h} compare the feature and whole-tree statistics at the same \(n\) independent trees per population. For each \(\rho\) and \(n\) we show the three quantities of Proposition~\ref{prop:feature-table-bias-variance} at the pair \((P_\rho,Q_\rho)\): the squared bias \(\{D_E(P_\rho,Q_\rho)-b_F(P_\rho,Q_\rho;n)\}^2\), the exact feature sampling variance \(V_F(P_\rho,Q_\rho;n)\), and their sum, the feature MSE, each relative to the whole-tree MSE \(b_E(P_\rho,Q_\rho;n)^2+V_E(P_\rho,Q_\rho;n)\).

This tells a very clear story: the omitted population discrepancy vanishes as \(\rho\) approaches one, while the feature sampling variance in \textbf{Supplemental Fig.~\ref{fig:supp-ising-phase}g} stays well below the whole-tree MSE, so the feature MSE in \textbf{Supplemental Fig.~\ref{fig:supp-ising-phase}h} beats the whole-tree where trees are scarce or when local dependence is sufficiently strong and loses where trees are plentiful and the compatible set is wide. The black contour marks the crossover, near \(97\) trees at \(\rho=0\).

In practice we have to impose a computational budget since these discrepancies are computed each step of the loss, and the two statistics do not cost the same per tree. We compare the sizes of exhaustively enumerated representations, measured in coordinates per tree. The whole-tree statistic averages over every ordering of each tree, a flattened vector of \(N\) coordinates per element of \(\Pi_{2,d}\), where a complete binary tree of depth \(d\) has \(2^{2^d-1}\) such orderings; the feature statistic averages over the eligible placements of each table, three coordinates per placement. Table~\ref{tab:supp-ising-coordinates} compares the whole-tree and feature statistics in these units, using every ancestor--middle--descendant triple at three distinct generations for paths and every distinct ordered descendant pair for common ancestors. This idealized all-triples family differs from the full-path operator used in training, where in practice we sample orderings, but its is a fair comparison against the whole tree here. On the seven-node tree the two representation sizes are comparable. By depth four the exhaustive whole-tree representation lists five hundred times as many coordinates, and by depth five sixteen million times as many. On trees of the depth we actually fit, exhaustive ordering enumeration is clearly intractable. This is a ``naive'' comparison, we are not giving lower bounds here and cleverer computational schemes could avoid enumerating every ordering. Nonetheless, combined with the statistical arguments presented above, this strongly supports a feature-based computation.

\begin{table}[!t]
    \centering
    \caption{\textbf{Exhaustive ordering enumeration requires exponentially more coordinates than feature-placement enumeration.} For a complete binary tree of depth \(d\) with \(N\) nodes, the whole-tree statistic averages over the \(2^{2^d-1}\) orderings of each tree, \(N\) coordinates each, while the feature representation lists every ancestor--middle--descendant triple at three distinct generations and every common-ancestor triple with distinct ordered descendants, three coordinates per placement. Marginal coordinates are omitted from this representation-size comparison.}
    \toctable{Exhaustive ordering enumeration requires exponentially more coordinates for the whole-tree}
    \label{tab:supp-ising-coordinates}
    \begin{tabular}{rrrrrr}
    \hline
    depth & nodes & orderings & placements & whole-tree coordinates & feature coordinates \\
    \hline
    2 & 7 & 8 & 22 & 56 & 66 \\
    3 & 15 & 128 & 134 & 1{,}920 & 402 \\
    4 & 31 & 32{,}768 & 646 & \(1.0\times10^{6}\) & 1{,}938 \\
    5 & 63 & \(2.1\times10^{9}\) & 2{,}790 & \(1.4\times10^{11}\) & 8{,}370 \\
    6 & 127 & \(9.2\times10^{18}\) & 11{,}430 & \(1.2\times10^{21}\) & 34{,}290 \\
    \hline
    \end{tabular}
\end{table}

Together, these calculations make the tradeoff of featurizing the loss explicit. The whole-tree energy provides a useful target, while the feature loss deliberately gives up sensitivity to some higher-order dependence and can have lower finite-sample MSE.

\FloatBarrier
\subsubsection{Proofs}

\begin{proof}[Proof of Proposition~\ref{prop:binary-energy-expansion}]
For sign-valued random variables \(U\) and \(W\), \(|U-W|=1-UW\). Expanding the three expectations in energy distance therefore gives
\[
    D_E(\operatorname{Law}(U),\operatorname{Law}(W))
    =\{\mathbb E(U)-\mathbb E(W)\}^2.
\]
Taking \(U=\prod_{v\in A}X_v\) for \(X\sim P_{\mathrm{ord}}\) and \(W=\prod_{v\in A}Y_v\) for \(Y\sim Q_{\mathrm{ord}}\) shows that the energy distance between the two product laws is \(\{m_P(A)-m_Q(A)\}^2\). It remains to establish the weighted expansion.

Let \(\nu(x)=p_P(x)-p_Q(x)\). Expanding the between- and within-law expectations gives
\[
    D_E(P,Q)
    =-\sum_{x,y}\nu(x)\nu(y)\lVert x-y\rVert.
\]
Walsh inversion gives
\[
    \nu(x)
    =2^{-N}
    \sum_{A\subseteq V}
    \{m_P(A)-m_Q(A)\}\prod_{v\in A}x_v.
\]
Both probability rows sum to one so the constant moment difference is zero. The distance \(\lVert x-y\rVert\) depends on \(x\) and \(y\) only through the elementwise product. Orthogonality of the Walsh products removes all cross-terms and leaves equation~\eqref{eq:binary-energy-expansion}, with
\[
    \gamma_{N,r}
    =-2^{-N}
    \sum_{u\in\{-1,1\}^V}
    \lVert\boldsymbol1-u\rVert
    \prod_{v\in A}u_v,
    \qquad |A|=r.
\]
By permutation symmetry this coefficient only depends on \(A\) through \(r\). The preceding Fourier calculation holds for arbitrary laws on the binary cube, whether or not they are child-order invariant. For positivity, consider the ordered-assignment law
\[
    p_\varepsilon(x)
    =
    2^{-N}\left\{1+\varepsilon\prod_{v\in A}x_v\right\},
    \qquad 0<|\varepsilon|<1.
\]
This differs from the uniform law and has only one nonzero nonconstant moment, equal to \(\varepsilon\) at \(A\). Energy distance is positive for distinct laws, so \(\gamma_{N,r}\varepsilon^2>0\).
\end{proof}

\begin{proof}[Proof of Proposition~\ref{prop:feature-table-bias-variance}]
We prove the unequal-count extension and then set \(n_P=n_Q=n\) for the stated result.
With the ordering-averaged vector \(\overline{\boldsymbol\phi}\) defined above, Proposition~\ref{prop:binary-energy-expansion} gives
\[
    D_E(P,Q)
    =
    \left\lVert
      \mathbb E_P\{\overline{\boldsymbol\phi}(\mathcal X)\}
      -\mathbb E_Q\{\overline{\boldsymbol\phi}(\mathcal X)\}
    \right\rVert^2.
\]
Let \(\boldsymbol F\) have orthonormal rows spanning the moment combinations fixed by the selected feature tables. The projected distance defined above is then
\[
    D_F(P,Q)
    =
    \left\lVert
      \mathbb E_P\{\boldsymbol F\overline{\boldsymbol\phi}(\mathcal X)\}
      -\mathbb E_Q\{\boldsymbol F\overline{\boldsymbol\phi}(\mathcal X)\}
    \right\rVert^2.
\]

For \(\mathcal X\sim R\), write
\[
    \boldsymbol\mu_R
    =\mathbb E_R\{\boldsymbol F\overline{\boldsymbol\phi}(\mathcal X)\},
    \qquad
    \boldsymbol\Sigma_R
    =\operatorname{Cov}_R\{\boldsymbol F\overline{\boldsymbol\phi}(\mathcal X)\}.
\]
For independent trees \(\mathcal X_1,\ldots,\mathcal X_{n_P}\sim P\) and \(\mathcal X'_1,\ldots,\mathcal X'_{n_Q}\sim Q\), the V-statistic is
\begin{equation}
\label{eq:feature-map-v-statistic}
\begin{aligned}
    \widehat D_F
    ={}&
    \frac1{n_P^2}
    \sum_{i,j}
    \left\langle
      \boldsymbol F\overline{\boldsymbol\phi}(\mathcal X_i),
      \boldsymbol F\overline{\boldsymbol\phi}(\mathcal X_j)
    \right\rangle\\
    &+\frac1{n_Q^2}
    \sum_{i,j}
    \left\langle
      \boldsymbol F\overline{\boldsymbol\phi}(\mathcal X'_i),
      \boldsymbol F\overline{\boldsymbol\phi}(\mathcal X'_j)
    \right\rangle\\
    &-\frac2{n_Pn_Q}
    \sum_{i,j}
    \left\langle
      \boldsymbol F\overline{\boldsymbol\phi}(\mathcal X_i),
      \boldsymbol F\overline{\boldsymbol\phi}(\mathcal X'_j)
    \right\rangle.
\end{aligned}
\end{equation}
To compute its bias and variance, center the projected per-tree vectors as
\[
    \boldsymbol\xi_R
    =\boldsymbol F\overline{\boldsymbol\phi}(\mathcal X)-\boldsymbol\mu_R,
    \qquad \mathcal X\sim R.
\]
For the sampled trees, write \(\boldsymbol\xi_{P,i}\) and \(\boldsymbol\xi_{Q,i}\) for the corresponding centered vectors, and let
\[
    \boldsymbol\varepsilon
    =\frac1{n_P}\sum_i\boldsymbol\xi_{P,i}
      -\frac1{n_Q}\sum_i\boldsymbol\xi_{Q,i},
    \qquad
    \boldsymbol\Sigma
    =\frac{\boldsymbol\Sigma_P}{n_P}
      +\frac{\boldsymbol\Sigma_Q}{n_Q}.
\]
Equation~\eqref{eq:feature-map-v-statistic} gives
\[
    \widehat D_F
    =D_F(P,Q)
      +2(\boldsymbol\mu_P-\boldsymbol\mu_Q)^\top\boldsymbol\varepsilon
      +\lVert\boldsymbol\varepsilon\rVert^2.
\]
Independence gives \(\mathbb E\boldsymbol\varepsilon=0\) and \(\operatorname{Cov}(\boldsymbol\varepsilon)=\boldsymbol\Sigma\), so
\[
    \mathbb E(\widehat D_F)
    =D_F(P,Q)+b_F(P,Q;n_P,n_Q),
    \qquad
    b_F(P,Q;n_P,n_Q)=\operatorname{tr}\boldsymbol\Sigma.
\]
Its exact variance is
\begin{equation}
\label{eq:feature-map-variance}
\begin{aligned}
    V_F(P,Q;n_P,n_Q)
    ={}&4(\boldsymbol\mu_P-\boldsymbol\mu_Q)^\top
          \boldsymbol\Sigma(\boldsymbol\mu_P-\boldsymbol\mu_Q)
       +2\lVert\boldsymbol\Sigma\rVert_{\mathrm F}^2\\
    &+4(\boldsymbol\mu_P-\boldsymbol\mu_Q)^\top
      \left(
        \frac{\mathbb E_P\{\boldsymbol\xi_P\lVert\boldsymbol\xi_P\rVert^2\}}{n_P^2}
        -\frac{\mathbb E_Q\{\boldsymbol\xi_Q\lVert\boldsymbol\xi_Q\rVert^2\}}{n_Q^2}
      \right)\\
    &+\sum_{R=P,Q}\frac1{n_R^3}
      \left[
        \mathbb E_R\lVert\boldsymbol\xi_R\rVert^4
        -(\operatorname{tr}\boldsymbol\Sigma_R)^2
        -2\lVert\boldsymbol\Sigma_R\rVert_{\mathrm F}^2
      \right].
\end{aligned}
\end{equation}
The sum includes one term for each sample. To verify this expression, independence and centering leave only repeated indices in the third- and fourth-moment sums, giving
\[
    \mathbb E\{\boldsymbol\varepsilon\lVert\boldsymbol\varepsilon\rVert^2\}
    =\frac{\mathbb E_P\{\boldsymbol\xi_P\lVert\boldsymbol\xi_P\rVert^2\}}{n_P^2}
     -\frac{\mathbb E_Q\{\boldsymbol\xi_Q\lVert\boldsymbol\xi_Q\rVert^2\}}{n_Q^2}
\]
and
\[
\begin{aligned}
    \operatorname{Var}(\lVert\boldsymbol\varepsilon\rVert^2)
    ={}&2\lVert\boldsymbol\Sigma\rVert_{\mathrm F}^2\\
    &+\sum_{R=P,Q}\frac1{n_R^3}
      \left[
        \mathbb E_R\lVert\boldsymbol\xi_R\rVert^4
        -(\operatorname{tr}\boldsymbol\Sigma_R)^2
        -2\lVert\boldsymbol\Sigma_R\rVert_{\mathrm F}^2
      \right].
\end{aligned}
\]
The linear term contributes the first term of equation~\eqref{eq:feature-map-variance}, the quadratic term contributes the second and fourth-moment terms, and twice their covariance contributes the third-moment term. All these moments are finite in the binary setting.

The estimator's bias for \(D_E(P,Q)\) is \(D_F(P,Q)+b_F(P,Q;n_P,n_Q)-D_E(P,Q)\), so its MSE is \(\{D_E(P,Q)-D_F(P,Q)-b_F(P,Q;n_P,n_Q)\}^2+V_F(P,Q;n_P,n_Q)\). For \(R\in\mathcal C_F(Q)\) the retained combinations of \(R\) and \(Q\) agree, so \(D_F(P,R)=D_F(P,Q)\), and \(Q\) itself lies in \(\mathcal C_F(Q)\); holding \(b_F(P,Q;n_P,n_Q)\) fixed, the squared bias is therefore at most the maximum in equation~\eqref{eq:feature-map-mse}, with equality when \(\mathcal C_F(Q)=\{Q\}\). Taking \(\boldsymbol F\) to be the identity gives \(\widehat D_E\), with sampling bias \(b_E(P,Q;n_P,n_Q)\) and MSE \(b_E(P,Q;n_P,n_Q)^2+V_E(P,Q;n_P,n_Q)\). Setting \(n_P=n_Q=n\) proves the result.
\end{proof}

\section{Supplementary methods: model implementation}
\label{supp:architecture}

The imputer $I_\theta$ and sampler $S_\phi$ use the architecture in Table~\ref{tab:architecture} with separate weights. Imputation takes one pass; sampling uses eight VP-SDE steps by default ($\beta_{\mathrm{min}}=0.1$, $\beta_{\mathrm{max}}=20$; Section~\ref{sec:diffusion}).

\begin{table}[H]
  \captionsetup{margin=0pt,skip=7pt}
  \caption{\textbf{Architecture of the state embedding and lineage models.}}
  \label{tab:architecture}
  \toctable{Architecture of the state embedding and lineage models}
  \begingroup
  \setlength{\tabcolsep}{8pt}
  \setlength{\fboxsep}{2pt}
  \renewcommand{\arraystretch}{1.0}
  \newcommand{\architectureband}[1]{%
    \multicolumn{2}{@{}l@{}}{\colorbox{black!5}{%
      \makebox[\dimexpr\linewidth-2\fboxsep\relax][l]{\strut\bfseries #1}}}\\}
  \begin{tabular}{@{}p{0.45\linewidth}p{\dimexpr0.55\linewidth-16pt\relax}@{}}
    \toprule
    Component & Dimensions \\
    \midrule
    \architectureband{State embedding (scVI)}
    Input genes & $3{,}956$ \\
    Encoder & $2$ hidden layers, width $512$ \\
    Latent width & $50$ \\
    Decoder & $2$ hidden layers, width $512$ \\
    \addlinespace[3pt]
    \architectureband{Inputs and embeddings}
    State projection & $50$ input features, $128$ output features \\
    Developmental-time embedding & $2$ linear layers, width $128$ \\
    Observed-node marker & $2$ learned vectors, width $128$ \\
    Interpolant-time embedding ($S_\phi$ only) & $2$ linear layers, width $128$ \\
    \addlinespace[3pt]
    \architectureband{Tree-aware transformer}
    Transformer layers & $4$ layers, width $128$ \\
    Attention heads & $4$ heads per layer, width $32$ \\
    Feed-forward sublayer & $2$ linear layers, widths $512$ and $128$ \\
    \addlinespace[3pt]
    \architectureband{Residual output}
    Residual prediction & $2$ linear layers, widths $512$ and $50$ \\
    \bottomrule
  \end{tabular}
  \endgroup
\end{table}

\begingroup
\setlength{\parskip}{4pt}
\titlespacing*{\subsection}{0pt}{7pt}{2pt}

\subsection{Transcriptional state embedding}
\label{supp:scvi}
We fit scVI to raw counts with a negative-binomial likelihood~\cite{Lopez2018}. Data source is a categorical decoder covariate, and gene-specific dispersion is shared across sources. We use the posterior mean of the Gaussian latent state as $x_v$; training does not use early stopping.

\subsection{Inputs and embeddings}
We treat each node as one token in the transformer. Observed nodes retain their measured states, while hidden nodes receive independent standard Gaussian noise at hidden nodes in $I_\theta$, and the noisy residual $x_v^s-b_v$ at hidden nodes in $S_\phi$. The projected inputs are combined with embeddings of developmental time and a token representing whether the node is masked or observed; for the sampler $S_\phi$ we add an interpolant-time embedding as well. Developmental times and patristic distances are in embryo days, without further normalization.

\subsection{Transformer layers}
Each layer applies attention and a feed-forward sublayer, with layer normalization before each and a residual connection around each. Patristic-distance bias has a learned scale for each layer and head, initialized to zero. Time embeddings, feed-forward sublayers and output layers use GELU (Gaussian error linear unit) between their linear layers.

\subsection{Residual output}
After final layer normalization, the output layers predict a clean-state residual, which we add to $b_v$. At an observed anchor, $b_v$ equals the observed state. Below an anchor, it copies the nearest observed ancestor's state; above the anchors, it recursively averages the child displacements with equal weight per child. Observed states remain fixed throughout.

\subsection{Numerical sampling}
We use forward Euler updates for CFM and Euler--Maruyama updates for VP-SDE, keeping the observed states fixed. For $N$ sampling steps on the uniform grid $s_j=j/N$, the model is evaluated at $j=0,\ldots,N-1$, and the last clean-state prediction is returned.
\endgroup

\FloatBarrier

\end{document}